\documentclass[a4paper,UKenglish,cleveref, autoref, thm-restate]{lipics-v2021}

\pdfoutput=1
\hideLIPIcs

\title{Problems with Fixpoints of Polynomials of Polynomials}

\author{Cécilia Pradic}{Swansea University, Wales \and \url{https://cpradic.web.deuxfleurs.fr} }{cecilia.pradic@ens-lyon.org}{https://orcid.org/0000-0002-1600-8846}{}%

\author{Ian Price}{Swansea University, Wales \and \url{https://countingishard.org}}{2274761@swansea.ac.uk}{https://orcid.org/0009-0009-4112-3385}{}

\authorrunning{C. Pradic and I. Price} %

\Copyright{Cécilia Pradic and Ian Price} %



\category{} %

\relatedversion{} %

\acknowledgements{
We thank Arno Pauly and Manlio Valenti for discussions on
Weihrauch reducibility, and Zeinab Galal for discussion on fibred fixpoints
and references. We also thank the anonymous reviewers for their comments that
improved this version of the paper.
}%

\nolinenumbers %

\EventEditors{John Q. Open and Joan R. Access}
\EventNoEds{2}
\EventLongTitle{42nd Conference on Very Important Topics (CVIT 2016)}
\EventShortTitle{CVIT 2016}
\EventAcronym{CVIT}
\EventYear{2016}
\EventDate{December 24--27, 2016}
\EventLocation{, United Kingdom}
\EventLogo{}
\SeriesVolume{42}
\ArticleNo{23}

\usepackage{scalerel,stackengine}
\usepackage{url}
\usepackage{cleveref}
\usepackage{amscd}
\usepackage{stmaryrd}
\usepackage{array}
\usepackage{tabularx}
\usepackage{multirow}
\usepackage{mdframed}
\usepackage{extarrows}
\usepackage{mathtools}
\usepackage{quiver}
\usepackage{diagbox}
\usepackage{hhline}
\usepackage{bbold}
\usepackage{yhmath}
\usepackage{makecell}
\usepackage{thm-restate}

\usepackage{tikz}
\usetikzlibrary{automata, positioning}
\tikzset{
    vertical/.style={anchor=south, rotate=90, inner sep=.5mm},
    cartesian/.style={"\square" marking},
    shorten <>/.style={shorten >=#1,shorten <=#1},%
}

\newcommand\mysubparagraph[1]{\subparagraph*{#1}}

\newcommand\mymysubparagraph{\@startsection{subparagraph}{5}{\z@}%
                                       {3.25ex \@plus1ex \@minus .2ex}%
                                       {-1em}%
                                      {\sffamily\normalsize\bfseries}}

\newcommand{\bigslant}[2]{{\raisebox{.2em}{$#1\negthinspace$}\left/\raisebox{-.2em}{$#2$}\right.}}

\newcommand{\Sierp}{\mathbb{S}}

\newcommand{\Cantor}{{2^\mathbb{N}}}

\newcommand{\id}{\textrm{id}}

\newcommand{\dom}{\operatorname{dom}}

\newcommand{\Baire}{{\mathbb{N}^\mathbb{N}}}
\newcommand{\hide}[1]{}

\newcommand{\C}{\mathsf{C}}

\newcommand{\UC}{\mathsf{UC}}
\newcommand{\lpo}{\mathsf{LPO}}

\newcommand{\LPO}{\lpo}
\newcommand{\CChoice}[1]{\mathsf{C}_{#1}}
\newcommand{\CChoiceNat}{\CChoice{\mathbb{N}}}
\newcommand{\CChoiceBaire}{\CChoice{\Baire}}
\newcommand{\WKL}{\mathsf{WKL}}
\newcommand{\KL}{\mathsf{KL}}
\newcommand{\RT}{\mathsf{RT}}
\newcommand{\CA}[1]{#1\text{-}\mathsf{CA}_0}
\newcommand{\PiooCA}{\CA{\mathbf{\Pi}^1_1}}

\newcommand{\WS}[1]{#1\text{-}\mathsf{FindWS}}
\newcommand{\WF}{\mathsf{WF}}

\newcommand{\WSCantor}[1]{#1\text{-}\mathsf{FindWS}_2}

\newcommand{\sTCBaire}{\mathsf{sTC}_{\Baire}}
\newcommand{\TCBaire}{\mathsf{TC}_{\Baire}}

\newcommand{\Nat}{\mathbb{N}}

\newcommand{\partto}{\rightharpoonup}

\newcommand{\slice}[2]{\bigslant{#1}{#2}}

\newcommand{\bnfeq}{\mathrel{::=}}
\newcommand{\bnfalt}{\;\; | \;\;}
\newcommand\cL{\mathcal{L}}
\newcommand\cS{\mathcal{S}}

\newcommand\Container{\mathsf{Cont}}

\newcommand{\Mod}{\mathsf{Mod}}

\newcommand{\Set}{\mathsf{Set}}
\newcommand{\Fam}{\mathsf{Fam}}
\newcommand{\Cat}{\mathsf{Cat}}

\newcommand{\KVPCA}{\mathcal{K}_2}
\newcommand{\KVPCAR}{\mathcal{K}_2^{\mathrm{rec}}}

\newcommand{\shape}{\mathsf{shape}} %
\newcommand{\base}{\mathsf{base}}
\newcommand{\directions}[1]{#1_\mathrm{dir}}

\newcommand\tuple[1]{\left\langle #1 \right\rangle}
\newcommand\cotuple[1]{{\left[ #1 \right]}}

\newcommand\longto{\longrightarrow}
\newcommand\tensor\otimes

\newcommand\Wexp{\Rightarrow} %

\newcommand\exits{\mathsf{exits}}
\newcommand\Strat{\mathsf{Strat}}

\newcommand\incopr{\mathsf{in}}
\newcommand\coproj\incopr
\newcommand\cod{\mathrm{cod}}
\newcommand\op{\mathrm{op}}

\newcommand\monoto\rightarrowtail
\newcommand{\isoto}{\xrightarrow{\sim}}

\makeatletter
\DeclareFontFamily{OMX}{MnSymbolE}{}
\DeclareSymbolFont{MnLargeSymbols}{OMX}{MnSymbolE}{m}{n}
\SetSymbolFont{MnLargeSymbols}{bold}{OMX}{MnSymbolE}{b}{n}
\DeclareFontShape{OMX}{MnSymbolE}{m}{n}{
    <-6>  MnSymbolE5
   <6-7>  MnSymbolE6
   <7-8>  MnSymbolE7
   <8-9>  MnSymbolE8
   <9-10> MnSymbolE9
  <10-12> MnSymbolE10
  <12->   MnSymbolE12
}{}
\DeclareFontShape{OMX}{MnSymbolE}{b}{n}{
    <-6>  MnSymbolE-Bold5
   <6-7>  MnSymbolE-Bold6
   <7-8>  MnSymbolE-Bold7
   <8-9>  MnSymbolE-Bold8
   <9-10> MnSymbolE-Bold9
  <10-12> MnSymbolE-Bold10
  <12->   MnSymbolE-Bold12
}{}

\let\llangle\@undefined
\let\rrangle\@undefined
\DeclareMathDelimiter{\llangle}{\mathopen}%
                     {MnLargeSymbols}{'164}{MnLargeSymbols}{'164}
\DeclareMathDelimiter{\rrangle}{\mathclose}%
                     {MnLargeSymbols}{'171}{MnLargeSymbols}{'171}
\makeatother

\newcommand\powerset{\cP}
\newcommand{\mkAnswerable}{\mathsf{Ans}}
\newcommand{\starW}{\star_{\mathrm{W}}}

\newcommand\PiWM{\Pi\cW\cM}

\newcommand\muExpr{\mu\mathrm{Expr}}
\newcommand\zetaExpr{\zeta\mathrm{Expr}}

\newcommand\ffunctorify[1]{\left\llangle #1 \right\rrangle}
\newcommand\functorify[1]{\left\llbracket #1 \right\rrbracket}
\newcommand\containerify[1]{ \llparenthesis #1 \rrparenthesis }
\newcommand\Tree{\mathsf{Tree}}

\newcommand\bbN{\mathbb{N}}

\newcommand\eqdef{\mathrel{:=}}

\newcommand\inl{\coproj_1}
\newcommand\inr{\coproj_2}

\newcommand\ReprSp{\mathsf{ReprSp}}
\newcommand\rpReprSp{\mathsf{rpReprSp}}
\newcommand\ttIso[2]{#1 \mathrel{\underline{\cong}} #2}

\newcommand\Even{\mathsf{Even}}
\newcommand\Odd{\mathsf{Odd}}

\newcommand\bfG{\mathbf{\Gamma}}
\newcommand\dualptcl[1]{{#1}^c}
\newcommand\bfS{\mathbf{\Sigma}}
\newcommand\bfP{\mathbf{\Pi}}
\newcommand\bfD{\mathbf{\Delta}}

\newcommand\Succ{\mathsf{S}}

\newcommand{\fold}[1]{\mathrm{fold}\,#1}
\newcommand{\unfold}[1]{\mathrm{unfold}\,#1}
\newcommand{\inAlg}{\mathfrak{in}}
\newcommand{\outCoAlg}{\mathfrak{out}}

\newcommand{\Obj}[1]{\mathsf{Obj}\left(#1\right)}

\ifdefined\cA
\else
\newcommand\cA{\mathcal{A}}
\fi
\ifdefined\cB
\else
\newcommand\cB{\mathcal{B}}
\fi
\ifdefined\cC
\else
\newcommand\cC{\mathcal{C}}
\fi
\ifdefined\cD
\else
\newcommand\cD{\mathcal{D}}
\fi
\ifdefined\cE
\else
\newcommand\cE{\mathcal{E}}
\fi
\ifdefined\cF
\else
\newcommand\cF{\mathcal{F}}
\fi
\ifdefined\cG
\else
\newcommand\cG{\mathcal{G}}
\fi
\ifdefined\cH
\else
\newcommand\cH{\mathcal{H}}
\fi
\ifdefined\cI
\else
\newcommand\cI{\mathcal{I}}
\fi
\ifdefined\cJ
\else
\newcommand\cJ{\mathcal{J}}
\fi
\ifdefined\cK
\else
\newcommand\cK{\mathcal{K}}
\fi
\ifdefined\cL
\else
\newcommand\cL{\mathcal{L}}
\fi
\ifdefined\cM
\else
\newcommand\cM{\mathcal{M}}
\fi
\ifdefined\cN
\else
\newcommand\cN{\mathcal{N}}
\fi
\ifdefined\cO
\else
\newcommand\cO{\mathcal{O}}
\fi
\ifdefined\cP
\else
\newcommand\cP{\mathcal{P}}
\fi
\ifdefined\cQ
\else
\newcommand\cQ{\mathcal{Q}}
\fi
\ifdefined\cR
\else
\newcommand\cR{\mathcal{R}}
\fi
\ifdefined\cS
\else
\newcommand\cS{\mathcal{S}}
\fi
\ifdefined\cT
\else
\newcommand\cT{\mathcal{T}}
\fi
\ifdefined\cU
\else
\newcommand\cU{\mathcal{U}}
\fi
\ifdefined\cV
\else
\newcommand\cV{\mathcal{V}}
\fi
\ifdefined\cW
\else
\newcommand\cW{\mathcal{W}}
\fi
\ifdefined\cX
\else
\newcommand\cX{\mathcal{X}}
\fi
\ifdefined\cY
\else
\newcommand\cY{\mathcal{Y}}
\fi
\ifdefined\cZ
\else
\newcommand\cZ{\mathcal{Z}}
\fi
\ifdefined\bA
\else
\newcommand\bA{\mathbb{A}}
\fi
\ifdefined\bB
\else
\newcommand\bB{\mathbb{B}}
\fi
\ifdefined\bC
\else
\newcommand\bC{\mathbb{C}}
\fi
\ifdefined\bD
\else
\newcommand\bD{\mathbb{D}}
\fi
\ifdefined\bE
\else
\newcommand\bE{\mathbb{E}}
\fi
\ifdefined\bF
\else
\newcommand\bF{\mathbb{F}}
\fi
\ifdefined\bG
\else
\newcommand\bG{\mathbb{G}}
\fi
\ifdefined\bH
\else
\newcommand\bH{\mathbb{H}}
\fi
\ifdefined\bI
\else
\newcommand\bI{\mathbb{I}}
\fi
\ifdefined\bJ
\else
\newcommand\bJ{\mathbb{J}}
\fi
\ifdefined\bK
\else
\newcommand\bK{\mathbb{K}}
\fi
\ifdefined\bL
\else
\newcommand\bL{\mathbb{L}}
\fi
\ifdefined\bM
\else
\newcommand\bM{\mathbb{M}}
\fi
\ifdefined\bN
\else
\newcommand\bN{\mathbb{N}}
\fi
\ifdefined\bO
\else
\newcommand\bO{\mathbb{O}}
\fi
\ifdefined\bP
\else
\newcommand\bP{\mathbb{P}}
\fi
\ifdefined\bQ
\else
\newcommand\bQ{\mathbb{Q}}
\fi
\ifdefined\bR
\else
\newcommand\bR{\mathbb{R}}
\fi
\ifdefined\bS
\else
\newcommand\bS{\mathbb{S}}
\fi
\ifdefined\bT
\else
\newcommand\bT{\mathbb{T}}
\fi
\ifdefined\bU
\else
\newcommand\bU{\mathbb{U}}
\fi
\ifdefined\bV
\else
\newcommand\bV{\mathbb{V}}
\fi
\ifdefined\bW
\else
\newcommand\bW{\mathbb{W}}
\fi
\ifdefined\bX
\else
\newcommand\bX{\mathbb{X}}
\fi
\ifdefined\bY
\else
\newcommand\bY{\mathbb{Y}}
\fi
\ifdefined\bZ
\else
\newcommand\bZ{\mathbb{Z}}
\fi

\newcommand\unit{\mathbf{I}}
\newcommand\Ninfty{\mathbb{N}_\infty}

\definecolor{ibmLightOrange}{RGB}{255,176,0}
\definecolor{ibmDarkOrange}{RGB}{254,97,0}
\definecolor{ibmPurple}{RGB}{220,38,127}
\definecolor{ibmLightBlue}{RGB}{100,143,255}
\definecolor{ibmDarkBlue}{RGB}{120,94,240}

\crefname{question}{Question}{Questions}
\crefname{construction}{construction}{constructions}
\crefname{intuitions}{intuition}{intuitions}
\crefname{convention}{convention}{conventions}
\crefname{remark}{remark}{remarks}

\newtheorem{question}[theorem]{Question}

\newtheorem{convention}[theorem]{Convention}

\begin{document}

\maketitle

\begin{abstract}
Motivated by applications in computable analysis,
we study fixpoints of certain
endofunctors over categories of containers.
More specifically,
we focus on fibred endofunctors over the fibrewise opposite of the codomain
fibration that can be themselves be represented by families of polynomial
endofunctors.
In this setting, we show how to compute initial algebras, terminal
coalgebras and another kind of fixpoint $\zeta$.
We then explore a number of
examples of derived operators inspired by Weihrauch complexity and
the usual construction of the free polynomial monad.

We introduce $\zeta$-expressions as the syntax of $\mu$-bicomplete categories,
extended with $\zeta$-binders and parallel products, which thus have a natural denotation
in containers.
By interpreting certain $\zeta$-expressions in a category of type-2
computable maps, we are able to capture a number
of meaningful Weihrauch degrees,
ranging from closed choice on $\{0,1\}$ to determinacy
of infinite parity games,
via an ``answerable part'' operator.
\end{abstract}

\section{Introduction}

\subsection{Internal families, polynomial functors and containers}

The fundamental notion of a family of types $(A_i)_{i : I}$ can be interpreted
in any category $\cC$ with finite limits. Under this interpretation, morphisms $A \to I$
are regarded as $I$-indexed families and taking a pullback along $f : J \to I$
corresponds to computing a reindexing $(A_{f(j)})_{j : J}$.
The most natural category whose objects are families internal to $\cC$ 
is $\cC^\to$, which has commuting squares for morphisms.
\[
\begin{array}{c !\qquad  !\qquad c}
\begin{array}{ll}
\varphi    : & I \longto J \\
\psi : & \prod\limits_{i : I} \left( A_i \to B_{\varphi(i)}\right)
\end{array}
&
\begin{tikzcd}
	A & B \\
	I & J
	\arrow["\psi", from=1-1, to=1-2]
	\arrow["P"', from=1-1, to=2-1]
	\arrow["Q", from=1-2, to=2-2]
	\arrow["\varphi"', from=2-1, to=2-2]
\end{tikzcd}
\end{array}
\]
The codomain functor $\cod : \cC^\to \to \cC$
is then a Grothendieck fibration corresponding to the $\cC$-indexed category
$(\slice{\cC}{I})_{I : \mathsf{Obj}(\cC)}$. This is the starting point of the rich theory
of fibrations and semantics of type theory~\cite{jacobs01book}.
That being said, another important notion of morphisms between families
is the one induced by the fibrewise opposite of the codomain
fibration\footnote{See~\cite[\S 1 \& 5]{streicherfibrations} for a formal definition;
informally, this is one of the few constructions which is easier to define
from the point of view of indexed categories
$\cC \to \Cat$, where it simply corresponds to composing with
the endofunctor $\op : \Cat \to \Cat$ that reverses all arrows.}.
\[
\begin{array}{c !\qquad  !\qquad c}
\begin{array}{ll}
\varphi    : & I \longto J \\
\psi : & \prod\limits_{i : I} \left(B_{\varphi(i)} \to A_i\right)
\end{array}
&
\begin{tikzcd}
	A & \cdot & B \\
	I & I & J
	\arrow["P"', from=1-1, to=2-1]
	\arrow["\psi"', from=1-2, to=1-1]
	\arrow[from=1-2, to=1-3]
	\arrow[from=1-2, to=2-2]
	\arrow["\lrcorner"{anchor=center, pos=0.125}, draw=none, from=1-2, to=2-3]
	\arrow["Q", from=1-3, to=2-3]
	\arrow[shift left, no head, from=2-1, to=2-2]
	\arrow[no head, from=2-2, to=2-1]
	\arrow["\varphi"', from=2-2, to=2-3]
\end{tikzcd}
\end{array}
\]
Internal families and these morphisms make up a category $\Container(\cC)$,
whose objects are sometimes called polynomials or \emph{containers}.
The category $\Container(\cC)$ and its dependent versions have been studied from a number
of viewpoints in mathematics and computer science:

\mysubparagraph{As functors} Assuming $\cC$ has $\Pi$ types, the container $(A_i)_{i : I}$ induces
an endofunctor over $\cC$:
\[ X \qquad \longmapsto \qquad \sum_{i : I} X^{A_i}\]
Endofunctors of this shape are called \emph{polynomial endofunctors}, and we have
that strong natural transformations are in one-to-one correspondence with container
morphisms. Polynomial functors have excellent structural properties
and have thus proved useful in many areas of logic and computer science~\cite{GKpoly},
including semantics~\cite{girard1988normal, glehnmoss18}.
It is also often the case that polynomial
endofunctors admit initial algebras and terminal coalgebras~\cite{moerdijk2000wellfounded,van2007non}
(which are categorifications of the notions of least and greatest fixpoints).
Intuitively, the initial algebra of $(A_i)_{i : I}$ consists well-founded
$I$-labelled trees where $A_i$ indexes the children of $i$-labelled nodes.
Having initial algebras for polynomial functors allows one to interpret inductive types
(also called $\cW$-types) of Martin-L\"of type theory. Dually the terminal coalgebras
($\cM$-types) interpret coinductive types~\cite{abbott2005containers}.

\mysubparagraph{As polymorphic types} The terminology ``containers'' was introduced in the context of functional
  programming~\cite{containers03} to formalize a notion of polymorphic datatypes such as $\alpha \mapsto \mathsf{List}(\alpha)$.
The idea there is that a datatype polymorphic in some parameter $\alpha$ is
concretely determined by a type of \emph{shapes} $I$ and an $I$-indexed families
of holes in which to put labels
taken from $\alpha$. For instance, the container for $\mathsf{List}(\alpha)$ in $\Set$
is the family $(\{0, \ldots, n-1\})_{n : \bbN}$ which says that for
every $n : \bbN$, there is a ``generic list'' of length $n$ with $n$ holes.
Then container morphisms correspond to polymorphic functions that describe
an output shape solely in terms of the input shape and provenance
of the labels of each hole of the output, which must come from holes in the input.
At the level of semantics, instantiating such a polymorphic definition at $\alpha$ then corresponds
to applying the corresponding polynomial functor to $\alpha$.

\mysubparagraph{As problems} Another viewpoint on internal families $(A_i)_{i : I}$ is that they may
be viewed as a formalism for input-output specifications which we may call
\emph{problems}. The idea is that members $i : I$ are inputs
or \emph{questions} and $A_i$ consists of \emph{answers} to $i$. Then a
full solution to such a problem $P : A \to I$ is a section $s : I \to A$. Even
if $P$ is ``surjective'', such full solutions have no reason to exist in general;
in our examples of interest it will be the case that $\cC$ contains only continuous
maps and that $P$ typically has no continuous sections. It is then remarkable that
container morphisms $P \to Q$ can be regarded as reductions that solve $P$
using exactly one oracle call to $Q$, a restriction reminiscent of linear logic.
This notion is extensively
studied in type 2 computability\footnote{But not exclusively there; see~\cite{HirschThesis90}
for a similar set-up in algebraic complexity theory. This is also 
similar to intuitions in terms of gamaes given in~\cite{Hyv14}.}, where the objects of $\cC$ are subspaces of $\Cantor$,
morphisms are represented by Turing machines operating on infinite bit strings
and problems are often given by non-constructive $\forall \exists$ theorems of
second-order arithmetic. In this setting, container morphisms are called
\emph{Weihrauch reductions}~\cite{brattka2022hagen} and capture comparative hardness,
often in a more fine-grained way than entailment does in reverse mathematics~\cite{simpson}.
This viewpoint is also helpful when working with modalities induced by containers~\cite{RSSmod,AB26}
in the context of synthetic computability theory~\cite{KiharaLT24,Swan24,Yoshimura2}.

\subsection{Fixpoints over the category of containers}

Taking the last perspective on containers, their excellent structural properties
translate to a multitude of operators to combine problems that find many concrete
applications~\cite{paulybrattka4}. In particular, we have tensorial products $\tensor$
and $\star$ for the parallel and sequential composition of problems,
and the coproduct $P + Q$ allows one to ask a question to either $P$ or $Q$ and get
a relevant answer. To relax
the linear nature of container morphisms, it is natural to ask for fixpoints
of such operators. For instance, a reduction from $Q$ to $P^\diamond$,
the carrier of the initial algebra
of the endofunctor $X \mapsto \unit + X \star P$ over $\Container(\cC)$, captures the
idea of being able to compute solutions to $Q$ questions if we are allowed
to use $P$ as an oracle an arbitrary finite number of times.
It is not only least fixpoints that are of interest: $\omega$-parallelization
is a standard tool~\cite{survey-brattka-gherardi-pauly} and sequential iterations
of length $\omega$ have recently been introduced~\cite{brattka2025loops,BrattkaS25}.
However, in spite of their natural definitions, they are \emph{not} greatest fixpoints!
For instance, in the case of the terminal coalgebra of $X \mapsto X \tensor P$,
the shape of its carrier consists of streams of $P$-questions, but it defines
no possible answers whatsoever to such streams of questions; $\omega$-parallelization in
Weihrauch reducibility contrastingly allows for relevant streams of $P$-answers.

\begin{figure}
\begin{center}
\begin{tabular}{|c|c|c|}
\hline
Notation & Operation on problems & Fixpoint expression \\
\hline \hline
$P^*$ & finite parallelization & $\mu X. ~ \unit \; + \; P \tensor X$ \\
\hline
\rule{0pt}{1.2em} $\widehat{P}$ & infinite parallelization & $\zeta X. \; P \tensor X$ \\
\hline
$\widetilde{P}$ & \begin{tabular}{c}
                    ask $\omega$ questions\\
                    get one answer \\
                  \end{tabular} & $\nu X. \; P \times X$\\
\hline
$P^\diamond$ & \begin{tabular}{c}
  finite iterations\\
   (free monad over $P$)
\end{tabular}& $\mu X. \; \unit + X \star P$ \\
\hline
$P^\infty$ & $\omega$-loop & $\zeta X. \; X \star P$ \\
\hline
\end{tabular}
\end{center}
\caption{Endofunctors $\Container(\cC) \to \Container(\cC)$ defined via fixpoint equations.
Since the definition of $\star$ requires $\cC$ to be a lccc, the last two definitions
are not quite the same as the operators on Weihrauch degrees (see~\cite[\S 4.4]{PricePradic25}).}
\label{fig:operators}
\end{figure}

We are thus compelled to look for a more structural understanding of such
fixpoints. In this paper, we focus on fixpoints for endofunctors
$F : \Container(\cC) \to \Container(\cC)$ that happen to be components of
fibred endofunctors over $\cod^\op$; that is, those functors such that
there is a (uniquely determined) $F_0 : \cC \to \cC$ such that
$\cod^\op \circ F = F_0 \circ \cod^\op$. Intuitively it means that the question
space 
of $F(P)$ only depends on the question space of $P$ and that $F$ induces functors
$F_X : \left(\slice{\cC}{X}\right)^\op \to 
\left(\slice{\cC}{F_0(X)}\right)^\op$ at the level of answers.
This suggests the following generic recipe to compute a terminal coalgebra of $F$:
first compute a terminal $F_0$-coalgebra $(\nu F_0, c_{\nu F_0})$, and then,
compute a terminal coalgebra for
$(c_{\nu F_0}^*)^\op \circ F_{\nu F_0} : 
\left(\slice{\cC}{\nu F_0}\right)^\op \to
\left(\slice{\cC}{\nu F_0}\right)^\op$,
that is, an \emph{initial} algebra
for 
$c_{\nu F_0}^* \circ F_{\nu F_0}^\op : 
\slice{\cC}{\nu F_0} \to
\slice{\cC}{\nu F_0}$. If we thus assume that $\cC$ has indexed $\cW$-types
and $\cM$-types and that $F_0$, as well as every $F_X^\op$, are polynomials,
then this construction shows us that a terminal $F$-coalgebra exists.
We can also build initial $F$-algebras in a completely dual manner, by first
computing the initial $F_0$-algebra $(\mu F_0, a_{\mu F_0})$,
and then a terminal coalgebra for $(a^{-1}_{\mu F_0})^* \circ F_{\mu F_0}$.
More excitingly, we may also mix the two approaches by taking terminal coalgebras
for both $F_0$ and $c_{\nu F_0}^* \circ F_{\nu F_0}^\op$, which yields a container
$\zeta F$ that comes with an isomorphism $F(\zeta F) \cong \zeta F$ via Lambek's
lemma. This middling fixpoint construction is the one that enables us to capture
$\omega$-iterations and parallelizations without trivializing the space of
answers (see~\Cref{fig:operators}).

Asssuming that $\cC$ is lextensive, has indexed $\cW$-types
and $\cM$-types, a niceness notion for endofunctors $\Container(\cC) \to \Container(\cC)$
(i.e. functors we call \emph{fibred polynomial} in \Cref{def:fibPolFun}), an existence theorem for our fixpoints (\Cref{thm:fp-exist})
and the knowledge that product, coproduct and $\tensor$ are nice (\Cref{lem:stufffibred}),
we have a supply of containers by interpreting the following generalization of the syntax
of $\mu$-bicomplete categories~\cite{santocanale2002mu,Hyv25}, which we call \emph{$\zeta$-expressions}:
\[ E, E'  \enspace \bnfeq\enspace X \bnfalt \mu X. \; E \bnfalt \zeta X. \; E \bnfalt \nu X. \; E \bnfalt E \times E' \bnfalt E \tensor E' \bnfalt E + E'\]
Much like in case of the free $\mu$-bicomplete category, interpreting this syntax
does yield some non-trivial results that can be related to strategies in parity games.
More specifically, we shall see that such expressions $E$ can be algorithmically
turned into regular infinite tree languages $L_E$ describing parity games (as in~\cite{ArnoldNiwinski07}).
Then $E$ gets interpreted as the problem $\ffunctorify{E}$, where a question is a tree $t \in L_E$
which is answered by a winning strategy in the relevant game for the player
with, say, the even objective (\Cref{thm:zetaExprAut}).

\begin{figure}[hbt!]

\begin{tabular}{c !\qquad !\qquad c}
\raisebox{-.5\height}{\includegraphics{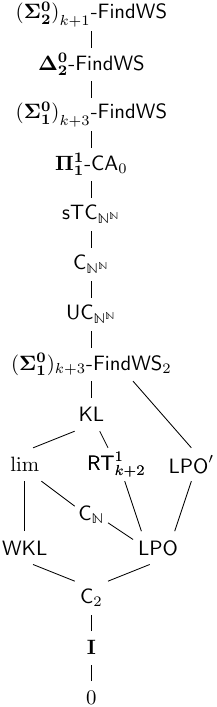}}
\hspace{-3em} &
\[
\begin{array}{c|c}
  \text{Weihrauch degree $D$} & \text{$\zeta$-expression $E$ with $\mkAnswerable\left(\ffunctorify{E}\right) \equiv D$} \\
\hline
\hline
\WS{(\mathbf{\Sigma^0_2)}_{2k}} & \zeta X_{2k}. \nu X_{2k - 1}. \ldots \zeta X_1. \nu X_0. \; \widehat{\widetilde{\sum\limits_{i = 0}^{2k} X_i}}
\\
\WS{\mathbf{\Delta^0_2}} & \mu Z. \; \zeta X. ~~ \widehat{\widetilde{X}} \; + \; \left(\nu Y.\; \widehat{\widetilde{Y}} + Z\right) \\
\sTCBaire & \left(\nu X. \widehat{X + \unit}\right) \times  \left(\zeta X. \widetilde{X + 1} \right)\\
\CChoiceBaire &  \zeta X. \; \widetilde{X + 1}\\
\WS{\mathbf{\Delta^0_1}} & \mu X. ~ \widehat{\widetilde{X}} + \unit + 1 \\
\KL & \zeta X. \; (\nu Y. \; X + Y) \times (\nu Y. \; X + Y) \\
\lim & \widehat{\left(\nu X. \; X + \unit\right) \times \left(\zeta X. \; X + 1\right)} \\
\WKL & \zeta X. \; 1 + X \times X \\
\RT^1_{k} &
\underbrace{\left(\zeta X. \nu Y. \; X + Y\right)
\times \ldots
}_{\text{$k$ times}}  \\
\LPO' &  \left(\nu X. \; \zeta Y. \; X + Y\right) \times \left(\zeta X. \; \nu Y. \; X + Y\right) \\
\CChoiceNat & \widetilde{\zeta X. \; X + 1}\\
\LPO & \left(\nu X. \; X + \unit\right) \times \left(\zeta X. \; X + 1\right) \\
\CChoice{k} & \underbrace{\left( \zeta X. \; X + 1 \right) \times \ldots
}_{\text{$k$ times}}\\
\unit & \zeta X. \; X \\
0 & \mu X. \; X\\
\end{array}\]
\end{tabular}

\caption{Problems from the Weihrauch lattice as the answerable part of
the interpretation of some $\zeta$-expressions in represented spaces.
On the left-hand side is the part of the Weihrauch lattice we capture,
where all inequalities are known to be strict below $\PiooCA$. $\WS{\bfG}$ is the
answerable part of the problem where a question is code for a $\bfG$-set $A$
which is answered by a winning strategy for the player with winning condition
$A$ in a Gale-Stewart game over $\Baire$.}
\label{fig:pblmsAsZetaMKA}
\end{figure}

In spite of the non-trivial combinatorics, the degree of the problem $\ffunctorify{E}$
is always trivial (\Cref{cor:zetaTrivial}). In many situations, we do obtain
problems with interesting question-answer pairs that also contain questions which
cannot be answered, which means that the problem typically has maximal degree (the
terminal object in $\Container(\cC)$ is the trivial family $(0)_{i : 1}$).
This means that we have some way to retrieve interesting information
if we have the ability to filter out these problematic unanswerable
questions. In the context of Weihrauch reducibility, we can do so via an operator
we call $\mkAnswerable$ (\Cref{def:answerablePart}).
Applying $\mkAnswerable$ to denotations of $\zeta$-expressions then allows us
to describe a healthy number of benchmark Weihrauch problems that include the computation of limits,
the infinite pigeonhole principle, closed choice principles on subspaces of $\Baire$
and determinacy for $\bfD^0_1$ and $\bfD^0_2$ objectives (see~\Cref{fig:pblmsAsZetaMKA}).
There are, of course, limits to how many degrees can be expressed this way.
The automata-theoretic characterization of $\ffunctorify{E}$ mean we can only
get $0$ or pointed degrees that are below determinacy for parity games (\Cref{prop:trichotomy});
this however does not rule out most named degrees studied in Weihrauch complexity so far.
We conjecture that some natural degrees, such as $\RT^2_2$ (Ramsey's theorem for pairs),
are not equivalent to any $\mkAnswerable(\ffunctorify{E})$.

\subsection{Related and further work}

\mysubparagraph{Structural aspects of Weihrauch reducibility}
The present paper obviously draws on
efforts to generalize and ``categorify'' Weihrauch reducibility and some of its variants~\cite{Bauer22,TVdP22,maschio2025},
which have only recently identified containers as a possible unifying framework~\cite{AhmanBauer24,PricePradic25}.
As a result, the algebraic theory of operators on Weihrauch degrees developed separately
from the literature on polynomial functors, which led to an apparent
duplication of results pertaining to basic properties of sequential composition for instance~\cite{paulybrattka4,westrick2020}.
Although some care has to be taken as there is a mismatch between
the usual sequential composition of polynomial functors and the composition of Weihrauch
problems. This is due to the fact that computability theorists always work \emph{intensionally}
with Weihrauch problems; they consider containers over a
category $\rpReprSp$ which is only weakly (locally) cartesian closed.
A consequence of this is that the sequential composition in $\Container(\rpReprSp)$ is only a quasifunctor.
While it ends up being
closely related to the classical composition of polynomial endofunctors over a supercategory
$\ReprSp$, more work is needed to definitely capture what is going on in an abstract setting; see~\cite[\S 4]{PricePradic25}
for more details.

\mysubparagraph{Fixpoints of endofunctors over containers} %
Fixpoints and containers enjoy a strong connection, as polynomial functors
do often admit initial algebras and terminal coalgebras~\cite{moerdijk2000wellfounded,van2007non,containers03,abbott2005containers}.
But we are not aware of any systematic study of fixpoints of endofunctors over $\Container(\cC)$
that focus on fibred functors comprised of polynomials like we consider here (rather than
fixpoints of polynomial enfofunctors over $\cC$).
That is not to say that such fixpoints have never been considered in the literature
on polynomial functors and containers: for instance, the free monad over
a polynomial $P$ is computed the initial algebra for $X \mapsto \unit + X \star P$
(and serependipitously coincide with the idea of finite iterations of a problem).

While we do focus on those fibred polynomial endofunctors, we can note that
we only use the ``polynomial'' aspect to be able to use the assumption that
$\cC$ has fixpoints of polynomial functors to begin with. \Cref{thm:fp-exist}
otherwise only requires the fibredness of the endofunctor $F : \Container(\cC) \to \Container(\cC)$
and could be adapted to settings where fixpoints can be obtained by other means
(e.g. via more powerful fixpoint theorems in categories of domains or in settings
with impredicative quantifications). This could plausibly lead to the study
of more interesting fixpoints, such as countable ordinal iterations of a
problem discussed in~\cite{paulycountableordinals}.
Lifting the restriction that $F$ be fibred does break~\Cref{thm:fp-exist}. We
only provide a single example of such a functor which does have a sensible
initial algebra (\Cref{prop:boundedDiamond}).

The definition of fibred polynomial functor that we adopt is useful
insofar as it captures examples relevant to us and makes~\Cref{thm:fp-exist} go through,
but it is unclear to us whether it is a rather ad-hoc notion or if there is
a more conceptual grounding to this notion, possibly related to generalizing
the pseudo-monad $\Container$ to $2$-categories beyond $\Cat$.

Finally, the idea of computing (co)algebras over total categories of fibrations,
by considering the base and then the fibers, has previously appeared in the literature
for posetal fibrations~\cite{BPPR14,HKC18,FGJ23}. Here we consider a specific
class of non-posetal fibrations, and expect that our straightforward approach
generalizes further to other classes of fibrations. However, we are currently not aware
of a useful general theorem along those lines.

\mysubparagraph{Weihrauch reductions and $\zeta$-expressions}
Aside from the problems $\WS{\bfG}$ (which essentially correspond to theorems
stating that infinite games over $\Baire$ with winning conditions in $\bfG$ are determined) with $\bfG$
above $(\bfS^0_1)_2$, all reductions and non-reductions between problems featured in~\Cref{fig:pblmsAsZetaMKA}
appear in the literature~\cite{survey-brattka-gherardi-pauly,cantorDetBCO,MVRamsey}.
To the best of our knowledge, problems around $\PiooCA$ or above have rarely
been investigated in the Weihrauch lattice thus far (exceptions include~\cite{d2021comparison,CiprianiMV25}).
In contrast, determinacy statements have been extensively studied
in reverse mathematics since Friedman's celebrated result that Borel determinacy
is unprovable without the powerset axiom~\cite{friedman1981necessary}; references
most relevant to the subclasses of parity games appearing in~\Cref{fig:pblmsAsZetaMKA}
include~\cite{tanaka90,Tanaka91,mollerfeld02phd,km:2016,PY22}. Work on the topological
complexity of game quantifiers also provides some insights on the logical complexity
of determinacy statements~\cite{burgess1983classical2,ArnoldNiwinski07}.
That being said, those results do not straightforwardly translate to results on
the Weihrauch degrees $\WS{\bfG}$; we therefore limit ourselves to observing
strictness of reductions for the difference hierarchy above $\bfS^0_1$ using
links to finite iterations of $\PiooCA$ already present in~\cite{tanaka90} (\Cref{cor:WSdiffStrict}),
and leave further studies of $\WS{\bfG}$ for higher $\bfG$s to the future.

While $\zeta$-expressions augmented with $\mkAnswerable$ are remarkably expressive,
the non-functoriality of $\mkAnswerable$ makes this formalism difficult
to integrate in a fixpoint logic to compositionally reason about Weihrauch problems (unlike the
equational axiomatizations proposed in~\cite{theoryWeiTimes,pradic25}).
We end by asking whether equivalence of the answerable parts of
two $\zeta$-expressions is decidable (\Cref{q:existsReduction}).

\subsection{Plan of the paper}
\Cref{sec:background} is dedicated to introducing basic notions and notations
pertaining to category theory, containers, type 2 computability, and (co)inductive
types (including a sketch of the syntax of $\mu$-bicomplete categories and its automata-theoretic interpretation).
\Cref{sec:ex} is chiefly concerned with the existence theorem for the three
aforementioned fixpoints for fibred polynomial endofunctors~\Cref{thm:fp-exist}
and its parametric version~\Cref{thm:fp-exist-multi}; it can be read independently
of any material on computability. \Cref{sec:operators} then discusses in more
detail the operators described in~\Cref{fig:operators} as well as all variations
obtained by taking other types of fixpoints, which we hope give a number of
concrete examples to see~\Cref{thm:fp-exist} in action.
Finally, \Cref{sec:ground} focuses more specifically on concrete containers definable
in the context of Weihrauch reducibility via $\mkAnswerable$ and $\zeta$-expressions
described in~\Cref{fig:pblmsAsZetaMKA}.
Due to the number of examples, it relies more heavily on the material on type 2
computability as well as intuitions built by way of previous examples and the proof
of~\Cref{thm:zetaExprAut} explaining how to concretely compute the denotation
of closed $\zeta$-expressions.

\section{Background}
\label{sec:background}

\subsection{Category theory}

We assume familiarity with the basics of category theory~\cite{mac2013categories}.
We write $\tuple{f_1, f_2} : Z \to A_1 \times A_2$ for the pairing of $f_i : Z \to A_i$
($i \in \{1,2\}$), and $\pi_i : A_1 \times A_2 \to A_i$ for projections ($i \in \{1,2\}$).
Dually we write $\cotuple{f_1, f_2} : A_1 + A_2 \to Z$
and $\incopr_i : A_i \to A_1 + A_2$ for the copairing and coprojections.
If $k$ is a natural number, we may regard it as the $k$-fold coproduct of $1$
in any category. 

\begin{definition}[\cite{carboni93extensive}]
A category with finite sums and products is called \emph{extensive}
when the canonical functor $+ : \slice{\cC}{A} \times \slice{\cC}{B} \to
\slice{\cC}{(A + B)}$ is an isomorphism.
We say a category is \emph{lextensive} if it has
all finite limits, all finite coproducts and is extensive.
\end{definition}

Henceforth, all categories in sight shall be lextensive. It means in particular
that they will be distributive and that $\slice{\cC}{k} \cong \cC^k$.

We also assume that the reader is familiar with the notion that morphisms $f : A \to I$
in a category $\cC$ can be regarded as \emph{$I$-indexed families
internal to $\cC$}, the idea being that for $\cC = \Set$, such an $f$ represents
the family $(f^{-1}(i))_{i \in I}$. All lextensive categories support reindexing
of such families through pullback functors $f^* : \slice{\cC}{J} \to \slice{\cC}{I}$, and
have internal sums via its left adjoint $- \circ f = \Sigma_f : \slice{\cC}{I} \to \slice{\cC}{J}$.
Lextensivity further ensures that external finite sums and internal sums socialize harmoniously.
We say that a category is \emph{locally cartesian closed} when it has all finite
limits and the pullback functors admit a right adjoint $\Pi_f$ that interprets 
internal products.
We will often refer to extensive locally cartesian 
closed categories as \emph{elcccs} in the sequel.

\begin{remark}
Extensional Martin-L\"of type theory (EMLTT) with $\prod, \sum, 0, 1, +$ can be
interpreted in (any splitting of) the codomain fibration over an elcccs $\cC$~\cite{CD14}.
For some proofs, we will thus freely use EMLTT as an internal language for such
$\cC$.
\end{remark}

\begin{figure*}[ht]
\begin{center}
\includegraphics[scale=0.65]{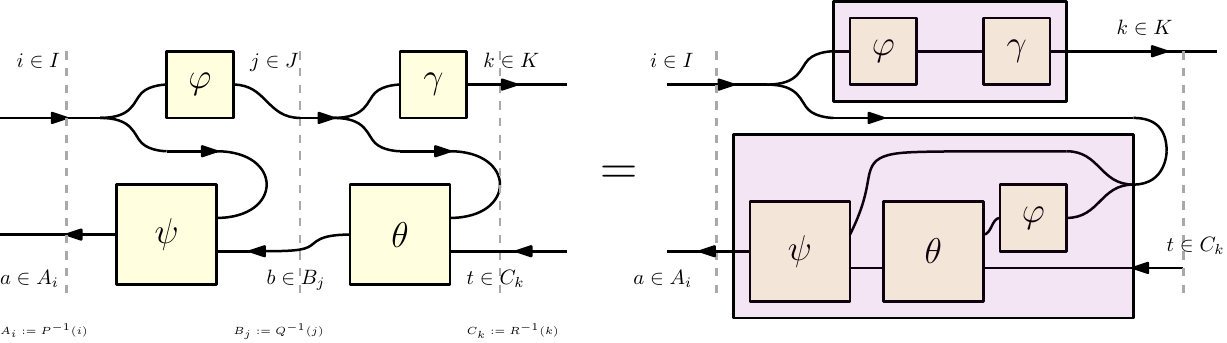}
\end{center}
\caption{Informal diagram representing
the composition of (unary 1-dimensional) container morphisms
$(\gamma, \varphi) \circ (\psi, \theta)$.
Note that this would be a sound string diagram in a category with cartesian
products if we
were working with non-dependent families (seen as a pair of objects); it would
then be formally
interpreted by $(\gamma \circ \psi, \psi \circ \langle\pi_1, \theta \circ (\varphi \times \id) \rangle)$.
}
\label{fig:redcompo}
\end{figure*}

We now discuss inductive and coinductive types in categories,
which are modelled by (co)inductive (co)algebras.
Given a functor $F : \cC \to \cC$, an $F$-algebra is a pair $(A, a)$ where
$A$ is an object of $\cC$, its \emph{carrier}, and $a$ is a morphism $F(A) \to A$. $F$-algebra
morphisms $(A, a) \to (A',a')$ are morphisms $f : A \to A'$ such that
$f \circ a = a' \circ F(f)$.
Initial $F$-algebras $(A, a)$ (and, similarly, terminal $F$-coalgebras) are fixpoints of $F$ in the
sense that $a$ is always an isomorphism; this is called \emph{Lambek's lemma}.
We will sometimes write $(\mu F, \inAlg)$
for a fixed initial $F$-algebra.
Carriers of initial algebras $(\bbN, [0, \Succ] : 1 + \bbN \to \bbN)$ for the functor $X \mapsto 1 + X$
are called \emph{natural number objects}. Dually,
carriers of terminal coalgebras $(\Ninfty, \mathsf{pred} : \Ninfty \to 1 + \Ninfty)$
are called \emph{conatural
number objects}.
The coalgebra structure map $\Ninfty \to 1 + \Ninfty$ essentially tells us whether
its input is zero, and if it is not, returns a predecessor.
We may define a ``point at infinity'' $\infty : 1 \to \Ninfty$ as
the coalgebra morphism $(1, \incopr_2 : 1 + 1)$ to the conatural number object.
In $\Set$, it can be shown that $\Ninfty$ consists of a copy of the natural number
plus this point at infinity; in topological spaces, $\Ninfty$ can be regarded
as the subspace of $\Cantor$ consisting of monotone maps.

\subsection{Polynomial functors and containers}

Let $\cC$ be a category with chosen pullbacks and $T$ and $U$ two objects of $\cC$.
A container from $T$ to $U$ is a triple of morphisms $(P, t, u)$ fitting in
a diagram of the following shape.
\[\begin{tikzcd}
	T & A && I & U
	\arrow["t"', from=1-2, to=1-1]
	\arrow["P", from=1-2, to=1-4]
	\arrow["u", from=1-4, to=1-5]
\end{tikzcd}\]
We call $I$ the \emph{shape} of the container, $A$ its \emph{directions},
$T$ is its \emph{arity} and $U$ its \emph{dimension}.
Containers between $T$ and $U$ are objects of a category $\Container(\cC)(T,U)$
such that a morphism $(P, t, u) \to (Q, t', u')$ is determined by a pair $(\varphi, \psi)$
making the following commute:
\[\begin{tikzcd}
	& A && I \\
	T & \cdot && I & U \\
	& B && J
	\arrow["P", from=1-2, to=1-4]
	\arrow["t"', curve={height=6pt}, from=1-2, to=2-1]
	\arrow[shift right, no head, from=1-4, to=2-4]
	\arrow[shift left, no head, from=1-4, to=2-4]
	\arrow["u", curve={height=-6pt}, from=1-4, to=2-5]
	\arrow["\psi"', from=2-2, to=1-2]
	\arrow[from=2-2, to=2-4]
	\arrow[from=2-2, to=3-2]
	\arrow["\lrcorner"{anchor=center, pos=0.125}, draw=none, from=2-2, to=3-4]
	\arrow["\varphi", from=2-4, to=3-4]
	\arrow["t'", curve={height=-6pt}, from=3-2, to=2-1]
	\arrow["Q"', from=3-2, to=3-4]
	\arrow["u'"', curve={height=6pt}, from=3-4, to=2-5]
\end{tikzcd}\]
We call $\varphi$ the \emph{forward} part of the morphism and $\psi$ the \emph{backward} part. The identity map $(P, t, u) \to (P, t, u)$ is given $(\id_I, \id_A)$
and the composition is definable using pullbacks in $\cC$; rather than spelling
out the details, we refer the reader to~\Cref{fig:redcompo} for an intuition of
how composition works and to~\cite{weber2015polynomials} for formal details.

We will often focus our discussions on containers of dimension $1$ or
containers which have both dimension and arity $1$. To make notations lighter,
we will omit the third component of containers of dimension $1$ (e.g.
$(P, t)$ should be regarded as a container of dimension $1$) and simply
identify containers of dimension and arity $1$ as morphisms (e.g. $P$ is a
container). We set
$\Container(\cC) \eqdef \Container(\cC)(1,1)$. In the sequel, containers are considered to be of
dimension and arity $1$ unless otherwise indicated. For such containers $P : A \to I$,
we will also sometimes write $\shape(P)$ for its shape ($I$ in this example) and $P_{\mathrm{dir}}$ for its
directions ($A$ here). $\shape$ extends to a functor $\Container(\cC) \to \cC$
by considering the map $(\varphi, \psi) \mapsto \varphi$ on morphisms. In fact,
$\shape$ is a Grothendieck fibration, and it is straightforward to check that
it is exactly the fibrewise opposite of the codomain fibration $\cod : \cC^\to \to \cC$
(see~\cite[\S 1 \& \S 5]{streicherfibrations} for details). So given a morphism
$P \to Q$ in $\Container(\cC)$ given by $(\varphi, \psi)$, we will call it
\emph{cartesian} or \emph{horizontal} if $\psi$ is an isomorphism and \emph{vertical}
if $\varphi = \id$.

If $\cC$ is locally cartesian closed, containers $(P, t, u)$ with arity
$T$ and dimension $U$ induce functors
$\functorify{P,t,u} \eqdef \Sigma_u \circ \Pi_P \circ t^*: \slice{\cC}{T} \to \slice{\cC}{U}$. 
In $\Set$, up to the usual equivalences $\slice{\Set}{T} \cong \Fam(T)$ and
$\slice{\Set}{U} \cong \Fam(U)$ between slice categories
and indexed families and regarding $(P,t,u)$ as the triply-indexed family
$(A_{i, \alpha, \beta})_{\alpha \in U, i \in I_\alpha, \beta \in T}$
(by taking $I_\alpha = u^{-1}(\alpha)$ and
$A_{i, \alpha, \beta} = P^{-1}(i) \cap t^{-1}(\beta)$), the container
$(P, t, u)$ induces the functor
\[
\begin{array}{lcl}
\Fam(T) &\longto& \Fam(U) \\
(X_\beta)_{\beta \in T} &\longmapsto& \left(\sum\limits_{i \in I_\alpha} \prod\limits_{\beta \in T}
\left(A_{i, \alpha,\beta} \to X_\beta\right)\right)_{\alpha \in U}
\end{array}
\]
In the one-dimensional case, $(P, t)$ induces a functor $\functorify{P,t} :
\slice{\cC}{T} \to \cC$, and in the one-dimensional unary case, $P$ induces
an endofunctor $\functorify{P} : \cC \to \cC$. The $\cC = \Set$ case then
simplifies to the following expression reminiscent of power series:
\[
\begin{array}{lcl}
\Set &\longto& \Set \\
X &\longmapsto& \sum\limits_{i \in I}
X^{A_{i}}
\end{array}
\]

\begin{proposition}[{Corollary of~\cite[Theorem 2.12]{GKpoly}}]
\label{prop:contPolyEquiv}
$\functorify{-}$ extends to a functor from $\Container(\cC)(I,J)$ to the
the category of strong functors $\slice{\cC}{I} \to \slice{\cC}{J}$ and
strong natural transformations. Furthermore, $\functorify{-}$ is full and
faithful.
\end{proposition}

\begin{figure}
\[
\begin{array}{|c|c|c|c|}
\hline
\text{Cartesian} & \text{Coproduct} & \text {Tensor} & \text {Sequential} \\
\text{product} &  & \text {product} & \text {product} \\
P \times Q & P + Q & P \tensor Q & Q \star P \\
\hline
\begin{tikzcd}
	{A \times J + I \times B} \\
	{I \times J}
	\arrow["{\left[P \times \id_J,\id_I\times Q\right]}", from=1-1, to=2-1]
\end{tikzcd}
&
\begin{tikzcd}
	{A + B} \\
	{I + J}
	\arrow["{P + Q}", from=1-1, to=2-1]
\end{tikzcd}
&
\begin{tikzcd}
	{A \times B} \\
	{I \times J}
	\arrow["{P \times Q}", from=1-1, to=2-1]
\end{tikzcd}
&
\begin{tikzcd}
	{\sum\limits_{a : A} B^{A_{P(a)}}} \\
    {\sum\limits_{i : I} J^{P^{-1}(i)}}
	\arrow[from=1-1, to=2-1]
\end{tikzcd}\\
\hline
\end{array}
\]
\caption{Definitions of the tensorial products on $\Container(\cC)$; we allow
ourself to use the internal language of $\cC$ as a locally cartesian closed
category for $\star$ and let the reader guess what is the unique possibility
for the morphism.}
\label{fig:monoidalOperators}
\end{figure}

When $\cC$ is lextensive, $\Container(\cC)$ is also lextensive and features
two additional monoidal products $\tensor$ and $\star$
outlined in~\Cref{fig:monoidalOperators}. The parallel product is essentially
the componentwise product of families, while the sequential product is more involved:
\begin{align*}
& (A_i)_{i : I}, (B_j)_{j : J} & \longmapsto &&& (A_i \times B_j)_{(i, j) : I \times J} \\
& (A_i)_{i : I}, (B_j)_{j : J} & \longmapsto &&& \left(\sum_{b : B_j} A_{f(b)}\right)_{(j, f) : \sum\limits_{j : J} I^{B_j}}
\end{align*}
The sequential product we will also sometimes call ``composition product''
since we have $\functorify{P \star Q} = \functorify{Q} \circ \functorify{P}$.
It is the only product in~\Cref{fig:monoidalOperators} which does not come
with a symmetry. All monoidal products come with a unit, with $\star$
and $\tensor$ having the same unit $\unit$ with $\shape(\unit) = \unit_{\mathrm{dir}} = 1$.
When $\cC$ is additionally locally cartesian closed, it is symmetric monoidal
closed for both $\times$ and $\tensor$.

\subsection{Weihrauch reducibility and containers}

One perspective on (unary and $1$-dimensional) containers is that they can
be viewed as \emph{problems} in the following way: the shape $I$ of a container
$P : A \to I$ should be regarded as a space of \emph{questions}, the directions
$A$ should be regarded as \emph{answers} and $P$ specifies what are the answers
$P^{-1}(i)$ to some question $i : I$.
Computability tells us that some problems cannot be solved computably nor
continuously, in which case we are interested in quantifying how large the
obstruction is via notions of \emph{reductions}. Container morphisms constitute
one such notion of reduction: a morphism $P \to Q$ gives
us a forward map $\varphi$ turning a $P$-question into a $Q$-question and the
backward map $\psi$ turns a $P$-question $i$ and a $Q$-answer to $\varphi(i)$ into
a $P$-answer to $i$.
Intuitively, this means solving $P$ with one, no more, no less, oracle call
to $Q$. In light of this, the monoidal operators from~\Cref{fig:monoidalOperators},
and their units, admit natural interpretations as problem transformers.

If we are strictly interested in the power of problems, we may want to talk
about the reduction preorder giving $P \le Q$ if and only if there exists
a morphism $P \to Q$ in $\Container(\cC)$. We will freely write $\le$ for this
preorder in the sequel, $\equiv$ for the associated equivalence relation over
containers, and call its equivalence classes \emph{degrees}.

To develop examples for this paper, we will be looking at the particular instantiation of this
framework in the context of Weihrauch reducibility, previously discussed in~\cite{PricePradic25}.
Up to some caveats, this amounts to picking a category for type 2 computability that matches
the practice in computable analysis; this will be the category of
represented spaces $\ReprSp$.

A represented space $(X, \delta_X)$ is a set together with a partial surjection
$\delta_X : \Cantor \partto X$ called its \emph{representation}.
A map $f : X \to Y$ between two represented spaces $(X, \delta_X)$ and $(Y, \delta_Y)$
is called \emph{computable} if it has a (type 2) computable realizer, that is,
if there is a multitape Turing $M$ such that, if given
$p \in \dom(\delta_X)$ as an oracle, it will write all prefixes of an element of
$\delta_Y^{-1}(f(\delta_X(p)))$ on a write-only output tape.

Represented spaces and computable maps
form an elccc,
which readers familiar with realizability will see to be equivalent the
full subcategory $\Mod(\KVPCAR, \KVPCA)$ of modest sets of the Kleene-Vesley topos
studied extensively in~\cite{bauerPhD}.

Every Weihrauch problem $P$ can be regarded as an object of $\Container(\ReprSp)$,
and every Weihrauch reduction from $P$ to $Q$ correspond to a morphism of
$\Container(\ReprSp)$. Let us give an example.

\begin{definition}
  \label{def:lpo}
Informally $\LPO$ is the problem ``given an infinite binary sequence $p \in \Cantor$,
return a boolean saying whether $p = 0^\omega$ or not''.
Spelling out the formal details, $\LPO$ is the following container:
\begin{itemize}
\item $\shape(\LPO)$ is Cantor space $\Cantor$
\item $\LPO_{\mathrm{dir}}$ is the coproduct\footnote{It is \emph{not}
isomorphic to $\Cantor$, as the first component is topologically isolated.}
$\{0^\omega\} ~~+~~ \Cantor \setminus \{0^\omega\}$.
\item the container is the copairing of the obvious inclusions.
\end{itemize}
\end{definition}

In the sequel, we will often simply state what problems are, and leave to the
reader figuring out the details of the relevant container.

For any set $A$, we write $A^*$ for the set of finite words over $A$.
Let us write $\sqsubseteq$ for the prefix relation
over finite and infinite words and, given any word $w$ of length $\ge n$, write
$w[n]$ for its restriction to its first $n$ letters.
Recall that (possibly countably-branching infinite) \emph{trees} $t$
can be regarded as maps $t : \bbN^* \to 2$ (which we sometimes confuse with subsets of $\bbN^*$)
such that $t(\varepsilon) = 1$, which are prefix-closed ($t(un) = 1 \Rightarrow t(u) = 1)$.

\begin{definition}
  \label{def:kl}
$\KL$ is the problem ``given a finitely branching tree $t : \bbN^* \to 2$,
find an infinite path through it''.
\end{definition}

\begin{example}
There is a container morphism $\LPO \to \KL$ defined as follows.
The forward map $\varphi$ takes as input a sequence $p \in \Cantor$ and outputs a
countably branching tree $\varphi(p) : \bbN^* \to 2$. We may define it by cases:
\begin{itemize}
\item $\varphi(p)(0^n) = 1$ whenever $0^n$ is a prefix of $p$,
\item $\varphi(p)((k+1)0^n) = 1$ whenever $0^k 1$ is a prefix of $p$ (for any $n \in \bbN$),
\item otherwise, $\varphi(p)(q) = 0$ for any other input $q \in \bbN^*$.
\end{itemize}
$\varphi$ can be implemented by a type 2 Turing machine that works as follows.
First the machine reads its second input $q \in \bbN^*$, and determines whether
it is of shape $0^n$, $(k+1)0^n$ or something else. In the last case it
returns $0$ immediately. In the first and second cases, it can return the correct
result after reading off $n$ and $k+1$ bits of $p$ respectively.

Now the tree given by $\varphi(p)$ always has a single infinite branch: if
$p = 0^\omega$, then the branch is $0^\omega$. If $p \neq 0^\omega$, then
$\varphi(p)$ will have an infinite branch $(k+1)0^\omega$ where $k$ is the position of
the first non-zero bit of $p$; it may additionally have another finite branch
when $k \neq 0$, but no further infinite branches.

So we have that $\varphi(p) \in \shape(\KL)$, and $\varphi(p)$ has a unique
infinite branch whose first element is $0$ if and only if $p = 0^\omega$. Hence,
one can define a backward map
\[
\begin{array}{llcl}
\psi : &\sum\limits_{p \in \Cantor} \{ r \in \Cantor \mid \text{$r$ is a branch of $p$}\}  & \longto& \{0^\omega\} + \Cantor \setminus \{0^\omega\}
\\
& (p , r) &\longmapsto & \left\{ \begin{array}{ll}
\inl(0^\omega) & \text{if $r_0 = 0$} \\
\inr(p) & \text{otherwise}\\
\end{array} \right.
\end{array}\]
so that $(\varphi, \psi)$ is\footnote{Technically, it rather is one representative of
such a morphism. Let us also note that we would have also remarked already that
$\LPO$, regarded as a morphism, is a pullback of $\varphi$ along $\KL$ (also regarded
as a morphism), so $(\varphi, \id)$ is also a representative of the same morphism.} a container morphism $\LPO \to \KL$.

Conversely, there is no container morphism $\KL \to \LPO$ for continuity reasons.
Let us sketch the argument: if there were such a morphism $(\varphi', \psi')$,
then $\varphi'$ cannot be constant, as it would otherwise mean that $\KL$ has a
continuous section. Hence, there is some tree $t$ such that $\varphi'(t) \neq 0^\omega$, which has no infinite path starting with $k$.
By continuity of $\psi'$, there is a finite approximation $t_f$ of $t$ such that
$\psi'(u, \inr(\varphi'(u)))_0 = k$ for some $k$ and all extensions $u$ of $t_f$.
This is a contradiction, as we can pick $u$ with $u(k'0^n) = 1$ for $k' \in \mathbb{N} \setminus \{k\} \cup \dom(t_f)$,
and $u(q) = 0$ for all other $q \not\in \dom(t_f)$, which has no infinite path starting with $k$.
\end{example}

While every Weihrauch problem corresponds to a container over $\ReprSp$, the
converse is not true for two reasons.
Firstly, the official definition of Weihrauch reducibility allows for reductions
that may not be extensional, so that more problems are identified.
To fix this, we may consider containers over the full subcategory $\rpReprSp$
of $\ReprSp$ whose objects are the  subspaces of $\Cantor$ (that is, objects
of $\ReprSp$ that are isomorphic to spaces represented by a restriction of the
identity map $\id_{\Cantor} : \Cantor \to \Cantor$).
$\rpReprSp$ shares finite
limits and colimits with $\ReprSp$, however it is only weakly locally cartesian
closed. Formally this means that reindexings $f^*$ only have weak right adjoints~\cite[Remark 3.2]{CR00}; so dependent function spaces exist, but functions
may have multiple representatives (type-theoretically, this means that $\eta$ laws
need not hold).

\begin{remark}
\label{rem:countableExponentials}
Certain exponentials do exist in $\rpReprSp$: for instance, for
every regular projective represented space $A$ and subspace $N$ of $\bbN$, $A^N$
can be built. We also have an analogue local statement: call a map
$f : X \to Y$ in $\rpReprSp$ \emph{locally countable} if there exists some
$s : X \to \bbN$ such that $\langle f, s\rangle : X \to Y \times \bbN$ is
a subspace embedding. For such an $f$, $f^*$ does have a right adjoint $\Pi_f$
in $\rpReprSp$.
\end{remark}

We will still work with elcccs and $\ReprSp$ in the rest of the paper. This
will mean that fixpoints involving $\star$ discussed in~\Cref{sec:operators} will \emph{not}
match those used in Weihrauch reducibility, but can be thought of as extensional
analogues which share many properties.
That said, \Cref{rem:countableExponentials} will mean that
most results of~\Cref{sec:ground} do discuss genuine Weihrauch degrees (see~\Cref{rem:genuineGroundProblems}).
Investigating the case of extensive
weak locally cartesian closed categories and $\Container(\rpReprSp)$ is left for
future work; we hope most results carry through modulo the development of a
theory of polynomial quasifunctors over weakly locally cartesian closed categories.

The second issue is that the notion of Weihrauch problem is usually restricted
so that every question always has an answer. Containers in general have no such
constraints, which makes their structural theory a bit nicer: for instance, the
terminal object in $\Container(\ReprSp)$ is a problem with a question with no
answers, and constructing exponentials in $\Container(\ReprSp)$ requires such
questions with no answers; in contrast, the Weihrauch lattice has no $\top$.

\begin{definition}[{\cite[Definition 7]{PricePradic25}}]
  \label{def:answerable}
  A container $P : A \to I$ in $\Container(\cC)$ will be called \emph{answerable}
if it is a pullback-stable epimorphism, that is, an epimorphism such that $f^*(P)$ is always an epimorphism for any $f : J \to I$.
\end{definition}

In represented spaces,
this condition is equivalent to $P$ being a set-theoretic surjection.
So overall, an official Weihrauch problem will be an answerable container
over $\rpReprSp$. In the sequel, we shall often build objects of $\Container(\rpReprSp)$
and then build the subobject that restricts to answerable questions.

\begin{definition}
  \label{def:answerablePart}
The \emph{answerable part} $\mkAnswerable(P)$ of a container $P : A \to I$ in $\Container(\ReprSp)$
is the corestriction of $P$ to its image.

That is, if $I$ is represented by the partial map
$\delta_I : \Cantor \partto I$,
the codomain of
$\mkAnswerable(P)$ is the represented space
$(\{P(a) \mid a \in A\}, \delta')$ with $\delta'$ defined as the largest possible
restriction of $\delta_I$. $\mkAnswerable(P)$ is then tracked by the same realizer
as $P$ (categorically, this corresponds to taking the left part of a pullback-stable epi-regular mono factorization of $P$).
\end{definition}

It is obvious that the answerable part of a container is answerable. Note that,
contrary to the other operations on containers we discussed, it is not functorial, nor
monotone when we look at the induced degrees: for instance, the answerable part
of $1$ is $0$, while the answerable part of $\unit + 1$ is $\unit$, although we
have $1 \equiv \unit + 1$.

\subsection{Fixpoints of polynomial functors}

A convenient feature of polynomial functors is that they often admit fixpoints
in categories of interest, which then interpret a variety of inductive
and coinductive types.

\begin{definition}
A category $\cC$ is said to have dependent $\cW$-types (respectively, $\cM$-types) if for every object $I$
of $\cC$ and container $(P, t, u)$ of $\Container(\cC)(I, I)$,
$\functorify{P,t,u} : \slice{\cC}{I} \to \slice{\cC}{I}$ has an initial algebra (respectively, a terminal coalgebra).
A $\PiWM$-category is an extensive locally cartesian closed category which
also has dependent $\cW$ and $\cM$-types.
\end{definition}

We will primarly be concerned with $\Set$ and $\ReprSp$, which are $\PiWM$-categories~\cite{bauerPhD} and in which one should interpret all results in the
reminder of this subsection (and~\Cref{sec:ground}).
$\PiWM$ categories have natural and conatural objects, as the polynomial $1 + X$ is induced the
container $\inr : 1 \to 2$.
More generally, they interpret $\mu$-expressions defined as follows:
\[ E, E' \enspace \bnfeq\enspace X \bnfalt \mu X. \; E \bnfalt \nu X. \; E \bnfalt E \times E' \bnfalt E + E'\]
Any such $\mu$-expression $E$ with $k$ free variables
thus induce a $k$-ary container that
we write $\containerify{E}$; we shall also write $\functorify{E}$
for $\functorify{\containerify{E}}$.
Intuitively, closed $\mu$-expressions denote types of
(possibly infinite) trees, which can systematically be characterized by
tree automata~\cite{santocanale2002mu,Hyv25}\footnote{The characterization
given in~\cite{santocanale2002mu,Hyv25} is as winning strategies in parity games
over finite arenas. Those winning strategies can be regarded as trees, and
the basic theory of automata over infinite trees tell us they form
regular tree languages for any given finite arena. We favor using a direct
description in terms of regular tree language in the text to minimize cognitive
load later. as we will need to regard those strategies/trees as arenas for (other) parity games in \Cref{sec:ground}.}.

\begin{example}
  \label{ex:muExprNat}
We have $\containerify{\mu X. 1 + X} \cong \left\{ \chi_{\{u \in 2^* \mid u \sqsubseteq w\}} \mid w  \in 0(10)^*1 \right\}$
where $\chi_A$ is the indicator function $2^* \to 2$ of $A \subseteq 2^*$.
For the greatest fixpoint, we have that $M_{\nu X. 1 + X}$ is the closure of
$M_{\mu X. 1 + X}$ in $2^* \to 2$, i.e. $M_{\mu X. 1 + X} \cup \left\{ \chi_{(01)^\omega}\right\}$.
\end{example}

\begin{example}
  \label{ex:muExprWF}
Let us consider $E = \mu X. \; \nu Y. \; (1 + X) \times Y$.
We have that $\containerify{E}$ is isomorphic to the set of well-founded $\bbN$-ary trees
$t \subseteq \bbN^*$.
While this can be derived from first principles, a thoughtless automata-theoretic
translation also tells us that
\[\containerify{E} \cong M_E \eqdef \{ \ell_t \mid t \subseteq \bbN^* \text{ and $t$ is well-founded}\}\]
with $\ell_t : 2^* \to 2$ defined as the indicator function of
\[ L_t = \bigcup_{w \in t}
\{u \sqsubseteq h(w)(10)^n00 \mid wn \notin t\}
\cup
\{u \sqsubseteq h(w)(10)^n01 \mid wn \in t\}
\]
with $h : t \to 2^*$ defined by 
$
h(\varepsilon) = \varepsilon$ and $h(wn) = 
h(w) 00(10)^n01$. $M_E$ is recognized by the non-deterministic
coB\"uchi infinite tree automaton~\cite[\S 14.3]{Toolbox} below, where we draw an arrow $q \xrightarrow{a ; i} r$
for a transition from state $q$ to $r$ reading a label $a \in 2$ and going down
in direction $i \in 2$.
The priority of all states are $0$ except for $X$, which has priority $1$,hence
an accepting run only goes through the state $X$ finitely many times along any branch, ensuring the encoded
$\bbN$-ary tree is well-founded. All transitions are deterministic, except
for those starting from state $\times$ along the direction $0$: there, a guess
of whether a node has a child or not is taking place.

\begin{center}
\begin{tikzpicture}[shorten >=1pt,node distance=2cm and 3cm,on grid,auto]
\tikzstyle{every node}=[font=\small]
\node[state,initial, initial where=left, initial text={}, accepting]  (X)                      {$X$};
  \node[state]          (Y) [right=of X] {$Y$};
  \node[state]          (P) [right=of Y] {$\times$};
  \node[state]          (L) [above=of P] {$1 \textcolor{gray}{+ X}$};
  \node[state]          (O) [left=of L] {$1$};
  \node[state]          (R) [below=of P] {$\textcolor{gray}{1 + } X$};
  \node[state]          (b) [above=of X] {$\bot$};
  \node[state]          (bb) [right=of P] {$\bot$};

  \path[->] (X)   edge              node [below] {$1 ; 0$} (Y)
                  edge              node [left] {$1; 1$} (b)
            (b)   edge [loop above] node         {$0; 0,1$} ()
            (bb)  edge [loop right] node         {$0; 0,1$} ()
            (Y)   edge [bend right]  node [below] {$1 ; 0$} (P)
                  edge              node [right] {$1; 1$} (b)
            (P)   edge              node [right] {$1 ; 0$} (L)
                  edge              node [right] {$1 ; 0$} (R)
                  edge [bend right] node [above] {$1 ; 1$} (Y)
            (L)   edge              node [above] {$1 ; 0$} (O)
                  edge  [bend left]   node [above] {$~~~~1 ; 1$} (bb)
            (O)   edge              node [above] {$1 ; 0,1$} (b)
            (R)   edge  [bend left] node [below] {$1 ; 1~~$} (X)
                  edge  [bend right] node [below] {$~~~~1 ; 0$} (bb);
\end{tikzpicture}
\end{center}
\end{example}

\section{Fixpoints of fibred polynomial functors}
\label{sec:ex}

\subsection{Fibred polynomial functors}

We now turn to introducing fixpoints of endofunctors over categories of
containers $\Container(\cC)$.
Here we focus on what we can achieve when
we assume that $\cC$ is a $\PiWM$-category on exploiting the
fibration $\cod^\op \simeq \shape : \Container(\cC) \to \cC$ (and more generally
its $k$-fold products $\shape^k : \Container(\cC)^k \to \cC^k$).
This and a number of basic examples lead us to the following convenient definition.

\begin{definition}
  \label{def:fibPolFun}
Say that $F : \Container(\cC)^k \to \Container(\cC)$ is a \emph{fibred polynomial functor}
when
\begin{enumerate}
\item $(F, F_0)$ is a fibred functor $\shape^k \to \shape$ for some uniquely determined
  $F_0 : \cC^k \to \cC$ (see~\cite[Definition 2.2]{streicherfibrations})
\item $F_0$ is a polynomial functor
\item for every object $I = (I_1, \ldots, I_k)$ in $\cC^k$, the functors on the fibers
\[F_I : \bigslant{\cC}{I_1 + \ldots + I_k} \longto \bigslant{\cC}{F_0(I)}\]
induced by the isomorphism $\bigslant{\cC}{I_1 + \ldots + I_k} \cong \bigslant{\cC^k}{I}$ and $F$ are all polynomial functors.
\end{enumerate}
Say that $F :  \Container(\cC)^k \to \Container(\cC)^m$ is a \emph{fibred polynomial functor}
if every component $\pi_i \circ F$ is a fibred polynomial functor.
\end{definition}

Composition of fibred polynomial functors is defined in the obvious way and does
result in fibred polynomial functors.
Identities are fibred polynomial functors
(as identity functors are polynomial functors~\cite[Example 1.6 (i)]{GKpoly}
and they trivially preserve cartesian morphisms).
Fibred polynomial functors thus form a category, which also admits
obvious finite cartesian products.
Now we turn to showing which of the monoidal products from~\Cref{fig:monoidalOperators}
are also fibred polynomial functors.

\begin{restatable}{lemma}{stufffibred}
\label{lem:stufffibred}
Constant functors and the bifunctors $\times, +, \otimes$ over $\Container(\cC)$
are all fibred polynomial.
\end{restatable}

\begin{restatable}{lemma}{trianglefibred}
\label{lem:trianglefibred}
The following functor is fibred polynomial
\[\begin{array}{lcl}
  \Container(\cC) &\longto& \Container(\cC) \\
  X &\longmapsto& X \star P \\
\end{array}
\]
\end{restatable}
\begin{remark}
  \label{rem:triangleNotfibred}
On the other hand, $X \mapsto P \star X$ is \emph{not} fibred.
\end{remark}

With these lemmas, we have gathered all the tools we shall need to produce
examples for of fibred polynomial functors in the rest of the paper. For instance,
$X \mapsto \unit + X \star A$ is seen to be fibred polynomial by combining the
above lemmas.

\begin{remark}
Unlike the usual polynomial endofunctors over a locally cartesian closed category,
fibred polynomial functors do not necessarily have a canonical strength. A
counter-example is the constant functor that maps any polynomial to $\unit$.
\end{remark}

\subsection{They have fixpoints}

We now turn to the main theorem, which states that we can, in a sense, lift
fixpoints of functors in $\cC$ to fixpoints of fibred polynomials endofunctors
over $\Container(\cC)$.

\begin{theorem}
  \label{thm:fp-exist}
If
$F : \Container(\cC) \to \Container(\cC)$
is a fibred polynomial functor and $\cC$ is an lextensive category with 
dependent $\cW$ and $\cM$-types, then
\begin{itemize}
\item $F$ has an initial algebra $(\mu F, a_{\mu F})$
\item $F$ has a terminal coalgebra $(\nu F, c_{\nu F})$
\item $F$ has a (co)algebra
$(\zeta F, b_{\zeta F})$ such that $(\shape(\zeta F), \shape(b_{\zeta F}))$ is a
final coalgebra for
$F_0$ and it induces a final coalgebra for the endofunctor induced by $F$
over $\bigslant{\cC}{F_0(\zeta F)}$.
\end{itemize}
\end{theorem}

The basic idea behind the proof of~\Cref{thm:fp-exist} is that we may compute
a fixpoint for a fibred polynomial endofunctor $F$ by first computing a
fixpoint $\gamma F_0$ for the induced polynomial functor $F_0$ in the base,
and then take a fixpoint of $i^* \circ F_{\gamma F_0}$ where
$i : \gamma F_0 \isoto F_0(\gamma F_0)$. In both steps, we have complete freedom
over which fixpoint we want to compute, so we may use either initial or terminal
coalgebras with the outcomes summarized in~\Cref{fig:fixpointKindMatrix}. Note
that there is a repeated entry in the matrix, due to the fact that taking
the initial algebra in the base determines all fixpoints of $i^* \circ F_{\gamma F_0}$
(\Cref{prop:2algebraMu}).

\begin{figure}
  \begin{center}
    \begin{tabular}{|c||c|c|}
\hline
\diagbox{base}{total space} & \makecell{initial \\ algebra $\mu$} &
\makecell{terminal \\ coalgebra $\nu$} \\
\hhline{|=||=|=|}
\makecell{initial algebra $\mu$} &
\makecell{initial \\ algebra $\mu$} &
\makecell{initial \\ algebra $\mu$}\\
\hline
\makecell{terminal coalgebra $\nu$} &
\makecell{terminal \\ coalgebra $\nu$} &
\makecell{(co)algebra $\zeta$}
\\
\hline
\end{tabular}
  \end{center}

\caption{Matrix of the kinds of fixpoints we get from applying the recipe
outlined after~\Cref{thm:fp-exist} and~\Cref{prop:2algebraMu}.
}
\label{fig:fixpointKindMatrix}
\end{figure}

\begin{proof}
We only treat the first item, i.e. showing that $F$ has an initial algebra.
The other algebras are built in a similar fashion, and the terminal coalgebra
is also shown to be terminal in a similar way.
\mysubparagraph{The carrier} Since $\cC \cong \slice{\cC}{1}$, $F_0 : \cC \to \cC$
has an initial algebra $\left(\mu, \inAlg : F_0(\mu) \to \mu\right)$.
By Lambek's lemma, $\inAlg$ is an isomorphism.
Similarly, the induced functor $F_\mu : \slice{\cC}{\mu} \to
\slice{\cC}{F_0(\mu)} \cong \slice{\cC}{\mu}$ has a terminal
coalgebra
$\left(\nu, \outCoAlg : \nu \to F_\mu(\nu)\right)$ in $\slice{\cC}{\mu}$.
  Together, since $F_{\nu} = \inAlg^*(F_\mu(\nu))$,
  these give a representative $(\inAlg, \outCoAlg)$ for a container morphism $F(\nu) \to \nu$.

\[\begin{tikzcd}
  F(\nu)_{\mathrm{dir}} \arrow[d, "F(\nu)"]
  & \nu \arrow[d] \arrow[l, "\outCoAlg", swap] \arrow[r, equal] \arrow[dr, phantom, "\lrcorner", very near start, color=black]
  & \nu \arrow[d, "\nu"]
  \\
  F_0(\mu) \arrow[r, equal]
  & F_0(\mu) \arrow[r, "\inAlg"]
  & \mu
\end{tikzcd}\]

\mysubparagraph{Weak initiality} Suppose we have an $F$-algebra with carrier $P : A \to I$ as follows:
\[\begin{tikzcd}
  \directions{F(P)} \arrow[d, "F(P)"]
  & B \arrow[d, "b"] \arrow[l, "\beta", swap] \arrow[r] \arrow[dr, phantom, "\lrcorner", very near start, color=black]
  & A \arrow[d, "P"]
  \\
  F_0(I) \arrow[r, equal]
  & F_0(I) \arrow[r, "\alpha"]
  & I
\end{tikzcd}\]

Let's now construct a container morphism $\nu \to P$.
Since $\left(\mu, \inAlg\right)$ is an initial $F_0$-algebra and $\alpha$
is an $F_0$-algebra, there exists a unique $F_0$-algebra homomorphism
$\fold{\alpha} : \mu \to I$. Pulling back the commuting square for this
homomorphism along $P$, we obtain the cube below where all 
vertical faces are pullback squares and $C \to D$ is an isomorphism.
\[\begin{tikzcd}%
  & D \arrow[dd, "R", near start] \arrow[rr]
    \arrow[dr, phantom, "\lrcorner", very near start, color=black]
  && A \arrow[dd, "P"]
  \\
  C \arrow[dd] \arrow[rr, crossing over] \arrow[ur, "\sim"] \arrow[dr, phantom, "\lrcorner", very near start, color=black]
  && B \arrow[ur] \arrow[dr, phantom, "\lrcorner", very near start, color=black]
  \\
  & \mu \arrow[rr, "\fold{\alpha}", near start]
  && I
  \\
  F_0(\mu) \arrow[rr, "F_0(\fold{\alpha})"] \arrow[ur, "\inAlg"]
  && F_0(I) \arrow[ur, "\alpha"]
     \arrow[from=uu, crossing over, "Q", near end, swap]
\end{tikzcd}\]

The rearmost square is a pullback square, which can be viewed as a
(representative of a) cartesian morphism between the containers
$R : D\to \mu$ and $P : A \to I$. Since $F$ is a fibred functor,
applying it to that square yields another pullback square, and hence a
  unique arrow $\gamma : C \to \directions{F(R)}$ making the diagram
  below commute.
\[\begin{tikzcd}
  C \arrow[ddr, bend right] \arrow[rrr] \arrow[dr, dashed, swap, "\gamma"]
  &&& B \arrow[d, "\beta"] \arrow[dd, bend left=70, "Q"]
  \\
  & \directions{F(R)} \arrow[d, "F(R)"] \arrow[rr]
     \arrow[drr, phantom, "\lrcorner", very near start, color=black]
  && \directions{F(P)} \arrow[d, "F(P)"]
  \\
  & F_0(\mu) \arrow[rr, "F_0(\fold{\alpha})"]
  && F_0(I)
\end{tikzcd}\]

The arrow $D \cong C \xrightarrow{\gamma} \directions{F(R)}$ in
$\cC$ yields a coalgebra $\mathfrak{d} : R \to F_\mu(R)$ in $\slice{\cC}{\mu}$,
and hence there is a unique coalgebra morphism $\unfold{\gamma} : R \to \nu$.
It follows that we have a container morphism $(\fold{\alpha},\unfold{\gamma}) : \nu \to P$.

\mysubparagraph{Uniqueness}
We now check that the $(\fold{\alpha},
\unfold{\gamma})$ is the unique algebra morphism from
 $(\nu, (\inAlg, \outCoAlg))$ to $(P, (\alpha,\beta))$.

The first component of any such algebra morphism is an $F_0$-algebra
morphism from $(\mu, \inAlg)$ to $(I, \alpha)$. Since $(\mu, \inAlg)$
is initial, that component therefore must be $\fold{\alpha}$.

Similarly, the second component must be an $F_\mu$-coalgebra morphism
from $(R, \mathfrak{d})$ to $(\nu, \outCoAlg)$ in $\slice{C}{\mu}$.
Since the morphisms $C \to F_\mu(\nu)$ are equal in both diagrams
below, it follows that the second component must be equal to $\unfold{\gamma}$.

\[\begin{tikzcd}
  & \directions{F(R)} \arrow[dl, "F(\unfold{\gamma})" swap,  bend right]
  & & C \arrow[d, "\sim" vertical] \arrow[ll, "\gamma"]
  \\
  F_\mu(\nu) \arrow[d, "F(\nu)"]
  & \nu \arrow[d] \arrow[l, "\sim", swap] \arrow[r, equal] \arrow[dr, phantom, "\lrcorner", very near start, color=black]
  & \nu \arrow[d]
  & D \arrow[d] \arrow[l, "\unfold{\gamma}", swap] \arrow[r] \arrow[dr, phantom, "\lrcorner", very near start, color=black]
  & A \arrow[d, "P"]
  \\
  F_0(\mu) \arrow[r, equal]
  & F_0(\mu) \arrow[r, "\sim"]
  & \mu \arrow[r, equal]
  & \mu \arrow[r, "\fold{\alpha}"]
  & I
\end{tikzcd}\]

\[\begin{tikzcd}
  & C \arrow[dd, bend right=60] \arrow[drrr, bend left=20] \arrow[d, dashed, "\exists!"]
      \arrow[drr, bend left=10, dashed, "\exists!"]
  \\
  F_\mu(\nu) \arrow[d, "F(\nu)"]
  & \directions{F(R)} \arrow[d] \arrow[l] \arrow[r] \arrow[dr, phantom, "\lrcorner", very near start, color=black]
  & \directions{F(P)} \arrow[d, "F(P)"]
  & B \arrow[d, "Q"] \arrow[l] \arrow[r] \arrow[dr, phantom, "\lrcorner", very near start, color=black]
  & A \arrow[d, "P"]
  \\
  F_0(\mu) \arrow[r, equal]
  & F_0(\mu) \arrow[r, "F_0(\fold{\alpha})", swap]
  & F_0(I) \arrow[r, equal]
  & F_0(I) \arrow[r, "\alpha", swap]
  & I
\end{tikzcd}\]

\end{proof}
\begin{restatable}{proposition}{algebraMuTwo}
  \label{prop:2algebraMu}
If $(F_0,F)$ is a fibred polynomial functor over a $\PiWM$-category $\cC$,
$\alpha : F_0(\mu F_0) \to \mu F_0$ is an initial algebra for $F_0$ and we have
$S_0, S_1 \in \slice{\cC}{\mu F_0}$ with $F_{\mu F_0}(S_i) \cong S_i \circ \alpha$
for $i \in \{0,1\}$, then $S_0 \cong S_1$.
\end{restatable}
\begin{example}
$\mu (X \mapsto X)$ is an initial object, $\nu (X \mapsto X)$ is a terminal object and
$\zeta (X \mapsto X)$ is the monoidal unit $\unit$ of $\tensor$.
\end{example}
\begin{example}
  \label{ex:NinftyZeta}
$X \mapsto 1 + X$ admits the following fixpoints:
\begin{itemize}
\item $\mu(X \mapsto 1 + X)$ is the container $0 \to \bbN$
\item $\nu(X \mapsto 1 + X)$ is the container $0 \to \Ninfty$
\item $\zeta(X \mapsto 1 + X)$ is the container $\infty : 1 \to \Ninfty$.
\end{itemize}
\end{example}
Finally,~\Cref{thm:fp-exist} also behaves well in presence of parameters.
\begin{restatable}{theorem}{fpexistmulti}
  \label{thm:fp-exist-multi}
Let $F : \Container(\cC)^{k + 1} \to \Container(\cC)$ be a fibred polynomial
functor. For any $\gamma \in \{\mu,\nu,\zeta\}$, the map
\[
  \begin{array}{llcl}
& \Obj{\Container(\cC)^k} &\longto& \Obj{\Container(\cC)}\\
 & (f_0, \ldots, f_{k - 1}) &\longmapsto& \gamma F(f_0, \dots, f_{k-1}, -)\\
  \end{array}
\]
extends to a fibred polynomial functor
$\gamma_k F : \Container(\cC)^k \to \Container(\cC)$.
\end{restatable}

\section{Operators definable as fixpoints}
\label{sec:operators}

Assuming a $\PiWM$-category $\cC$, we discuss a number of
endofunctors over $\Container(\cC)$ defined in terms of fixpoints as per~\Cref{thm:fp-exist}.
Those useful for the sequel are summarized in~\Cref{fig:operators}.

\subsection{Parallelizations}

Fix a container $P : A \to I$ (write $(A_i)_{i : I}$ for the corresponding
family) and let us consider fixpoints of the
endofunctor $X \mapsto P \times X$ over $\Container(\cC)$.
We can first note that we have the following isomorphisms in $\cC$.
\[
  \begin{array}{l@{\;}c@{\;}lr}
\shape(\mu(X \mapsto P \times X)) &\cong& \mu(X \mapsto I \times X) \cong 0 & \\
  \shape(\gamma(X \mapsto P \times X)) &\cong& \nu(X \mapsto I \times X) \cong I^\bbN
            &\text{for $\gamma \in \{\nu, \zeta\}$}
\end{array}\]
As a result, we have $\mu(X \mapsto P \times X) \cong 0$. The situation with the other
fixpoints is slightly more exciting. The greatest fixpoint yields the 
family $\left(\sum\limits_{n : \bbN} A_{s(n)}\right)_{s : I^\bbN}$. Interpreted as a transformation
of problems, the corresponding operator takes a problem $P$ to the problem $\widetilde{P}$ where
a question is a sequence of $P$-questions, and an answer picks an index $n$
and answers the $n$th question.
\[
\begin{array}{llcl}
  \widetilde{-} :  \Container(\cC) &\longto& \Container(\cC) \\
                   P &\longmapsto& \nu(X \mapsto P \times X)
\end{array}
\]
The $\zeta$ fixpoint yields the family
$\left(\sum\limits_{x : \Ninfty} \prod\limits_{n : \bbN} \left(x = \underline{n} \to A_{s(n)}\right)\right)_{s : I^\bbN}$.
In terms of problems, an answer now consists of a conatural number and, if this conatural
number corresponds to a natural number, then an answer is eventually given for the relevant question.

In terms of degrees, $\widetilde{P}$ is not very interesting because
$P \equiv \widetilde{P}$: we can reduce $P$ to $\widetilde{P}$ by asking the
same question $\omega$ times, and conversely, we can reduce $\widetilde{P}$
to $P$ by solving only the first $P$-question in the input to $\widetilde{P}$.
It is thus typically more interesting
to consider another operator where one gets a stream of answers for all questions
asked; such an operator is called the (infinite) parallelization
operator~\cite[Definition 1.2]{survey-brattka-gherardi-pauly}. For this
we naturally consider a fixpoint of the functor $X \mapsto P \tensor X$
induced by the binary parallelization $\tensor$. Since $X \mapsto P \tensor X$
and $X \mapsto P \times X$ induce the same endofunctor $Y \mapsto I \times Y$
over $\cC$, the shapes of the fixpoints are the same and $\mu(X \mapsto P \tensor X) \cong 0$.
This time around, the greatest fixpoint is of little interest, but the
middle fixpoint $\zeta$ gets us what we want; we write $\widehat{P}$
for the infinite parallelization of $P$ defined thus.
\[
\begin{array}{llcl}
  \widehat{-} :  \Container(\cC) &\longto& \Container(\cC) \\
                   P &\longmapsto& \zeta(X \mapsto P \tensor X)
\end{array}
\]

It is also customary to consider a finitary version of parallelization by
taking the least fixpoint of $X \mapsto \unit + P \tensor X$. The finite
parallelization of $P$ is typically denoted $P^*$.

\subsection{Iterations}

Now, let us turn to fixpoints built with the sequential composition of problems $\star$.
As $\star$ is not commutative there are, a priori, two ways of defining
functors whose fixpoints are iterations of some problem $P$.
However, one of these is not fibred (\Cref{rem:triangleNotfibred}) and the
other is (\Cref{lem:trianglefibred}). We will thus mostly focus on the latter,
although it is worth noting that $X \mapsto \unit + P \star X$ does have a non-trivial
initial algebra with a natural interpretation. It corresponds to the problem
whose instances $(n , f)$ are given by a number $n$ of successive questions we want to
ask to $P$ and then a recipe $f$ to ask those questions, and where results are
the answers to $n$ questions we could have successively asked.

\begin{restatable}{proposition}{boundedDiamond}
\label{prop:boundedDiamond}
The carrier of the initial algebra for $X \mapsto \unit + X \star P$ is isomorphic to
$\sum\limits_{n : \mathbb{N}} \underbrace{P \star \ldots \star P}_{\text{$n$ times}}$.
\end{restatable}

On the other hand, the functor $X \mapsto \unit + X \star P$ is
fibred polynomial by Lemmas~\ref{lem:stufffibred} and~\ref{lem:trianglefibred}.
So~\Cref{thm:fp-exist-multi}
says we have a corresponding polynomial endofunctor $(-)^\diamond : \Container(\cC) \to \Container(\cC)$.
Regarding $P^\diamond$ as a problem, inputs are recipes to ask finitely
many questions sequentially to $P$ and outputs are possible sequences
of answers to that. Note that in general, this is strictly more powerful than
the operator from~\Cref{prop:boundedDiamond} as the number of questions need
not be known in advance.

While finite iterations of problems have been introduced in 2016
to the literature on Weihrauch complexity~\cite{topol-comput-neumann-pauly},
the endofunctor $(-)^\diamond$ has been known for much longer to category theorists:
$\functorify{P^\diamond}$ was characterized as the free monad over the
polynomial functor $\functorify{P}$~\cite{kelly1980unified}.
The construction given in~\cite[Section 4]{GKpoly}
correspond to what is given by~\Cref{thm:fp-exist} using the $\mu/\mu$
recipe.
From that point of view, assuming for a moment $P \in \Container(\Set)$,
the endofunctor $\functorify{P^\diamond}$ may be described
as follows. The shape of $P^\diamond$
consists of well-founded trees labelled by some $i : I$ or by a special symbol $\bullet$, and
the arity of the node is $P^{-1}(i)$ or $0$ if the special symbol
is used. Then $\functorify{P^\diamond}(X)$
consists of all those trees from $P^\diamond$ where every leaf $\ell$ is labelled by
some $x_\ell \in X$.

\begin{remark}
While $\shape(P^\diamond)$ can, a fortiori, be regarded as a sort of well-founded
tree structure in $\Container(\ReprSp)$, it is important to note that those trees
are very different from the ones we encode in $\Baire$ for the rest of this paper:
any node has potentially continuum-many children and knowing ``the arity''
of a node corresponds to the ability to answer the $P$-question of its label,
which is typically impossible to do continuously.
\end{remark}

A key property of $(-)^\diamond$ is that it associates well with composition
of other problems, which is captured by the following.

\begin{restatable}{theorem}{kleeneinduction}
\label{thm:kleeneinduction} $Q \star P^\diamond$ is the carrier of an initial algebra of the fibred polynomial endofunctor
$X \mapsto Q + X \star P$.
\end{restatable}

One useful application of \Cref{thm:kleeneinduction} is a straightforward
construction of a monad multiplication
$\functorify{P^\diamond} \circ \functorify{P^\diamond} \cong \functorify{P^\diamond \star P^\diamond} \to \functorify{P^\diamond}$
by exhibiting an algebra $X \mapsto P^\diamond + X \star P$ whose carrier is $P^\diamond$.
We leave that as an exercise, together with the consequences that $P^\diamond$ is the
free polynomial monad over $P$ and that $(-)^\diamond$ is a monad.
\Cref{thm:kleeneinduction} can also gives an induction principle on degrees:
\[ Q \star P ~\le~ Q \qquad\Longrightarrow\qquad Q \star P^\diamond ~\le~ Q\]
This scheme corresponds to an induction axiom of the theory of right-handed Kleene
algebras~\cite{Kozen-Silva-2012}. It is also a key part of a complete axiomatization
of the Weihrauch lattice augmented with iterated composition~\cite{pradic25}, or
rather, a slightly generalized version of it which may be derived using the
cartesian closed structure of $\Container(\cC)$.

As finite parallelization has a infinitary countable counterpart, so does
sequential iteration by considering a suitable fixpoint of the
endofunctor $X \longmapsto X \star P$.
Much like in the case of parallelization, the more interesting operator
will be derived using the middle fixpoint construction $\zeta$.
Then $P^\infty := \zeta X. \; X \star P$ corresponds to the problem whose
instances are recipes to ask countably many questions to $P$ in a sequential
manner and whose solutions are $\omega$-sequences of possible answers. An analogous operator
has recently been introduced in Weihrauch reducibility~\cite{brattka2025loops}.
In the case of $\Set$, the corresponding endofunctor $\functorify{P^\infty}$
maps $X$ to the set of pairs $(t, x)$ consisting of a tree $t$ whose nodes
are all labelled by $P$-questions, arities consist of the $P$-answers, and $x$
maps infinite paths $p$ through $t$ to elements $x(p) \in X$.

When it comes to the less interesting fixpoints, similarly to what happened
with $X \mapsto P \times X$, we have $\mu(X \mapsto X \star P) \cong 0$
and $\nu(X \mapsto X \star P) \cong 1 \star P^\infty$. It is
clear that since $(1 \star P)^\infty \cong 1 \star P$, those
two constructions yield idempotent operators on $P$. However, $(-)^\infty$ is
not idempotent. From the perspective of Weihrauch complexity,
we have that $P^{\infty\infty}$ is the problem where one can ask $\omega^2$
questions to $P$ sequentially and then get all of the answers, which is more powerful
than asking $\omega$ questions sequentially (for instance, 
$\lpo^\infty < \lpo^{\infty\infty}$~\cite[Corollary 23]{brattka2025loops}).

\section{Ground problems}
\label{sec:ground}

\subsection{$\zeta$-expressions}

As announced before, we may define the following syntax for fibred
polynomial functors over $\Container(\cC)$ when $\cC$ is a $\PiWM$-category.
\[ E, E' \enspace \bnfeq\enspace X \bnfalt \mu X. \; E \bnfalt \zeta X. \; E \bnfalt \nu X. \; E \bnfalt E \times E' \bnfalt E \tensor E' \bnfalt E + E'\]
We write $\zetaExpr(\vec X)$ for the set of such \emph{$\zeta$-expressions} with free variables in $\vec X$.
$\zeta$-expressions $E$ with $k$ free variables have an obvious interpretation as
fibred polynomial functors $\ffunctorify{E} : \Container(\cC)^k \to \Container(\cC)$.
When $k = 0$, we regard $\ffunctorify{E}$ as an object of $\Container(\cC)$.
We will also use
the following obvious shorthands for $\zeta$-expressions
(where $X$ is a fresh variable):
\[
  \begin{array}{l@{\;}c@{\;}l !\quad l@{\;}c@{\;}l !\quad l@{\;}c@{\;}l}
  0 &\eqdef& \mu X. X &
  \unit &\eqdef& \zeta X. X &
  1 &\eqdef& \nu X. X \\
  E^* &\eqdef& \mu X. \unit + E \tensor X &
  \widehat{E} &\eqdef& \zeta X. E \tensor X &
  \widetilde{E} &\eqdef& \nu X. E \times X \\
\end{array}
\]
Before turning to examples of $\zeta$-expressions we wish to interpret
in $\Container(\ReprSp)$, we will now give an automata-theoretic characterization
of their denotations, extending a characterization of $\mu$-expressions
alluded to in~\Cref{sec:background}. For this, we assume some familiarity
with the theory of regular languages of infinite trees~\cite[\S 14.3]{Toolbox}
and write $\Tree(\Gamma, \chi)$ for the set of trees whose internal nodes
are binary and labelled by elements of $\Gamma$ and whose leaves are labelled by
elements of $\chi$ (if $\chi = \emptyset$, we may write $\Tree(\Gamma)$).

\begin{definition}
  \label{def:gamingAut}
  The \emph{gaming alphabet} $\cG_{i,k}$ with ranks $(i, k)$ (for $i \le k \in \bbN$)
is the finite set $\{\bot\} ~\uplus~ \{\Even, \Odd\} \times \{i, \ldots, k\}$.
A \emph{binary game tree} with ranks $(i, k)$ and exits $\chi$ is an
element $ \ell \in \Tree(\cG_{i,k}, \chi)$ such that $\ell(w) = \bot$ implies
$\ell(wi) = \bot$ for any $w \in 2^*$, $i \in 2$.
An $\Even$-\emph{strategy} of a binary game tree $\ell$ is a subset $S \subseteq \dom(\ell)$
such that $\varepsilon \in S$, $S \cap \ell^{-1}(\bot) = \emptyset$ and, for
every $w \in S$:
\begin{itemize}
\item if $\pi_1(\ell(w)) = \Even$, then there exists exactly one $i \in 2$ with
$wi \in S$
\item otherwise if $\pi_1(\ell(w)) = \Odd$, for every $i \in 2$ such that
$\ell(wi) \neq \bot$, $wi \in S$.
\end{itemize}
A strategy $S$ is \emph{winning} for $\Even$ if for any infinite path
$p \in \Cantor$ through $S$, $\limsup\limits_{n \to +\infty} \pi_2(\ell(p[n]))$ is even.
Let us call $\Strat(\ell)$ the set of such winning strategies.
\end{definition}

For any finite set $\chi$, any game tree language
$L \subseteq \Tree(\cG_{i, k}, \chi)$ induces a fibred polynomial functor
$\ffunctorify{L} : \Container(\ReprSp)^\chi \to \Container(\ReprSp)$.
To describe its action on tuples of containers, first recall
(from the discussion leading to~\Cref{prop:muExprAut}) that $L$
also determines a functor $\functorify{L} : \ReprSp^\chi \to \ReprSp$
mapping a tuple of spaces $(V_x)_{x \in \chi}$ to the space of pairs $\ell \in L$
and a family of maps $\rho_x : \ell^{-1}(x) \to V_x$ that labels the paths
of $\ell$ ending in $x \in \chi$ with some element of $V_x$. We compute the shape
of $\ffunctorify{L}$ by $\shape(\ffunctorify{L}) = \functorify{L}$. Intuitively,
the answers corresponding to a pair $(\ell, \rho)$ will then consist of a winning strategy
for $\Even$ in $\ell$ together with answers for the questions occurring on
maximal finite paths of the strategy.
All in all, writing $S_x \eqdef \ell^{-1}(x) \cap S$ for $S \in \Strat(\ell)$ and $x \in \chi$,
we have
\[
  \begin{array}{l@{\;}lcl}
\ffunctorify{L} :& \Container(\ReprSp)^{\chi} &\longto& \Container(\ReprSp) \\
                 & ((A_i)_{i \in I_x})_{x \in \chi} &\longmapsto&
                 \left(\sum\limits_{S \in \Strat(\ell)} \prod\limits_{x \in \chi} \prod\limits_{p \in S_x} A_{\rho(p)} \right)_{(\ell, \rho) \in \functorify{L}}
\end{array}\]

\begin{restatable}{theorem}{zetaExprAut}
  \label{thm:zetaExprAut}
From a $\zeta$-expression $E$ with free variables $\chi$, %
we can compute a regular game tree language $L_E$
such that $\ffunctorify{E}$ and $\ffunctorify{L_E}$ are isomorphic in $\Container(\ReprSp)^{\chi} \to \Container(\ReprSp)$.
\end{restatable}

Let us now unroll the construction behind
\Cref{thm:zetaExprAut}
on a number of examples.

\begin{example}
  We have $L_{\mu X. X} = \emptyset$ and $L_{\gamma X. X} = \{\ell_\gamma\} \subseteq \Tree(\cG_{0,1})$
for $\gamma \in \{\nu,\zeta\}$
where we have $\ell_\gamma(p) = \bot$ if $p$ contains a $1$, $\ell_\nu(0^n) = (\Even, 1)$
and $\ell_{\zeta}(0^n) = (\Even, 0)$.
\end{example}

\begin{example}
\label{ex:zetaRT}
The language $L_{\zeta X. \nu Y. Y + X}$ includes only trees $\ell \in \Tree(\cG_{0,1})$ where
$\ell$ is constantly $\bot$ except along a unique infinite path $p$ belonging to the
regular 
language of infinite words $00((1 + \varepsilon)00)^\omega$. The labelling is fully
determined by the path in question: $\ell(w1) = (\Even, 0)$ and
$\ell(w0) = (\Even, 1)$ for any $w \sqsubset p$.
Accordingly, $\Strat(\ell)$ is $\{p\}$ if $\exists^\infty n \in \bbN. \; p_n = 1$
and $\emptyset$ otherwise.
\end{example}

\begin{example}
\label{ex:zetaWKL}
The tree language $L_{\zeta X. \; 1 + X \times X}$ is recognized by the following
B\"uchi tree automaton assuming $a = (\Even, 0)$ and all states are accepting:
  \begin{center}
\begin{tikzpicture}[shorten >=1pt,node distance=2cm and 4cm,on grid,auto]
\tikzstyle{every node}=[font=\small]
\node[state,initial, initial where=above, initial text={}]  (q_0)                      {$X$};
  \node[state]          (q_1) [left=of q_0] {$1 \textcolor{gray}{+X^2}$};
  \node[state]          (q_2) [below right=of q_0] {$\textcolor{gray}{1+} X^2$};
  \node[state]          (q_4) [above=of q_2] {$X^2$};
  \node[state]          (q_3) [below=of q_1] {$1$};
  \node[state]          (b) [below=of q_0] {$\bot$};

  \path[->] (q_0) edge              node [above] {$a ; 0$} (q_1)
                  edge              node [above] {$a ; 0$} (q_2)
                  edge              node         {$a ; 1$} (b)
            (b)   edge [loop below] node         {$\bot; 0,1$} ()
            (q_1) edge              node        {$a ; 0$} (q_3)
                  edge              node        {$a ; 1$} (b)
            (q_2) edge              node [swap] {$a ; 1$} (q_4)
                  edge              node        {$a ; 0$} (b)
            (q_4) edge              node [swap] {$a ; 0,1$} (q_0)
            (q_3) edge              node [swap] {$a ; 0,1$} (b);
\end{tikzpicture}
  \end{center}

For any $\ell \in L_{\zeta X. \; 1 + X \times X}$, we have that $\ell^{-1}(a)$ is
a subtree of the full binary tree $2^*$. By keeping every third letter of those
$w \in \ell^{-1}(a)$, we would obtain a subtree $\ell'$ of $2^*$ where every
node is either a leaf or has two children; it is not difficult to see that all such
subtrees can be obtained in this way. Winning $\Even$ strategies in the
game tree $\ell$ are in bijection with infinite paths of $\ell'$.
\end{example}

\begin{example}
\label{ex:zetaClopenDet}
Let us now consider $D_1 = \mu X. \; \widehat{\widetilde{X}} + 1 + \unit$.
The tree language $L_{D_1} \subseteq \Tree(\cG_{0,1})$ is recognized by the coB\"uchi automaton below, where
$a = (\Even, 0) \in \cG_{0,1}$, $X$ is the only rejecting state and $\bot$
denotes a subautomaton recognizing only the constant labelling by $\bot$ of
the full binary tree.
The states are labelled with the corresponding subexpressions of $D_1$, except
for $e$ and $c$ which should be respectively read as ``end'' and ``continue''. 
\begin{center}
\begin{tikzpicture}[shorten >=1pt,node distance=1.5cm and 3cm,on grid,auto]
\tikzstyle{every node}=[font=\scriptsize]
\node[state,initial, initial where=above, initial text={}, accepting]  (X)                      {$X$};
  \node[state]          (E) [below right=of X] {$e$};
  \node          (b) [below left=of E] {$\bot$};
  \node[state]          (C) [left=of b] {$c$};
  \node[state]          (I) [above=of E] {$\textcolor{gray}{1+}\unit$};
  \node[state]          (II) [right=of I] {$\unit$};
  \node[state]          (O) [below right=of E] {$1\textcolor{gray}{+\unit}$};
  \node[state]          (OO) [below =of E] {$1$};
  \node[state]          (ODD) [above=of C] {$\widehat{\widetilde{X}}$};
  \node[state]          (EVE) [above=of ODD] {$\widetilde{X}$};
  \node          (bb) [below=of II] {$\bot$};

  \path[->] (X)   edge              node [above] {$a ; 0$} (E)
                  edge              node [left] {$a; 0$} (C)
                  edge              node [right] {$a; 1$} (b)
            (E)   edge              node [left] {$a ; 1$} (I)
            edge   node [right] {$a;  1$} (O)
            edge   node [right] {$a;  0$} (b)
            (C)   edge   node [left] {$a ; 0$} (ODD)
                  edge              node [below] {$a ; 1$} (b)
            (ODD) edge  node [left] {$(\Odd, 0) ; 0$} (EVE)
                  edge  [loop left] node [left] {$(\Odd, 0) ; 1$} ()
                  (EVE) edge  [bend left]  node [above] {$(\Even, 1) ; 0$} (X)
                  edge  [loop left] node [left] {$(\Even, 1) ; 1$} ()
            (O)   edge              node [below] {$a ; 0$} (OO)
                 edge              node [right] {$a ; 0$} (bb)
            (OO)   edge              node [below] {$a ; 0, 1$} (b)
            (II)   edge                node [right] {$(\Odd, 0) ; 0,1$} (bb)
            (I)   edge                node [above] {$a ; 1$} (II)
                  edge                node [above] {$a ; 0$} (bb)
            ;
\end{tikzpicture}
\end{center}
More intuitively, one can see that on shapes, $\ffunctorify{D_1}$ is the carrier
of an initial algebra $X \mapsto 2 + X^{\bbN^2}$. It thus consists of
well-founded trees where all nodes are either leaves labelled by a boolean or
are internal nodes of arity $\bbN^2$.
A direction corresponding to such a tree is a winning strategy for
the second player in the following game whose positions are paths in the tree:
\begin{itemize}
  \item if the path is a leaf, then the first player wins if and only if its
    label is $0$,
  \item otherwise, the first player first picks a number $n \in \bbN$, the second
    player picks a number $m \in \bbN$, and the game proceeds from the $(n,m)$th
    child of the current position.
\end{itemize}
Questions to $\ffunctorify{D_1}$ seen as a problem are essentially clopen
games over $\Baire$ and answers are winning strategies for the second player.
\end{example}

\subsection{Some Weihrauch degrees as answerable parts of $\zeta$-expressions}

Let us now sketch how the fragment of the Weihrauch lattice described in~\Cref{fig:pblmsAsZetaMKA}
can be captured closed $\zeta$-expressions, or more accurately, the \emph{answerable
part} (\Cref{def:answerablePart}) of their denotations. The latter step is important
as all $\zeta$-expressions are trivial from a degree-theoretic point of view.

\begin{proposition}
For any closed $\zeta$-expression $E$,
$\ffunctorify{E} \equiv 0$, $\unit$ or $1$. 
\end{proposition}
\begin{proof}[Proof idea]
We generalize the statement inductively to open expressions to
state that the set of degrees $\{0, \unit, 1\}$ are preserved by any $\ffunctorify{E}$.
The more interesting case is the fixpoint case. In that case, we can show that
for any fibred polynomial $F : \Container(\cC) \to \Container(\cC)$ that preserve the set
of degrees $\{0, \unit, 1\}$, we have
$\mu F \equiv F(F(0))$, $\zeta F \equiv F(\unit)$ and $\nu F \equiv F(F(1))$.
\end{proof}

For instance, the expression $D_1$ from~\Cref{ex:zetaClopenDet} gives $\ffunctorify{D_1} \equiv 1$
because one can compute a clopen game where $\Even$ loses. On the other hand,
$\mkAnswerable(\ffunctorify{D_1})$ captures $\mathbf{\Delta}^0_1$-determinacy,
a non-trivial Weihrauch problem equivalent to $\UC_\Baire$.

\begin{remark}
  \label{rem:genuineGroundProblems}
By~\Cref{rem:countableExponentials}, we can inspect the definition of $\ffunctorify{L}$
to check it
always restricts to a functor $\Container(\rpReprSp)^\chi \to \Container(\rpReprSp)$.
Hence, by~\Cref{thm:zetaExprAut}, the answerable parts of denotations of $\zeta$-expressions correspond to genuine Weihrauch degrees (in contrast to the
composition operator $\star$ and its iterations in the previous section).
\end{remark}

Let us now treat a selection of examples from~\Cref{fig:pblmsAsZetaMKA}, starting
with the first example of a degree we introduced in the paper.%

\begin{proposition}
\label{prop:lpo}
$\LPO$ is Weihrauch-equivalent to
$\mkAnswerable\left(\ffunctorify{\left(\zeta X. \; X + 1\right) \times \left(\nu X. \; X + \unit\right)}\right)$.
\end{proposition}
\begin{proof}
Computing the isomorphism types of
$\ffunctorify{\zeta X. \; X + 1}$
and
$\ffunctorify{\nu X. \; X + \unit}$,
we obtain the containers corresponding respectively to $\infty : 1 \to \Ninfty$ and
to the embedding of $\bbN$
into $\Ninfty$. We thus
have two problems that take as input some $x \in \Ninfty$ which can have a
trivial answer. That answer exists only if $x = \infty$ in the first case, and
conversely only if $x = \underline{n}$ for some $n \in \Nat$ in the second case.
Hence, questions to $\ffunctorify{\left(\zeta X. \; X + 1\right) \times \left(\nu X. \; X + \unit\right)}$
are pairs 
$(x, y) \in \Ninfty^2$ and legal answers are booleans $b \in 2$ such that
\[ b = 0 \Longrightarrow x = \infty \qquad \text{and} \qquad b = 1 \Longrightarrow y < \infty\]
Taking the answerable part ensures that we only have inputs such that $y = \infty \Rightarrow x = \infty$.
$\LPO$ on the other hand is easily seen to be equivalent to the problem taking
some $x \in \Ninfty$ and outputting the boolean $b$ with $b = 0 \Leftrightarrow x =\infty$.
So we can take the forward part of a reduction $\LPO \le \ffunctorify{\left(\zeta X. \; X + 1\right) \times \left(\nu X. \; X + \unit\right)}$
to be the diagonal map $x \mapsto (x,x)$ and its backwards part to be trivial.
For the converse reduction, the forward part is simply the second projection,
and the backward part is trivial. 
\end{proof}

The next degree we consider captures \emph{closed choice problems} over (subspaces of) $\bbN$ with
the cofinite topology.

\begin{definition}
An enumeration $A \subseteq \bbN$ is a map $e : \bbN \to \{\bot\} \uplus \bbN$
such that $A = e(\bbN) \cap \bbN$.
For $A \subseteq \bbN$, define $\C_A$ to be the following problem: questions are
enumerations $e$ of strict subsets of $A$, and answers are those $m \in A$ which
are not enumerated by $A$.
\end{definition}

\begin{proposition}
\label{prop:cchoicefin}
$\C_\bbN \equiv \mkAnswerable\left(\ffunctorify{\widetilde{\zeta X. X + 1}}\right)$ and, for $k \in \bbN$,
\[
\CChoice{k} ~~ \equiv ~~ \mkAnswerable\left(\ffunctorify{\underbrace{\left( \zeta X. \; X + 1 \right) \times \ldots \times \left( \zeta X. \; X +
1\right)}_{\text{$k$ times}}}\right)\]
\end{proposition}
\begin{proof}
Recall that $\ffunctorify{\zeta X. X + 1}$ is isomorphic to the container
$\infty : 1 \to \Ninfty$ (\Cref{ex:NinftyZeta}). As a problem, a question
is a conatural number, and a question can be (trivially) answered if and only
if that conatural number is $\infty$.
This means that $\widetilde{\ffunctorify{\zeta X. X + 1}}$ is the problem
whose questions are sequences of conatural numbers $x \in \Ninfty^\bbN$
which are answered by any $n \in \bbN$ such that $x_n = \infty$; a fortiori,
$\mkAnswerable{\widetilde{\ffunctorify{\zeta X. X + 1}}}$ restricts the questions
so that an $n \in \bbN$ with $x_n = \infty$ does exist.

The forward part of a reduction $\C_\bbN \le \mkAnswerable{\widetilde{\ffunctorify{\zeta X. X + 1}}}$
simply sends an enumeration $e$ to the sequence $(\sup \{ x \in \Ninfty \mid \forall k \in \bbN. \; \underline{k} < x \Rightarrow e_k \neq n\})_{n \in \bbN}$,
and the corresponding backwards part is trivial.
Conversely, from a sequence $x \in \Ninfty^\bbN$, we can compute
\[e(\langle n, m \rangle) = \left\{ \begin{array}{ll}
n & \text{if $x_n \le \underline{m}$} \\
\infty & \text{otherwise}\\
\end{array}
\right.\]
as part of a forward reduction (assuming $\langle -, -\rangle : \bbN^2 \to \bbN$
is the standard Cantor pairing function) and have a trivial backwards reduction.
The proof for $\C_k$ is essentially the same.
\end{proof}

Another important instance of closed choice is $\C_{2^\bbN}$, which is equivalent
to weak K\"onig's lemma.

\begin{definition}
\label{def:wkl}
$\WKL$ is the problem whose questions are infinite binary trees $t \subseteq 2^*$,
and where such a question is answered by an infinite path $p \in 2^\bbN$ through
$t$.
\end{definition}

\begin{proposition}
$\WKL$ is Weihrauch equivalent to the answerable parts of
$\ffunctorify{\zeta X. \; 1 + X \times X}$ and
of the infinite parallelization of
$
{\ffunctorify{\zeta X. \;X + 1}} \times
{\ffunctorify{\zeta X. \;X + 1}}
$.
\end{proposition}
\begin{proof}
The first part was explained in~\Cref{ex:zetaWKL}; the only minor difference
with the official definition of $\WKL$ that makes one of the reductions different
from the identity is that input trees to $\zeta X. 1 + X \times X$
should not have internal unary nodes, which is remedied by adding leaves in the
forward reduction.
As for the second part, it follows from $\widehat{\C_2} \equiv \WKL$~\cite[Theorem 8.2]{brattka2},
$\mkAnswerable(\widehat{P}) \cong \widehat{\mkAnswerable(P)}$ and~\Cref{prop:cchoicefin}.
\end{proof}

\subsection{Problems on game tree languages}

We have seen that $\mkAnswerable$ dramatically increases  the expressiveness of
closed $\zeta$-expressions $E \in \zetaExpr$. We may ask whether we can see any limit to this
approach, and whether there are any nice structural properties of the set
of problems expressible in this way. The following basic observation hints that
considering containers $\ffunctorify{L}$ associated to regular game tree
languages might be worthwhile.

\begin{proposition}
For any closed $\zeta$-expression $E$, we can compute a regular game tree
language $L_E \cap W$ such that $\ffunctorify{L_E \cap W} \cong \mkAnswerable(\ffunctorify{E})$.
\end{proposition}
\begin{proof}
Assuming $L_E$ (from~\Cref{thm:zetaExprAut}) is over the alphabet $\cG_{0, k}$,
simply note that the languages $W_{0,k} \subseteq \Tree(\cG_{0,k})$ of game trees
where $\Even$ has a winning strategy is easily seen to be regular, and regular
tree languages are closed under finite intersections.
\end{proof}
Standard automata theoretic techniques lead to an effective trichotomy that
puts a (very high) bound on what we may capture by $\zeta$-expressions.
\begin{proposition}
  \label{prop:trichotomy}
For any regular game tree language $L \subseteq \Tree(\cG_{0,k})$, we are in one of
the following three cases:
\[\ffunctorify{L} \equiv 0 \qquad \unit \le \ffunctorify{L} \le \WS{(\bfS^0_2)_k} \quad \text{or} \quad \ffunctorify{L} \equiv 1\]
Furthermore, one can compute in which case we are and the appropriate reductions from
an automaton recognizing $L$.
\end{proposition}
\begin{proof}
Given any automaton for $L \subseteq \Tree(\cG_{0,k})$,
first, decide whether $L$ is empty (in which case $\ffunctorify{L} \cong 0$),
or compute a regular tree $\ell \in L$, which yields a reduction $\unit \le \ffunctorify{L}$.
Next, decide if $L \setminus W$ is inhabited. If it is, compute a regular tree $\ell' \in L \setminus W$,
which yields a reduction $1 \le \ffunctorify{L}$.
Otherwise, $L \cap W = L$, then $\ffunctorify{L}$ is answerable and we have
  a trivial reduction $\ffunctorify{L} \le \WSCantor{(\bfS^0_2)_{k}}$.
\end{proof}

This means in particular that $\mkAnswerable(\ffunctorify{L})$ cannot be a problem strictly
higher than $\WSCantor{(\bfS^0_2)_k}$s in the Weihrauch lattice.
The reductions $\ffunctorify{L} \le \WS{(\bfS^0_2)_k}$ are
quite trivial and do not guarantee that $k$ is picked optimally. So we feel
motivated to ask the following question inspired by previous work on
algebraic syntaxes for operators on Weihrauch problems~\cite{theoryWeiTimes, pradic25}.

\begin{question}
\label{q:existsReduction}
Given automata recognizing $L, M \subseteq \Tree(\cG_{0,k}, \chi)$, can we
computably decide whether the following hold or not?
\[\forall P \in \Container(\ReprSp)^\chi. \qquad \ffunctorify{L}(P) \le \ffunctorify{M}(P)\]
What about the cases where $L$ and $M$ arise from $\zeta$-expressions (as per~\Cref{thm:zetaExprAut})?
\end{question}

A positive solution to the question above with $\chi = \emptyset$ would in particular allow
us to compute, from any automaton for some $L \subseteq \Tree(\cG_{0,k})$,
the optimal $m \le k$ such that $\ffunctorify{L} \le \WS{(\bfS^0_2)_m}$ when it
exists. This is reminiscent of the Rabin-Mostowski index problem asking whether optimizing
for indices for infinite tree automata is doable computably.
This is a long-standing open problem, although
solved positively for interesting subclasses of
all tree automata, including deterministic ones~\cite{FMS16,niwinski2021guidable,idir2025using}.
However, \Cref{q:existsReduction} and its variants also involve
thinking about reductions rather than pure language recognition, so a solution
might also involve taming some non-trivial continuous infinite tree transductions.

\bibliography{bi}

\begin{thebibliography}{10}

\bibitem{abbott2005containers}
Michael Abbott, Thorsten Altenkirch, and Neil Ghani.
\newblock Containers: Constructing strictly positive types.
\newblock {\em Theoretical Computer Science}, 342(1):3--27, 2005.

\bibitem{containers03}
Michael~Gordon Abbott, Thorsten Altenkirch, and Neil Ghani.
\newblock Categories of containers.
\newblock In Andrew~D. Gordon, editor, {\em {FOSSACS} 2003 proceedings}, volume
  2620 of {\em Lecture Notes in Computer Science}, pages 23--38. Springer,
  2003.
\newblock \href {https://doi.org/10.1007/3-540-36576-1\_2}
  {\path{doi:10.1007/3-540-36576-1\_2}}.

\bibitem{AhmanBauer24}
Danel Ahman and Andrej Bauer.
\newblock Comodule representations of second-order functionals.
\newblock {\em J. Log. Algebraic Methods Program.}, 2025.
\newblock \href {https://doi.org/10.1016/J.JLAMP.2025.101071}
  {\path{doi:10.1016/J.JLAMP.2025.101071}}.

\bibitem{AB26}
Danel Ahman and Andrej Bauer.
\newblock Sheaves as oracle computations, 2026.
\newblock URL: \url{https://arxiv.org/abs/2602.22135}, \href
  {https://arxiv.org/abs/2602.22135} {\path{arXiv:2602.22135}}.

\bibitem{ArnoldNiwinski07}
André Arnold and Damian Niwiński.
\newblock Continuous separation of game languages.
\newblock {\em Fundamenta Informaticae}, 81(1-3):19--28, 2007.
\newblock \href {https://doi.org/10.3233/FUN-2007-811-303}
  {\path{doi:10.3233/FUN-2007-811-303}}.

\bibitem{bauerPhD}
Andrej Bauer.
\newblock {\em The realizability approach to computable analysis and topology}.
\newblock PhD thesis, Carnegie Mellon University, 2000.

\bibitem{Bauer22}
Andrej Bauer.
\newblock Instance reducibility and {W}eihrauch degrees.
\newblock {\em Log. Methods Comput. Sci.}, 18(3), 2022.
\newblock URL: \url{https://doi.org/10.46298/lmcs-18(3:20)2022}, \href
  {https://doi.org/10.46298/LMCS-18(3:20)2022}
  {\path{doi:10.46298/LMCS-18(3:20)2022}}.

\bibitem{Toolbox}
Mikołaj Bojańczyk and Wojciech Czerwiński.
\newblock An automata toolbox.
\newblock {D}ecember 30th 2025 version, 2018.
\newblock URL:
  \url{https://www.mimuw.edu.pl/~bojan/paper/automata-toolbox-book}.

\bibitem{BPPR14}
Filippo Bonchi, Daniela Petrisan, Damien Pous, and Jurriaan Rot.
\newblock Coinduction up-to in a fibrational setting.
\newblock In {\em CSL-LICS}. {ACM}, 2014.
\newblock \href {https://doi.org/10.1145/2603088.2603149}
  {\path{doi:10.1145/2603088.2603149}}.

\bibitem{brattka2022hagen}
Vasco Brattka.
\newblock Weihrauch complexity and the {H}agen school of computable analysis,
  2022.
\newblock URL: \url{https://arxiv.org/abs/2203.06166}, \href
  {https://arxiv.org/abs/2203.06166} {\path{arXiv:2203.06166}}.

\bibitem{brattka2025loops}
Vasco Brattka.
\newblock Loops, inverse limits and non-determinism, 2025.
\newblock URL: \url{https://arxiv.org/abs/2501.17734}, \href
  {https://arxiv.org/abs/2501.17734} {\path{arXiv:2501.17734}}.

\bibitem{brattka2}
Vasco Brattka and Guido Gherardi.
\newblock {W}eihrauch degrees, omniscience principles and weak computability.
\newblock {\em Journal of Symbolic Logic}, 76:143 -- 176, 2011.
\newblock arXiv:0905.4679.

\bibitem{survey-brattka-gherardi-pauly}
Vasco Brattka, Guido Gherardi, and Arno Pauly.
\newblock {\em {W}eihrauch Complexity in Computable Analysis}, pages 367--417.
\newblock Springer International Publishing, Cham, 2021.
\newblock \href {https://doi.org/10.1007/978-3-030-59234-9\_11}
  {\path{doi:10.1007/978-3-030-59234-9\_11}}.

\bibitem{paulybrattka4}
Vasco Brattka and Arno Pauly.
\newblock On the algebraic structure of {W}eihrauch degrees.
\newblock {\em Logical Methods in Computer Science}, 14(4), 2018.
\newblock \href {https://doi.org/10.23638/LMCS-14(4:4)2018}
  {\path{doi:10.23638/LMCS-14(4:4)2018}}.

\bibitem{BrattkaS25}
Vasco Brattka and Hendrik Smischliaew.
\newblock Computability of initial value problems.
\newblock In {\em CiE}, 2025.
\newblock \href {https://doi.org/10.1007/978-3-031-95908-0\_14}
  {\path{doi:10.1007/978-3-031-95908-0\_14}}.

\bibitem{burgess1983classical1}
John Burgess.
\newblock Classical hierarchies from a modern standpoint. part {I}. {C}-sets.
\newblock {\em Fundamenta Mathematicae}, 115(2):81--95, 1983.

\bibitem{burgess1983classical2}
John Burgess.
\newblock Classical hierarchies from a modern standpoint. part {II}. {R}-sets.
\newblock {\em Fundamenta Mathematicae}, 115:97--105, 1983.

\bibitem{CR00}
A.~Carboni and G.~Rosolini.
\newblock Locally cartesian closed exact completions.
\newblock {\em Journal of Pure and Applied Algebra}, 154(1):103--116, 2000.
\newblock \href {https://doi.org/10.1016/S0022-4049(99)00192-9}
  {\path{doi:10.1016/S0022-4049(99)00192-9}}.

\bibitem{carboni93extensive}
Aurelio Carboni, Stephen Lack, and R.F.C. Walters.
\newblock Introduction to extensive and distributive categories.
\newblock {\em Journal of Pure and Applied Algebra}, 84(2):145--158, 1993.
\newblock \href {https://doi.org/10.1016/0022-4049(93)90035-R}
  {\path{doi:10.1016/0022-4049(93)90035-R}}.

\bibitem{CiprianiMV25}
Vittorio Cipriani, Alberto Marcone, and Manlio Valenti.
\newblock {The} {W}eihrauch lattice at the level of
  $\mathbf{\Pi}^1_1$-$\mathsf{CA}_0$ : the {C}antor-{B}endixson theorem.
\newblock {\em J. Symb. Log.}, 90(2):752--790, 2025.
\newblock \href {https://doi.org/10.1017/JSL.2024.72}
  {\path{doi:10.1017/JSL.2024.72}}.

\bibitem{CD14}
Pierre Clairambault and Peter Dybjer.
\newblock The biequivalence of locally cartesian closed categories and
  {M}artin-{L}{\"{o}}f type theories.
\newblock {\em Math. Struct. Comput. Sci.}, 24(6), 2014.
\newblock \href {https://doi.org/10.1017/S0960129513000881}
  {\path{doi:10.1017/S0960129513000881}}.

\bibitem{d2021comparison}
Paul-Elliot~Angl{\`e}s d’Auriac and Takayuki Kihara.
\newblock A comparison of various analytic choice principles.
\newblock {\em The Journal of Symbolic Logic}, 86(4):1452--1485, 2021.

\bibitem{FMS16}
Alessandro Facchini, Filip Murlak, and Michal Skrzypczak.
\newblock Index problems for game automata.
\newblock {\em {ACM} Trans. Comput. Log.}, 17(4):24, 2016.
\newblock \href {https://doi.org/10.1145/2946800} {\path{doi:10.1145/2946800}}.

\bibitem{FGJ23}
Marcelo Fiore, Zeinab Galal, and Farzad Jafarrahmani.
\newblock Fixpoint constructions in focused orthogonality models of linear
  logic.
\newblock In {\em MFPS}, ENTICS, 2023.
\newblock \href {https://doi.org/10.46298/ENTICS.12302}
  {\path{doi:10.46298/ENTICS.12302}}.

\bibitem{friedman1981necessary}
Harvey Friedman.
\newblock On the necessary use of abstract set theory.
\newblock {\em Advances in Mathematics}, 41(3):209--280, 1981.

\bibitem{gambino_wellfounded_2004}
Nicola Gambino and Martin Hyland.
\newblock Wellfounded {Trees} and {Dependent} {Polynomial} {Functors}.
\newblock In Stefano Berardi, Mario Coppo, and Ferruccio Damiani, editors, {\em
  Types for Proofs and Programs}, pages 210--225, Berlin, Heidelberg, 2004.
  Springer Berlin Heidelberg.
\newblock ISSN: 0302-9743, 1611-3349.
\newblock \href {https://doi.org/10.1007/978-3-540-24849-1\_14}
  {\path{doi:10.1007/978-3-540-24849-1\_14}}.

\bibitem{GKpoly}
Nicola Gambino and Joachim Kock.
\newblock Polynomial functors and polynomial monads.
\newblock {\em Mathematical Proceedings of the Cambridge Philosophical
  Society}, 154(1):153–192, September 2012.
\newblock \href {https://doi.org/10.1017/s0305004112000394}
  {\path{doi:10.1017/s0305004112000394}}.

\bibitem{girard1988normal}
Jean-Yves Girard.
\newblock Normal functors, power series and $\lambda$-calculus.
\newblock {\em Annals of pure and applied logic}, 37(2):129--177, 1988.

\bibitem{HKC18}
Ichiro Hasuo, Toshiki Kataoka, and Kenta Cho.
\newblock Coinductive predicates and final sequences in a fibration.
\newblock {\em Math. Struct. Comput. Sci.}, 2018.
\newblock \href {https://doi.org/10.1017/S0960129517000056}
  {\path{doi:10.1017/S0960129517000056}}.

\bibitem{HirschThesis90}
Michael~David Hirsch.
\newblock {\em Applications of topology to lower bound estimates in computer
  science}.
\newblock PhD thesis, University of California, Berkeley, 1990.

\bibitem{Hyv14}
Pierre Hyvernat.
\newblock A linear category of polynomial functors (extensional part).
\newblock {\em Log. Methods Comput. Sci.}, 2014.
\newblock \href {https://doi.org/10.2168/LMCS-10(2:2)2014}
  {\path{doi:10.2168/LMCS-10(2:2)2014}}.

\bibitem{Hyv25}
Pierre Hyvernat.
\newblock Totality for mixed inductive and coinductive types.
\newblock {\em Logical Methods in Computer Science}, 2025.
\newblock \href {https://doi.org/10.46298/lmcs-21(3:19)2025}
  {\path{doi:10.46298/lmcs-21(3:19)2025}}.

\bibitem{idir2025using}
Olivier Idir and Karoliina Lehtinen.
\newblock Using games and universal trees to characterise the nondeterministic
  index of tree languages.
\newblock {\em arXiv preprint arXiv:2504.16819}, 2025.

\bibitem{jacobs01book}
B.~Jacobs.
\newblock {\em {Categorical Logic and Type Theory}}.
\newblock Studies in logic and the foundations of mathematics. Elsevier, 2001.

\bibitem{kechris}
Alexander Kechris.
\newblock {\em Classical descriptive set theory}, volume 156.
\newblock Springer Science \& Business Media, 2012.

\bibitem{kelly1980unified}
G~Max Kelly.
\newblock A unified treatment of transfinite constructions for free algebras,
  free monoids, colimits, associated sheaves, and so on.
\newblock {\em Bulletin of the Australian Mathematical Society}, 22(1):1--83,
  1980.

\bibitem{KiharaLT24}
Takayuki Kihara.
\newblock Rethinking the notion of oracle: A prequel to {L}awvere-{T}ierney
  topologies for computability theorists, 2024.
\newblock URL: \url{https://arxiv.org/abs/2202.00188}, \href
  {https://arxiv.org/abs/2202.00188} {\path{arXiv:2202.00188}}.

\bibitem{kmp20}
Takayuki Kihara, Alberto Marcone, and Arno Pauly.
\newblock Searching for an analogue of {$\rm ATR_0$} in the {W}eihrauch
  lattice.
\newblock {\em The Journal of Symbolic Logic}, 85(3):1006--1043, 2020.
\newblock \href {https://doi.org/10.1017/jsl.2020.12}
  {\path{doi:10.1017/jsl.2020.12}}.

\bibitem{km:2016}
Leszek~Aleksander Ko{\l}odziejczyk and Henryk Michalewski.
\newblock How unprovable is {R}abin's decidability theorem?
\newblock In {\em Proceedings of the 31st Annual ACM/IEEE Symposium on Logic in
  Computer Science}, pages 788--797, 2016.

\bibitem{Kozen-Silva-2012}
Dexter Kozen and Alexandra Silva.
\newblock Left-handed completeness.
\newblock In Wolfram Kahl and Timothy~G. Griffin, editors, {\em Relational and
  Algebraic Methods in Computer Science}, pages 162--178, Berlin, Heidelberg,
  2012. Springer Berlin Heidelberg.

\bibitem{cantorDetBCO}
St{\'e}phane Le~Roux and Arno Pauly.
\newblock Weihrauch degrees of finding equilibria in sequential games.
\newblock In {\em Conference on Computability in Europe}, pages 246--257.
  Springer, 2015.

\bibitem{mac2013categories}
Saunders Mac~Lane.
\newblock {\em Categories for the working mathematician}, volume~5.
\newblock Springer Science \& Business Media, 2013.

\bibitem{MVRamsey}
Alberto Marcone and Manlio Valenti.
\newblock The open and clopen {R}amsey theorems in the {W}eihrauch lattice.
\newblock {\em The Journal of Symbolic Logic}, 86(1):316--351, March 2021.
\newblock \href {https://doi.org/10.1017/jsl.2021.10}
  {\path{doi:10.1017/jsl.2021.10}}.

\bibitem{maschio2025}
Samuele Maschio and Davide Trotta.
\newblock A topos for extended {W}eihrauch degrees, 2025.
\newblock URL: \url{https://arxiv.org/abs/2505.08697}, \href
  {https://arxiv.org/abs/2505.08697} {\path{arXiv:2505.08697}}.

\bibitem{moerdijk2000wellfounded}
Ieke Moerdijk and Erik Palmgren.
\newblock Wellfounded trees in categories.
\newblock {\em Annals of Pure and Applied Logic}, 104(1-3):189--218, 2000.

\bibitem{mollerfeld02phd}
M.~M{\"o}llerfeld.
\newblock {\em {Generalized Inductive Definitions - The $\mu$-calculus and
  $\Pi^1_2$-comprehension}}.
\newblock PhD thesis, Westf{\"a}lischen Wilhelms-Universit{\"a}t M{\"u}nster,
  2002.
\newblock Available at
  \url{https://miami.uni-muenster.de/Record/9dfa74b6-186b-4e95-a51f-9965d7e1e508}.

\bibitem{glehnmoss18}
Sean~K. Moss and Tamara von Glehn.
\newblock {D}ialectica models of type theory.
\newblock In {\em Proceedings of the 33rd Annual {ACM/IEEE} Symposium on Logic
  in Computer Science, {LICS} 2018, Oxford, UK, July 09-12, 2018}, pages
  739--748, 2018.

\bibitem{topol-comput-neumann-pauly}
Eike Neumann and Arno Pauly.
\newblock A topological view on algebraic computation models.
\newblock {\em J. Complex.}, 44:1--22, 2018.
\newblock \href {https://doi.org/10.1016/j.jco.2017.08.003}
  {\path{doi:10.1016/j.jco.2017.08.003}}.

\bibitem{theoryWeiTimes}
Eike Neumann, Arno Pauly, and C{\'{e}}cilia Pradic.
\newblock The equational theory of the {W}eihrauch lattice with multiplication.
\newblock {\em CoRR}, abs/2403.13975, 2024.
\newblock \href {https://arxiv.org/abs/2403.13975} {\path{arXiv:2403.13975}},
  \href {https://doi.org/10.48550/ARXIV.2403.13975}
  {\path{doi:10.48550/ARXIV.2403.13975}}.

\bibitem{polybook}
Nelson Niu and David~I. Spivak.
\newblock Polynomial functors: A mathematical theory of interaction, 2024.
\newblock URL: \url{https://arxiv.org/abs/2312.00990}, \href
  {https://arxiv.org/abs/2312.00990} {\path{arXiv:2312.00990}}.

\bibitem{niwinski2021guidable}
Damian Niwi{\'n}ski and Micha{\l} Skrzypczak.
\newblock On guidable index of tree automata.
\newblock In {\em 46th International Symposium on Mathematical Foundations of
  Computer Science (MFCS 2021)}, pages 81--1. Schloss Dagstuhl--Leibniz-Zentrum
  f{\"u}r Informatik, 2021.

\bibitem{PY22}
Leonardo Pacheco and Keita Yokoyama.
\newblock Determinacy and reflection principles in second-order arithmetic,
  2023.
\newblock URL: \url{https://arxiv.org/abs/2209.04082}, \href
  {https://arxiv.org/abs/2209.04082} {\path{arXiv:2209.04082}}.

\bibitem{paulycountableordinals}
Arno Pauly.
\newblock Computability on the space of countable ordinals, 2017.
\newblock URL: \url{https://arxiv.org/abs/1501.00386}, \href
  {https://arxiv.org/abs/1501.00386} {\path{arXiv:1501.00386}}.

\bibitem{PricePradic25}
C{\'e}cilia Pradic and Ian Price.
\newblock Weihrauch problems as containers.
\newblock In Arnold Beckmann, Isabel Oitavem, and Florin Manea, editors, {\em
  Crossroads of Computability and Logic: Insights, Inspirations, and
  Innovations}, pages 395--409, Cham, 2025. Springer Nature Switzerland.
\newblock URL: \url{https://arxiv.org/abs/2501.17250}.

\bibitem{pradic25}
Cécilia Pradic.
\newblock The equational theory of the {W}eihrauch lattice with (iterated)
  composition, 2025.
\newblock URL: \url{https://arxiv.org/abs/2408.14999}, \href
  {https://arxiv.org/abs/2408.14999} {\path{arXiv:2408.14999}}.

\bibitem{RSSmod}
Egbert Rijke, Michael Shulman, and Bas Spitters.
\newblock Modalities in homotopy type theory.
\newblock {\em Log. Methods Comput. Sci.}, 16(1), 2020.
\newblock \href {https://doi.org/10.23638/LMCS-16(1:2)2020}
  {\path{doi:10.23638/LMCS-16(1:2)2020}}.

\bibitem{rogers1987}
Hartley Rogers.
\newblock {\em Theory of Recursive Functions and Effective Computability}.
\newblock MIT Press, 1987.

\bibitem{santocanale2002mu}
Luigi Santocanale.
\newblock $\mu$-bicomplete categories and parity games.
\newblock {\em RAIRO-Theoretical Informatics and Applications}, 36(2):195--227,
  2002.

\bibitem{simpson}
Stephen~G. Simpson.
\newblock {\em Subsystems of second order arithmetic}.
\newblock Perspectives in Mathematical Logic. CUP, 1999.
\newblock \href {https://doi.org/10.1007/978-3-642-59971-2}
  {\path{doi:10.1007/978-3-642-59971-2}}.

\bibitem{streicherfibrations}
Thomas Streicher.
\newblock Fibered categories à la {J}ean {B}énabou, 2023.
\newblock \href {https://arxiv.org/abs/1801.02927} {\path{arXiv:1801.02927}}.

\bibitem{Swan24}
Andrew Swan.
\newblock Oracle modalities.
\newblock {\em CoRR}, abs/2406.05818, 2024.
\newblock URL: \url{https://arxiv.org/abs/2406.05818}.

\bibitem{tanaka90}
Kazuyuki Tanaka.
\newblock Weak axioms of determinacy and subsystems of analysis {I}:
  {$\Delta^0_2$} games.
\newblock {\em Mathematical Logic Quarterly}, 36(6):481--491, 1990.

\bibitem{Tanaka91}
Kazuyuki Tanaka.
\newblock Weak axioms of determinacy and subsystems of analysis {II}:
  {$\Sigma^0_2$} games.
\newblock {\em Ann. Pure Appl. Logic}, 52(1-2):181--193, 1991.
\newblock \href {https://doi.org/10.1016/0168-0072(91)90045-N}
  {\path{doi:10.1016/0168-0072(91)90045-N}}.

\bibitem{TVdP22}
Davide Trotta, Manlio Valenti, and Valeria de~Paiva.
\newblock Categorifying computable reducibilities.
\newblock {\em Logical Methods in Computer Science}, Volume 21, Issue 1, Feb
  2025.
\newblock \href {https://doi.org/10.46298/lmcs-21(1:15)2025}
  {\path{doi:10.46298/lmcs-21(1:15)2025}}.

\bibitem{hottbook}
The {Univalent Foundations Program}.
\newblock {\em Homotopy Type Theory: Univalent Foundations of Mathematics}.
\newblock \url{https://homotopytypetheory.org/book}, Institute for Advanced
  Study, 2013.

\bibitem{van2007non}
Benno van~den Berg and Federico De~Marchi.
\newblock Non-well-founded trees in categories.
\newblock {\em Annals of Pure and Applied Logic}, 146(1):40--59, 2007.

\bibitem{vanOosten}
Jaap Van~Oosten.
\newblock {\em Realizability: an introduction to its categorical side}, volume
  152.
\newblock Elsevier, 2008.

\bibitem{weber2015polynomials}
Mark Weber.
\newblock Polynomials in categories with pullbacks.
\newblock {\em Theory Appl. Categ}, 30(16):533--598, 2015.

\bibitem{westrick2020}
Linda Westrick.
\newblock A note on the diamond operator.
\newblock {\em Computability}, 10(2):107--110, 2021.
\newblock \href {https://doi.org/10.3233/COM-200295}
  {\path{doi:10.3233/COM-200295}}.

\bibitem{Yoshimura2}
K.~Yoshimura.
\newblock General treatment of non-standard realizabilites, 2016.
\newblock Unpublished.

\end{thebibliography}

\appendix

\section{Supplementary background material}

\subsection{On notations and strength}

\begin{convention}
In the sequel, we will prioritize the ``container'' terminology, but may also
use simultaneously the terminology ``polynomial functor''. This is because
we are interested in endofunctors over the category of containers that happen
to be polynomials; that they are representable themselves by containers is mostly
true, but only relevant insofar as they have nice properties.
When we talk of non-dependent containers $P \in \Container(\cC)$, we will by default regard
$P : A \to I$ as an internal family $(A_i)_{i : I}$ of $\cC$, so formally, as a morphism of $\cC$.
We generally use the letters $P, Q, R\ldots$ to denote non-dependent containers,
$I, J, K\ldots$ for objects that are meant to be shapes and $A,B,C$ for objects
that are meant to collect directions.
\end{convention}

All polynomial functors $\functorify{P}$ admit a natural transformation
$\mathrm{strength}_{X,Y} : \functorify{P}(X) \times Y \to \functorify{P}(X \times Y)$.
A natural transformation $\eta_X : \functorify{P}(X) \to \functorify{Q}(X)$ is
called \emph{strong} if the following diagram commutes
\[\begin{tikzcd}[column sep=huge]
	{\functorify{P}(X) \times Y} & {\functorify{P}(X \times Y)} \\
	{\functorify{Q}(X) \times Y} & {\functorify{Q}(X \times Y)}
	\arrow["{\mathrm{strength}_{X,Y}}", from=1-1, to=1-2]
	\arrow["{\eta_X \times \id_Y}"', from=1-1, to=2-1]
	\arrow["{\eta_{X \times Y}}", from=1-2, to=2-2]
	\arrow["{\mathrm{strength}_{X,Y}}"', from=2-1, to=2-2]
\end{tikzcd}\]
With this we can make formal sense of the equivalence between containers in elcccs
and polynomial endofunctors (\Cref{prop:contPolyEquiv}).

\subsection{On Weihrauch reducibility}

\mysubparagraph{Monoidal products over problems}
Here is a more detailed explanation of the monoidal products and their units
seen as problem transformers, assuming we have some containers $P : A \to I$ and $Q : B \to J$
around:
\begin{itemize}
\item $0$ is the problem which has no questions nor answers
\item $1$ is the problem with a single trivial question with no answers
\item $\unit$ is the problem with a singe trivial question with a single trivial
\item $P + Q$ is the problem whose questions $k : I + J$ are either a question to $P$ or a
question to $Q$, and which returns an answer to $k$
\item a question for $P \times Q$ is a pair of questions $(i, j) : I \times J$
and an answer is either an answer to $i$ or an answer to $j$ (as well as a tag telling us
which question is answered)
\item a question for $P \tensor Q$ is a pair of questions $(i, j) : I \times J$,
and an answer consists of an answer to $i$ and of an answer to $j$
\item a question for $P \star Q$ is a pair $(j, f)$ %
where $j$ is a $Q$-question and $f$ is a function turning a $Q$-answer to $j$ into a $P$-question.
An answer to $(j, f)$ is a pair $(b, a)$ %
of an answer $b$ to $j$ and an answer $a$ to $f(b)$. Intuitively, there is a
reduction $R \to P \star Q$ if and only if it is possible to solve $R$ using
an oracle call to $Q$ and then an oracle call to $P$.
\end{itemize}

\mysubparagraph{The nice structure of $\ReprSp$} Without going into details, the reason this category is cartesian closed
and has nice structure is because it can interpret the full $\lambda$-calculus.
This is because one may encode any partial computable maps
$f : \Baire \partto \bbN$ (and a fortiori $\Baire \partto \Baire$ via
$\Baire \cong \bbN \times \Baire$)
as a (potentially infinite) binary tree $t$ whose leaves are labelled by elements
of $\bbN$, such that that $f(p)$ is the label of the leaf the infinite path
leads to in $t$ (if $p$ does not lead to a leaf, $f$ is undefined on $P$).
Given a standard embedding of such trees in $\Cantor$ as an oracle, a type 2 
Turing machine can compute $f$.

One thing that makes type 2 computability sometimes nicer to work with than the analogue for type 1 computability is that $\Baire$ space, its subspaces and continuous functions between them can be faithfully represented in this setting. Indeed, any partial continuous map $\Baire \to \Baire$ can be implemented by a type 2
Turing machine with access to a suitable oracle. In contrast, $\Cantor$ can
only have computable points in the type 1 setting; an unfortunate side
effect of that is that it ends up being not compact due to the
existence of Kleene trees~\cite{vanOosten}.
Of course, issues of being able to faithfully represent classical objects
do show up at higher orders.

Aside from the structure
afforded by elcccs and fixpoints, the only tool we will employ to introduce new represented
space will essentially be taking subspace of known spaces (e.g., $\Baire$, spaces of finite sequences, infinite sequences of elements of $\Cantor$,
$\bbN$ and $\Ninfty$ can all be regarded as subspaces of $\Cantor$ in a standard
way).

\mysubparagraph{Why we consider $\rpReprSp$}
The following example captures why computable analysts do step out of
$\ReprSp$ when looking at the strength of problems inspired by computable
topology.

\begin{example}
The Sierpi\'nski space $\Sierp$ can be computed in $\ReprSp$ as the coequalizer of the
identity and successor maps $\id_{\Ninfty}, \underline{\Succ} : \Ninfty \to \Ninfty$. As usual, call $\bot_\Sierp$ and $\top_\Sierp$
the equivalence classes of $\infty$ and the zero element.
Then we have a morphism from $\LPO$ to the container $\mathrm{BoolSierp} : 2 \to \Sierp$ defined by
$\mathrm{BoolSierp}(0) = \bot_\Sierp$ and $\mathrm{BoolSierp}(1) = \top_\Sierp$ but not the
other way around. However, computable analysts often want to consider those
problems equivalent.
\end{example}

\mysubparagraph{More on answerable parts}
While $\mkAnswerable$ is not functorial, we have the following interactions
between $\mkAnswerable$ and the monoidal products.

\begin{proposition}
For any containers $P$ and $Q$ over $\ReprSp$, we have
\[
\begin{array}{lcl}
  \mkAnswerable(0) &\cong& \mkAnswerable(1)~~ \cong~~ 0
\\
  \mkAnswerable(\unit) & \cong & \unit
\\
\mkAnswerable(P \times Q) &\ge&
\mkAnswerable(P) \times \mkAnswerable(Q)
\\
\mkAnswerable(P + Q) &\cong& \mkAnswerable(P) + \mkAnswerable(Q)
\\
\mkAnswerable(P \tensor Q) &\cong&
\mkAnswerable(P) \tensor \mkAnswerable(Q)
\\
\mkAnswerable(P \star Q) &\cong&
\mkAnswerable(P) \star \mkAnswerable(Q)
\end{array}
\]
\end{proposition}

In the sequel, we will repeatedly exploit one generalization of this to the
infinite parallelization:
\[ \mkAnswerable(\widehat{P}) \qquad \cong \qquad \widehat{\mkAnswerable(P)}\]

\subsection{On $\mu$-expressions and automata}

$\mu$-expressions correspond to certain regular tree languages
in the sense of automata theory. In order to make that formal and accomodate
expressions with free variables, we will adopt a definition of tree automata
allowing to recognize trees with holes in the spirit of~\cite{FMS16}.
For any two disjoint sets $\Gamma$ and $\chi$, 
define $\Tree(\Gamma, \chi)$ to be the set of 
partial maps
$\ell : 2^* \partto \Gamma \uplus \chi$ such that $\ell(\varepsilon)$ is defined,
whenever $\ell(w) \in \Gamma$, $\ell(w0)$ and $\ell(w1)$ are both defined, and
$\ell(wi)$ is not defined when $\ell(w) \in \chi$.
Paths $w$ with $\ell(w) \in \chi$ are called \emph{exits} of $\ell$. For $x \in \chi$,
let us define $\exits_x(\ell) \eqdef \ell^{-1}(x)$ and $\exits(\ell) \eqdef
\bigcup\limits_{x\in \chi} \exits(x)$.
For short, write $\Tree(\Gamma)$ for $\Tree(\Gamma, \emptyset)$ in the sequel.

\begin{definition}
\label{def:NFA}
A non-deterministic infinite tree automaton with priorities $(i, k)$ and
set of exits $\chi$ over a
finite alphabet $\Gamma$ is a
tuple $\cA = (Q, I, \Delta, F, c)$ where:
\begin{itemize}
\item $Q$ is a finite set of states
\item $I \subseteq Q$ is a set of initial states
\item $\Delta : \Gamma \times Q \times 2 \to \powerset(Q)$ is a transition
function taking as input a letter $a \in \Gamma$, a direction $i \in 2$, a
state $Q$ and outputs a set of states
\item $F \subseteq \chi \times Q$ is a set of legal final configurations
\item $c : Q \to \{i, \ldots, k\}$ assigns to each state a priority.
\end{itemize}
\end{definition}
We say that $\ell \in \Tree(\Gamma, \chi)$ is
recognized by $\cA$ if we have a map $\mathrm{run} : \dom(\ell) \to Q$ such that
\begin{itemize}
\item $\mathrm{run}(\varepsilon) \in I$
\item for every $w \in 2^*$ and $i \in 2$ such that $wi \in \dom(\ell)$,
$\mathrm{run}(w i) \in \Delta(\ell(w), \mathrm{run}(w), i)$
\item for any maximal $w \in \dom(\ell)$, $(\ell(w), \mathrm{run}(w)) \in F$
\item for any path $p \in \Cantor$ through $\dom(\ell)$, $\limsup\limits_{n \to + \infty} c(\mathrm{run}(p[n]))$ is
even.
\end{itemize}

Assuming $\Gamma$ and $\chi$ are finite, $\Tree(\Gamma, \chi)$ can be regarded
as a subspace of $\Baire$ in a standard way; we thus assume that it is a
represented space. We write $\cL(\cA)$ for the subspace of $\Tree(\Gamma,\chi)$
recognized by the automaton $\cA$. When $L = \cL(\cA)$ for some automaton,
we call $L$ a \emph{regular tree language}. Regular tree languages are
effectively closed under boolean operations, projections and there are
effective procedures to decide emptiness and compute witnesses of non-emptiness
(see e.g.~\cite[\S 14.3]{Toolbox} for an account of the theory).

Now, for any $L \subseteq \Tree(\Gamma, \chi)$, we can associate a one-dimensional
$n$-ary container $\containerify{L}$:
\[\begin{tikzcd}
	\chi && {\sum\limits_{\ell \in L} \exits(\ell)} && L
    \arrow["{\text{evaluation}}"', from=1-3, to=1-1]
	\arrow["{\text{first projection}}", from=1-3, to=1-5]
\end{tikzcd}\]
The induced polynomial functor $\functorify{L}$ applied to a tuple $\vec X$
yields the set of trees in $L$ where leaves are labelled by elements of an
$X_i$.

\begin{proposition}
\label{prop:muExprAut}
For any expression $E \in \muExpr(\vec X)$, we can compute a regular tree language
$M_E \subseteq \Tree(2, \{\vec X\})$ such that $\containerify{E} \cong \containerify{M_E}$
in represented spaces.
\end{proposition}

We shall prove a generalization of~\Cref{prop:muExprAut} in detail later
(\Cref{subsec:zetaExprAut}), so we will only discuss a couple further examples for now
to guide the intuition. The basic idea of the construction is that the alphabet
$2$ delimits a subtree of the full binary trees; so we will have
that for any $\ell \in M_E$ and $w, u \in 2^*$ that $\ell(w) = 0$
implies $\ell(wu) = 0$. In the examples, we will stick with the languages obtained
via the proof of~\Cref{thm:zetaExprAut}; those are not always the most trim encoding
of the relevant $\mu$-expressions.

First, a trivial example.

\begin{example}
  \label{ex:muExprTrivial}
We have $M_{\mu X. X} = \emptyset$ and $M_{\nu X. X}$ is a singleton containing
the labelling $\ell : 2^* \to 2$ such that $\ell(w) = 0$ if and only if $w$
contains a $1$.
\end{example}

Then following example shows how free variables are treated and can be
used as a building block in formally treating~\Cref{ex:muExprWF}.

\begin{example}
  \label{ex:muExprStream}
We have that $M_{\nu Y. X \times Y} \subseteq \Tree(2, \{X\})$ is a singleton $\{\ell_{\mathsf{Stream}(X)}\}$
with $\ell_{\mathsf{Stream}(X)}(w0) = X$, $\ell_{\mathsf{Stream}(X)}(w1) = \ell_{\mathsf{Stream}(X)}(w) = 1$ and $\ell_{\mathsf{Stream}(X)}(w01u) = 0$ for any
$w \in (01)^*1$ and $u \in 2^*$. Its denotation $\functorify{M_{\nu Y. X \times Y}} \cong \functorify{\nu Y. X \times Y}$
is the functor $X \mapsto X^\bbN$.
\end{example}

\section{Supplementary material for Section~\ref{sec:ex}}

\subsection{Fibredness of some functors}

\begin{lemma}\label{lem:coproduct-pullbacks-commute}
Let $\cC$ be an extensive category. If $P$ is the pullback of $f : X
\to I$ along $g : Y \to I$ and $Q$ is the pullback of $h : Z \to J$
along $j : W \to J$, then $P+Q$ is the pullback of $f+h$ along $g+j$
\end{lemma}
\begin{proof}
Since $\coproj_1 \circ g = (g + j) \circ \coproj_1$, pulling $f+h$
along either gives an object $P^\prime$ isomorphic to $P$. Similarly, pulling back along $(g + j) \circ \coproj_2$ gives an object
$Q^\prime$ isomorphic to $Q$. By extensivity~\cite[Proposition
  2.2]{carboni93extensive}, $X+Y \cong P^\prime + Q^\prime \cong P + Q$.
\end{proof}

\stufffibred*

\begin{proof}
These properties follow from basic properties of polynomial functors and
extensivity:
\begin{itemize}
\item Constant functors are polynomial functors~\cite[Example 1.6
  (ii)]{GKpoly} and are cartesian since all container morphisms are
  mapped to the identity morphism, which is cartesian.
\item $+$ and $\tensor$ are clearly polynomial functors. $\tensor$
  preserves cartesian morphisms, since cartesian morphisms in
  $\Container(\cC)$ are essentially pullback squares, and both products
  and pullbacks are limits so they commute up to isomorphism. For $+$,
  we use \Cref{lem:coproduct-pullbacks-commute}.
\item Recall that $P \times Q = [P \times \id, \id \times Q]$.
  As fibred polynomials are closed under composition,
  it remains to check that the cotupling of polynomials is a fibred
  polynomial functor, but this again holds by
  extensivity~\cite[Proposition 2.2]{carboni93extensive}.
\end{itemize}
\end{proof}

\trianglefibred*

\begin{proof}
Given containers $P : A \to I$, $Q : B \to J$, then the composite $Q \star P$ is
given by the obvious map
\[\sum_{i : I} \sum_{a : A_i} \sum_{f : A_i \to J}
B_{f(a)} \qquad \longto \qquad \sum_{i : I} \left(A_i \to B\right)\]
(assuming we set $A_i = P^{-1}(i)$ and $B_j = Q^{-1}(j)$).
The base of the fibred polynomial functor $J \mapsto \sum\limits_{i : I} J^{A_i}$ is clearly
polynomial, as are the functors $ (B_j)_{j : J} \mapsto \left(\sum\limits_{a : A_i} B_{f(a)}\right)_{(i, f) : \sum\limits_{i : I} J^{A_i}}$.

That the functor is fibred follows since if $\phi : P \to P^\prime, \psi : Q \to Q^\prime$
are cartesian, then so is $\psi \star \phi$ \cite[Proposition 6.88]{polybook},
and $\id$ is an isomorphism, hence cartesian.
\end{proof}

\subsection{Fixpoints}

\algebraMuTwo*
\begin{proof}
Let us reason informally in the internal language of elcccs with $\cW$-types, i.e.,
EMLTT augmented with $\cW$-types. Within EMLTT, for any families $A$ and $B$,
we can write a type $\ttIso{A}{B}$ that is inhabited whenever $A$ and $B$
are isomorphic, a fortiori with $\prod_{x : \mu F_0} \ttIso{S_0(x)}{S_1(x)}$
holding if and only if $S_0 \cong S_1$.
So as to exploit the internal language,
let us assume that
\[ F_0(X) = \prod_{i :  I} X^{A_i} \enspace \text{and} \enspace
F((Y_x)_{x  :  X}) = \left(\sum_{j  :  J_{i,f}} \prod_{x  :  X} B_{i,f,j,x} \to Y_x\right)_{(i, f)  :  F_0(X)}\]
for a suitable type $I$ and families $(A_i)_{i : I}$, $(J_{i,f})_{(i,f) : F_0(X)}$ and $(B_{i,f,j,x})_{j : J_{i,f}, x : X}$.
Recall $\mu F_0$ is a $\cW$-type with constructor $\alpha$, so
we may proceed to show $\prod_{x  :  \mu F_0} \ttIso{S_0(x)}{S_1(x)}$ is
inhabited by recursion over $x$.
For this, it suffices to build a function in
$\ttIso{S_0(\alpha(i, f))}{S_1(\alpha(i,f)))}$ assuming $i : I$, $f : A_i \to \mu F_0$
and that we inductively have an inhabitant of
\[\prod_{a : A_i} 
\ttIso{S_0(f(a))}{S_1(f(a))}\]
So let us do half of that and conclude by symmetry.
\[\begin{array}{lcll}
S_0(\alpha(i, f))
&\cong& 
F_{\mu F_0}(S_0)(i,f)
& \text{since $F_{\mu F_0}(S_0) \cong S_0 \circ \alpha$}
\\
&\cong&
\sum\limits_{j  :  J_i} \prod\limits_{x  :  \mu F_0} B_{i,f, j, x} \to S_0(x)
& \text{by definition} \\
&\cong&
\sum\limits_{j  :  J_i} \prod\limits_{x  :  \mu F_0} B_{i,f, j, x} \to S_1(x)
& \text{by functoriality of $\Pi$, $\Sigma$ and $S_0(x) \cong S_1(x)$} \\
\end{array}
\]
\end{proof}

\fpexistmulti*

\begin{proof}
The following proof is only for $\mu$, similar arguments holds for
$\nu$ and $\zeta$.

\mysubparagraph{Functoriality}
Since $F$ is fibred polynomial, so is $F(f_0, \dots, f_{k-1}, -)$ for any
objects $f_i : X_i \to I_i$ in $\Container(\cC)$ (see, e.g., \Cref{lem:stufffibred}),
so it is a well-defined map on objects.

Suppose we have objects $g_i : Y_i \to J_i $ and morphisms
$(\varphi_i, \psi_i)$ in $\Container(\cC)$ and let
\[
\begin{array}{lclcl}
  \mu F(f_0, \dots, f_{k-1}, -) &:& \nu F_{\mu F_0(I_0, \dots, I_{k-1}, -)}(f_0, \dots, f_{k-1}, -) &\longto& \mu F_0(I_0, \dots, I_{k-1}, -)\\
  \mu F(g_0, \dots, g_{k-1}, -) &:& \nu F_{\mu F_0(J_0, \dots, J_{k-1},
  -)}(g_0, \dots, g_{k-1}, -) &\longto& \mu F_0(J_0, \dots, J_{k-1}, -)
\end{array}
\]
be
the respective fixpoints.
On the bases, we have two $F_0(I_0, \dots, I_{k-1}, -)$-algebra
morphisms: the initial one, and one obtained from the initial
$F_0(J_0, \dots, J_{k-1}, -)$-algebra by composing with
$F_0\left(\varphi_0, \dots, \varphi_{k-1}, \psi_{k-1}, \id\right)$.
Thus we obtain a unique algebra morphism, which we take to be the forward map of
$\mu F\left(\left(\varphi_0, \psi_0\right), \dots, \left(\varphi_{k-1}, \psi_{k-1}\right)\right)$.
Similarly, we obtain the backwards map by doing the same thing on the
domains, but using that we have terminal coalgebras and the $\psi_i$. Functoriality of
this construction is guaranteed by the uniqueness of the algebra and
coalgebra maps.

\mysubparagraph{Fibred polynomial}
It remains to see that this construction gives a fibred functor which
is polynomial on the base and the fibres.

To check that we have a morphism of fibrations, we define the base functor
by
\[(\mu F)_0(I_0, \dots, I_{k-1}) \eqdef \mu F_0(I_0, \dots, I_{k-1}, -)\]
Then $(\mu F)_0 \circ \shape = \shape \circ \mu F$ holds by construction.
$(\mu F)_0$ is a ``fixpoint functor'' of the polynomial functor $F_0$ and so
is itself polynomial~\cite[Theorem 14]{gambino_wellfounded_2004}. To see that
$\mu F$ preserves cartesian morphisms, recall that cartesian morphisms
$(\varphi, \psi)$ are those for which $\psi$ is an isomorphism. In
this case, the unique coalgebra morphism is an isomorphism, since we
can do the same construction in the opposite direction.

Similarly, $\mu F(X_0, \dots, X_{k-1}, -) = \nu F_{(\mu F)_0}(X_0, \dots, X_{k-1}, -)$
is a fixpoint functor, and hence a polynomial functor.
\end{proof}

\section{Supplementary material for Section~\ref{sec:operators}}
\subsection{Parallelizations}

\begin{proposition}
The $\nu$-fixpoint and $\zeta$-fixpoint of $X \mapsto P \times X$ are respectively (isomorphic to)
the following containers:
\[
  \sum_{s : I^\bbN} \sum_{n : \bbN} A_{s(n)} \quad \xrightarrow{\pi_1} \quad I^\bbN
\]
and
\[\sum_{s : I^\bbN} \sum_{x : \Ninfty} \prod_{n : \bbN} \left(x = \underline{n} \to A_{s(n)}\right)\quad \xrightarrow{\pi_1}\quad I^\bbN\]
\end{proposition}
\begin{proof}
We already argued that the shapes were $I^\bbN$. Now following~\Cref{thm:fp-exist}
let us work in type theory, assuming that $P : A \to I$ represents the internal family $(A_i)_{i : \bbN}$.
In this proof, let us write $\Succ : \bbN \to \bbN$ for the successor map. 
It suffices to show that $\left(\sum_{n : \bbN} A_{s(n)}\right)_{s : I^\bbN}$ and
$\left(\sum_{x : \Ninfty} \prod_{n : \bbN} \left(x = \underline{n} \to A_{s(n)}\right)\right)_{s : I^\bbN}$
are respectively carriers for an initial algebra and a terminal coalgebra for the
functor $F : \slice{\cC}{I^\bbN} \to \slice{\cC}{I^\bbN}$ taking $(B_s)_{s : I^\bbN}$ to $(A_{s(0)} + B_{s \circ \Succ})_{s : I^\bbN}$.

Now $\left(\sum_{n : \bbN} A_{s(n)}\right)_{s : I^\bbN}$ is the carrier of
an $F$-algebra whose underlying map is built using the following isomorphism
that exists fibrewise in $s$ by extensivity of $\cC$:
\[A_{s(0)} + \sum_{n : \bbN} A_{s(n+1)} \qquad \cong 
\qquad 
\sum_{n : \bbN} A_{s(n)}\]
To see that it is initial, it suffices to show that for any $F$-algebra
$\alpha_s : A_{s(0)} + B_{s \circ \Succ} \to B_s$, there \emph{purely}\footnote{In the
type theoretic sense~\cite{hottbook}, meaning we can explicitly construct
the witness in function of the parameters in type theory.} exists a unique
family of maps $f_s : \sum_{n : \bbN} A_{s(n)} ~~ \to ~~ B_s$ such that, for every $n : \bbN$,
\[ f_s(0, a) = \alpha_s(\inl(a)) \qquad \text{and} \qquad f_s(n + 1, a) = \alpha_{s}(\inr(f_{s}(n, a)))\]
This is easy to show by induction over $n : \bbN$.

Similarly, to check that
$\left(\sum_{x : \Ninfty} \prod_{n : \bbN} \left(x = \underline{n} \to A_{s(n)}\right)\right)_{s : I^\bbN}$
is the carrier of the obvious terminal $F$-coalgebra, we need to check that for any
$F$-coalgebra $\gamma_s : B_s \to A_{s(0)} + B_{s \circ \Succ}$, there purely exists
a unique $f_s : B_s \to \left(\sum_{x : \Ninfty} \prod_{n : \bbN} \left(x = \underline{n} \to A_{s(n)}\right)\right)$
such that, for any $b : B_s$:
\begin{itemize}
\item either $\gamma_s(b) = \inl(a)$ for some $a : A_{s(0)}$ and $f_s(b) = (\underline{0}, h)$
  with $h : \prod_{n : \bbN} \left( \underline{0} = \underline{n} \to A_{s(n)}\right)$
constant and equal to $a$.
\item or $\gamma_s(b) = \inr(b')$ for some $b' : B_{s \circ \Succ}$,
$f_{s \circ \Succ}(b') = (x, h \circ \Succ)$ and
$f_s(b) = (\underline{\Succ}(x), h)$
for some $x : \Ninfty$ and $h : \prod_{n : \bbN} \left( \underline{\Succ}(x) = \underline{n} \to A_{s(n)}\right)$.
\end{itemize}
This determines the first component $\ell_s \eqdef \pi_1 \circ f_s : B_s \to \Ninfty$ as we
have
\[\gamma_s(b) = \inl(a) ~~ \Leftrightarrow ~~ \ell_s(b) = \underline{0}
  \quad\text{and}\quad
\gamma_s(b) = \inr(b') ~~ \Leftrightarrow ~~ \ell_s(b) = \underline{\Succ}(\ell_{s \circ \Succ}(b'))\]
which uniquely determines $\ell_s$ by corecursion over the conatural numbers.
Up to currying and swapping arguments, the second component $r_s$ of $f_s$ has type
\[ r_s : \prod_{n : \bbN} \prod_{b : B_s} \left(\ell_s(b) = \underline{n} \to A_{s(n)}\right)\]
and is inductively determined by
\[
\begin{array}{lcl !\enspace l}
  r_s(0)(b)(e) &=& a &\text{when $\gamma_s(b) = \inl(a)$}\\
  r_s(n+1)(b)(e) &=& r_s(n)(b')(\mathrm{Sinj}(e)) & \text{when $\gamma_s(b) = \inr(b')$}
\end{array}
\]
where $\mathrm{Sinj}$ is the obvious term of type $\underline{\Succ}(\ell_{s \circ \Succ}(b')) = \underline{\Succ}(\underline{n}) \to
\ell_{s \circ \Succ}(b') = \underline{n}
$. This can be improved to a recursive definition over $n : \bbN$ by pattern-matching
over $\gamma_s(b)$, whose remaining cases admit absurdity proofs due to $e$.

\end{proof}

Let us state the other straightforward results evoked in this subsection without proofs.

\begin{proposition}
The $\nu$-fixpoint and $\zeta$-fixpoint of $X \mapsto P \tensor X$ are respectively (isomorphic to)
the following containers:
\[
  0 \longto I^\bbN
\qquad
\text{and}
\qquad
A^\bbN \xrightarrow{P^\bbN} I^\bbN\]
In particular\footnote{For any $P$, $1 \star P$ interpreted as a problem is the problem
with the same questions as $P$, but with no answers whatsoever.} $\nu(X \mapsto P \tensor X) \cong 1 \star \widehat{P}$.
\end{proposition}

\begin{proposition}
$P^*$ is isomorphic to
$\sum\limits_{n : \mathbb{N}} \underbrace{P \tensor \ldots \tensor P}_{\text{$n$ times}}$.
On the other hand, $\nu(X \mapsto \unit + P \tensor X)$
and $\zeta(X \mapsto \unit + P \tensor X)$ are respectively (isomorphic to) 
\[
  \begin{array}{lcl}
  \sum\limits_{n : \bN} A^n &\longto& \sum\limits_{m : \Ninfty} \prod\limits_{\substack{k : \bN\\k < m}} I^k \\
  (n, a) &\longmapsto& (\underline{n}, P \circ a)
  \end{array}\]
and
\[  \begin{array}{lcl}
  \sum\limits_{m : \Ninfty} \prod\limits_{\substack{k : \bN\\k < m}} A^k 
  &\longto&
  \sum\limits_{m : \Ninfty} \prod\limits_{\substack{k : \bN\\k < m}} I^k\\ 
  (x, a) &\longmapsto& (x, P \circ a)
  \end{array}
\]
\end{proposition}

\subsection{Iterations}

\boundedDiamond*
\begin{proof}
It is straightforward to define an algebra map $s : \unit + S \star P \to S$ by using that
$\underbrace{P \star \ldots \star P}_{\text{$0$ times}} \cong I$ and
that $\star$ distributes over $+$.
\begin{equation*}
  \unit + (\sum\limits_{n : \mathbb{N}} \underbrace{P \star \ldots \star P}_{\text{$n$ times}})
  \star P
  \cong \unit + \sum\limits_{n : \mathbb{N}} \underbrace{P \star
      \ldots \star P}_{\text{$n+1$ times}}
  \cong \sum\limits_{n : \mathbb{N}} \underbrace{P \star \ldots \star P}_{\text{$n$ times}}
\end{equation*}

Given another algebra $(A, a : \unit + A \star P \to A)$, we can
define a family of maps $s_i : S_i \to A$ by induction, where
$S_i = \sum\limits_{n=0}^i \underbrace{P \star \ldots \star P}_{\text{$n$ times}}$.
First, set $s_0 : S_0 \cong \unit \xrightarrow{\incopr_1} \unit + A
\star P \xrightarrow{a} A$. Then we inductively define
$s_{i+1} : S_{i+1} \cong \unit + (S_i \star P) \xrightarrow{I + s_i \star P} I + A
\star P \xrightarrow{a} A$. The unique algebra morphism $(S, s) \to (A, a)$ is
given by the coproduct of $s_0$ and $s_i \circ \incopr_i :
\underbrace{P \star \ldots \star P}_{\text{$i$ times}} \to A$ for all
$i > 0$, and hence $(S, s)$ is an initial algebra for the functor.
\end{proof}

\kleeneinduction*
\begin{proof}
Let us call $F$ the functor $X \mapsto Q + X \star P$
for some containers $P : A \to I$ and $Q : B \to J$.
In the rest of the proof, we will work in Martin-L\"of type theory with $\cW$-types;
in this context, let us call $(A_i)_{i : I}$ and $(B_j)_{j : J}$ the respective
type families corresponding to $P$ and $Q$.
We take the following
isomorphism to be the $F$-algebra structure on $Q \star P^\diamond$ that we claim
to be initial:
\[
(\inAlg^\diamond, \outCoAlg^\diamond) : F(Q \star P^\diamond) =
Q + Q \star P^\diamond \star P
\cong
Q \star (\unit + P^\diamond \star P) \cong Q \star P^\diamond\]

Suppose we have an $F$-algebra $\left(R , (\varphi, \psi)\right)$
for some container $R : C \to K$. In type-theoretic notations, we have
\[ \varphi = [\varphi_1, \varphi_2] : J + \sum_{i : I} K^{A_i} ~~ \longto ~~ R\]
Let us also write $\left(\psi_j : C_{\varphi_1(j)} \to B_j\right)_{j : J}$ and
$(\psi_{i,f} : C_{\varphi_2(i,f)} \to \sum\limits_{a : A_i} C_{f(a)})_{(i, f) : \sum_{i : I} K^{A_i}}$ the
families of maps induced by $\psi$.

We will explicitly define the unique algebra morphism
$(\varphi^\diamond, \psi^\diamond) : Q \star P^\diamond \to R$
by constructing the forward and backward maps $\varphi^\diamond$ and $\psi^\diamond$
in type theory with $\cW$-types.
For notational convenience, let us write $I^\diamond$ for $\shape(P^\diamond)$
and $(A^\diamond_i)_{i : I^\diamond}$ for the type family corresponding to $P^\diamond$;
recall that, by~\Cref{thm:fp-exist}, we have maps
$\mathsf{leaf} : 1 \to I^\diamond$ and
$\mathsf{node} : \sum\limits_{i : I} (A_i \to I^\diamond) \to I^\diamond$
such that $(I^\diamond, [\mathsf{leaf}, \mathsf{node}])$ is the initial algebra of
the endofunctor $X \mapsto \unit + X \star P$.
In what follows, let us bear in mind that we have
$ \shape(Q \star P^\diamond) \cong \sum\limits_{i : I^\diamond} (A^\diamond_i \to J)$
and that the direction corresponding to $(i, f) : \shape(Q \star P^\diamond)$
are $\sum\limits_{a : A^\diamond_i} B_{f(a)}$.

We define the maps by well-founded recursion on the first component of the
argument $i : I^\diamond$ for both the forward and the backwards map; up to
dependent currying, this is a standard dependent recursion.
For the forward map, we have

\[\begin{array}{rlcl}
  \varphi^\diamond :
  &
\sum\limits_{i : I^\diamond}(A^\diamond_i \to J) & \longto & K
\\
         & (\mathsf{leaf}(\bullet), f) &\longmapsto& \varphi_1(f(\bullet)) \\
         & (\mathsf{node}(i, g), f) &\longmapsto& \varphi_2\left(i, \lambda a.\;\varphi^{\diamond}\left(g(a),\lambda b.\; f(a, b)\right)\right) \\
\end{array}
\]
For the backward map, we need to define a family
\[\begin{array}{rlcl}
  \psi^\diamond_{(i, f)} : & C_{\varphi^\diamond(i, f)} &
  \longto
  & 
  \sum\limits_{a : A_i^\diamond} B_{f(a)} \\
\end{array}
\]
We again proceed by well-founded recursion on $i$. The two cases are given by
\[
\begin{array}{rlcl}
  \psi^\diamond_{(\mathsf{leaf}(\bullet), f)} : & C_{\varphi_1(f(\bullet))} &
  \longto
  & 
  \sum\limits_{\bullet : 1} B_{f(\bullet)} \\
         & r &\longmapsto & (\bullet, \psi_{f(\bullet)}(r))\\
\end{array}
\]
and, up to associativity of $\sum$ in the codomain
and with  $\beta(g,f) \eqdef \lambda a.\;\varphi^{\diamond}\left(g(a),\lambda b.\; f(a, b)\right)$
:
\[
\begin{array}{rlcl}
  \psi^\diamond_{(\mathsf{node}(i, g), f)} : & C_{\varphi_2\left(i, \beta(g,f)\right)}
 &
  \longto
  & 
  \sum\limits_{a : A_i} \sum\limits_{\alpha : A^\diamond_{g(a)}} B_{f(a,\alpha)} \\
         & r
&\longmapsto&
\begin{array}{llcll}
\mathsf{let} & (a, r') & = & \psi_{
(i, \beta(g, f))
  }(r) & \mathsf{in} \\
\mathsf{let} & (b, w) & = &
\psi^\diamond_{(g(a), \lambda b. \; f(a,b))}(r')
& \mathsf{in} \\
\multicolumn{5}{l}{~~(a, b, w)} \\
\end{array}
\end{array}
\]
We leave checking that
the unique $F$-algebra morphism
from $({Q \star P^\diamond},\allowbreak {(\inAlg^\diamond,\outCoAlg^\diamond)})$
to $(R, (\varphi, \psi))$
is
$(\varphi^\diamond, \psi^\diamond)$
to the reader.
\end{proof}

\begin{lemma}\label{lem:compose-mul}
There is a natural transformation
$\left(Q \star P\right) \times R \to (Q \times R) \star P$.
\end{lemma}
\begin{proof}
Taking the same conventions as in the proof of~\Cref{thm:kleeneinduction},
the following forward and backward maps are definable in the internal language:
\[\begin{array}{rlcl}
  \varphi : & \sum\limits_{i : I} (A_i \to J) \times K
  &\longto& \sum\limits_{i : I } (A_i \to J \times K) \\
         & (i, g, k) &\longmapsto& (i, \lambda c. (g(c), k)) \\\\
  \psi_{(i, g, k)} : & \sum\limits_{a : A_i}(B_{g(a)} + C_k) &\longto&
  \sum\limits_{a : A_i}B_{g(a)}
  + C_k\\
         & (c, \coproj_1(b)) &\longmapsto& \coproj_1(c, b)\\
         & (c, \coproj_2(a)) &\longmapsto& \coproj_2(a) \\
\end{array}
\]
\end{proof}

\begin{lemma}%
\label{lem:generalised-diamond-induction}
Let $P$, $Q$, $R$ be containers. 
If there is a natural transformation $R \times (Q \star P) \to Q$
then there is a natural transformation $R \times (Q \star P^\diamond) \to Q$.
\end{lemma}
\begin{proof}
In what follows, let us denote the exponential in $\Container(\cC)$ by $\Wexp$.
First, construct a natural transformation
$R \times \left((R \Wexp Q) \star P\right) \to Q$
as follows:

\begin{align*}
R \times \left(\left(R \Wexp Q\right) \star P\right)
& \to
R \times R \times \left(\left(R \Wexp Q\right) \star P\right)
  && \text{via the diagonal map $R \to R \times R$} \\
& \to
R \times \left(\left(\left(R \Wexp Q\right) \times R\right) \star P\right)
  && \text{by \Cref{lem:compose-mul}} \\
& \to R \times \left(Q \star P\right)
  && \text{by using the evaluation map of $\Rightarrow$} \\
& \to Q && \text{by assumption.}
\end{align*}

After currying, we obtain a natural transformation
$(R \Wexp Q) \star P \to R \Wexp Q$,
so by~\Cref{thm:kleeneinduction}, we have
$R \times \left((R \Wexp Q) \star P^\diamond\right) \to Q$.
And then we can precompose by the natural transformation $Q \to R \Wexp Q$ to
conclude.
\end{proof}

\section{Supplementary material for Section~\ref{sec:ground}}

\subsection{Fixpoints and tensors alone trivialize}

We have discussed a number of examples of endofunctors on $\Container(\cC)$
that can helpfully be defined using our three fixpoint constructions.
Now can we use these fixpoints to construct interesting containers
from the point of view of Weihrauch reducibility? So far, we have only seen
that taking fixpoints of the identity functor yield the containers $0$, $\unit$
and $1$ (using $\mu$, $\zeta$ and $\nu$ respectively). 
Simply using these fixpoints without parameters can
only produce trivial containers if we decide to regard them as problems.

\begin{definition}
Call a fibred polynomial functor $F : \Container(\cC)^k \to \Container(\cC)$ answer-trivial
if for any tuple of objects $\vec{X}$ among $\{0, \unit, 1\}$, $F(\vec{X})$
is degree-equivalent to $0$, $\unit$ or $1$.
\end{definition}

\begin{figure}
\[
\begin{array}{c !\qquad c}
\begin{array}{|c||c|c|c|}
\hline
+ & 0 & \unit & 1 \\
\hline
\hline
0 & 0 & \unit & 1 \\
\hline
\unit & \unit & \unit & 1 \\
\hline
1 & 1 & 1 & 1 \\
\hline
\end{array}
&
\begin{array}{|c||c|c|c|}
\hline
\times & 0 & \unit & 1 \\
\hline
\hline
0 & 0 & 0 & 0 \\
\hline
\unit & 0 & \unit & \unit \\
\hline
1 & 0 & \unit & 1 \\
\hline
\end{array}
\\ \hfill \\
\begin{array}{|c||c|c|c|}
\hline
\tensor & 0 & \unit & 1 \\
\hline
\hline
0 & 0 & 0 & 0 \\
\hline
\unit & 0 & \unit & 1 \\
\hline
1 & 0 & 1 & 1 \\
\hline
\end{array}
&
\begin{array}{|c||c|c|c|}
\hline
\star & 0 & \unit & 1 \\
\hline
\hline
0 & 0 & 0 & 1 \\
\hline
\unit & 0 & \unit & 1 \\
\hline
1 & 0 & 1 & 1 \\
\hline
\end{array}
\end{array}
\]
\caption{Multiplication table for degrees given by the monoidal products and
their units.}
\label{fig:multiplicationTables}
\end{figure}

\begin{lemma}
  \label{lem:fpAnswerTrivial}
If $F : \Container(\cC) \to \Container(\cC)$ is answer-trivial, so is $\gamma F : \Container(\cC)$
for any $\gamma \in \{\mu, \nu, \zeta\}$.
\end{lemma}
\begin{proof}
Note that
we have $\gamma F \equiv F^2(\gamma F) = F(F(\gamma F))$,
$F^2(X) \equiv F^3(X)$ for any $X \in \{0, \unit, 1\}$ and that $F^2(X)$ is
degree-equivalent to $0$, $\unit$ or $1$. We are going to establish
that
\[
\mu F \quad \equiv \quad F^2(0)
\qquad \qquad
\zeta F \quad \equiv \quad F(\unit)
\qquad \text{and} \qquad
\nu F \quad \equiv \quad F^2(1)
\]
\begin{itemize}
\item If $\gamma = \mu$, then we have $F^2(0) \le F^2(\mu F) \equiv \mu F$. Conversely,
we can regard the map witnessing $F^3(0) \le F^2(0)$ as an $F$-algebra and get
$\mu F \le F^2(0)$.
\item If $\gamma = \nu$, we proceed dually with $F^2(1)$.
\item If $\gamma = \zeta$, let us make a case distinction on the degree of $F(\unit)$,
calling $F_0 : \cC \to \cC$ the functor such that $\shape \circ F = F_0 \circ \shape$:
\begin{itemize}
\item if $F(\unit) \equiv 0$, then $F_0(0) \cong 0$. Since $F_0$ is a polynomial
endofunctor and $\cC$ is extensive, it means that $F$ is constantly isomorphic to $0$,
and thus $\zeta F \cong 0$.
\item Otherwise $F(\unit) \equiv 1$ or $\unit$ and $F_0(1) \equiv 1$. We then have a map $\varphi : 1 \to F_0(1)$.
Since $(1, \varphi)$ is a $F_0$-coalgebra, it induces a coalgebra map
$\theta : 1 \to \nu F_0$.
Then, by diagram chasing, we have $\theta^*(\zeta F) \cong \varphi^*(F(\theta^*(\zeta F)))$, which allows us
to build a map $\theta^*(\zeta F) \to \varphi^*(F(\unit))$ and conclude that $F(\unit) \le \zeta F$.
If $F(\unit) \equiv 1$, we are done. Otherwise, if $F(\unit) \equiv \unit$,
$F(\unit)$ has a section (when regarded as a morphism of $\cC$). Since
$F$ is fibred, this entails that $F(\id_{\nu F_0})$ also has a section $s$.
We then have a $c^* \circ F_{\nu F_0}$-coalgebra with carrier the terminal object
of $\slice{\cC}{\nu F_0}$ for $F_{\nu F_0}$ taken as in~\Cref{def:fibPolFun}
and $c : \nu F_0 \to F_0(\nu F_0)$ the terminal $F_0$-coalgebra in $\cC$.
Since $\zeta F$ is the carrier of the terminal $c^* \circ F_{\nu F_0}$ coalgebra,
we have a coalgebra map $\id_{\nu F_0} \to \zeta F$ which is nothing more than
a section of $\zeta F$, witnessing $\zeta F \le \unit$.
\end{itemize}
\end{itemize}
\end{proof}

\begin{corollary}
\label{cor:zetaTrivial}
Any fibred polynomial functor $F$ built inductively from projections,
$\tensor$, $\times$, $+$, tupling, composition and $\gamma$-fixpoints for $\gamma \in \{\mu, \nu, \zeta\}$
is answer-trivial.
\end{corollary}
\begin{proof}
The proof is via a straightforward induction after considering the following facts:
\begin{itemize}
\item \Cref{lem:fpAnswerTrivial} also implies that for any
answer-trivial $F : \Container(\cC)^{k+1} \to \Container(\cC)$,
$\vec X \mapsto \gamma (Y \mapsto F(\vec X, Y))$ is also answer-trivial.
\item All of the monoidal products we consider are answer-trivial (see~\Cref{fig:multiplicationTables}).
\item The tupling and composition of answer-trivial functors is answer-trivial.
\item The identity is answer-trivial.
\end{itemize}
\end{proof}

Accordingly, if the domain of $F$ is $\Container(\cC)^0$, $\containerify{F}$
is degree-equivalent to either $0$, $1$ or $\unit$.
That being said, this does \emph{not} mean that $\containerify{F}$ is isomorphic
to one of these three options. We will thus try to deepen our understanding
of those containers obtainable by fixpoint constructions and tensors alone in
the next subsection, and see how the answerable part of those problems may be highly non-trivial.

\subsection{$\zeta$-expressions}
\label{subsec:zetaExprAut}

\begin{remark}
Since our monoidal products and fixpoints are computed layerwise,
we also have a syntactic map $\base$ such that $\shape \circ \ffunctorify{E} = \functorify{\base(E)} \circ \shape^k$
which turns a $\zeta$-expression with $k$ free variables into a $\mu$-expression
with the same free variables.
\[
\begin{array}{llcl}
\base :& \zetaExpr(\vec X) &\longto& \muExpr(\vec X) \\
        & X &\longmapsto & X \\
        & \mu X. E &\longmapsto& \mu X. \; \base(E) \\
        & \zeta X. E &\longmapsto& \nu X. \; \base(E) \\
        & \nu X. E &\longmapsto& \nu X. \; \base(E) \\
        & E + E' &\longmapsto& \base(E) + \base(E') \\
        & E \tensor E' &\longmapsto& \base(E) \times \base(E') \\
        & E \times E' &\longmapsto& \base(E) \times \base(E') \\
\end{array}
\]
\end{remark}

\zetaExprAut*

This theorem is proven by induction over $E$ and has~\Cref{prop:muExprAut}
as an immediate corollary. Rather than giving the inductive steps explicitly,
let us simply give the key lemmas to make the induction go through.
We first start with two composition lemmas, that
allows us to treat the cases of variables and to interpret the substitution of
two expressions.

\begin{remark}
Trees with holes $\ell \subseteq \Tree(\Gamma, \chi)$ are called regular if they
have only finitely many subtrees up to isomorphism.
Finite sets of regular trees are regular tree languages.
Note also that
automata with priorities $(i, k)$ can be trivially turned into equivalent automata
with priorities $(i + 2, k + 2)$.
Similarly, bumping priorities by $2$ in the alphabet does not affect the denotations of
game tree languages as containers.
\end{remark}

\begin{lemma}
\label{lem:zetaExprAutProj}
For all projections $\pi_x : \Container(\ReprSp)^\chi \to \Container(\ReprSp)$
(with $x \in \chi$), there is a finite game tree
$\ell_{\pi, \chi, x}$ with
$\ffunctorify{\{\ell_{\pi,\chi,x}\}} \cong \pi_x$.
\end{lemma}
\begin{proof}
We simply set $\ell_{\pi, \chi, x}(\varepsilon) = x$.
\end{proof}

\begin{lemma}
\label{lem:zetaExprAutSubst}
Given a regular game tree language $L \subseteq \Tree(\cG_{i, k}, \chi)$
and a family of regular game tree languages $(M_x)_{x \in \chi}$ with $M_x \subseteq \Tree(\cG_{i, k}, \chi')$,
there is a regular game tree language $ L \circ M \subseteq \Tree(\cG_{i,k}, \chi')$
such that $\ffunctorify{L \circ M} = \ffunctorify{L} \circ \langle M_x \rangle_{x \in \chi}$.
\end{lemma}
\begin{proof}[Proof idea]
Add transitions to an automaton recognizing $L$ to automata recognizing $M_x$
when it is possible to end a run on an exit labelled by $x$.
\end{proof}

\begin{lemma}
\label{lem:zetaExprAutConnectives}
There are finite regular game tree languages $L_\boxempty$ for $\boxempty \in \{\times, \tensor, +\}$
such that $\ffunctorify{L_\boxempty} \cong \boxempty$ in $\Container(\ReprSp)^2 \to \Container(\ReprSp)$.
\end{lemma}
\begin{proof}
The languages may be pictured as follows, where $i \in \bbN$ is an arbitrary priority:
\[
  \begin{array}{c}
L_{X + Y} \qquad=\qquad \left\{\quad
\begin{tikzcd}[column sep=tiny,row sep=small]
	& {(\Even, i)} \\
	X && \bot
	\arrow["0"', from=1-2, to=2-1]
	\arrow["1", from=1-2, to=2-3]
	\arrow["{{0,1}}"', shift left, from=2-3, to=2-3, loop, in=30, out=330, distance=5mm]
\end{tikzcd}
\qquad {\Huge ,} \qquad
\begin{tikzcd}[column sep=tiny,row sep=small]
	& {(\Even, i)} \\
	\bot && Y
	\arrow["0"', from=1-2, to=2-1]
	\arrow["1", from=1-2, to=2-3]
	\arrow["{{0,1}}", shift right, from=2-1, to=2-1, loop, in=150, out=210, distance=5mm]
\end{tikzcd}
\quad
\right\}
\\\\\\
L_{X \tensor Y} 
\quad = \quad
\left\{
\begin{tikzcd}[column sep=tiny,row sep=small]
	& {(\Odd, i)} \\
	X && Y
	\arrow["0"', from=1-2, to=2-1]
	\arrow["1", from=1-2, to=2-3]
\end{tikzcd}
\right\}
\qquad
L_{X \times Y} 
\quad = \quad
\left\{
\begin{tikzcd}[column sep=tiny,row sep=small]
	& {(\Even, i)} \\
	X && Y
	\arrow["0"', from=1-2, to=2-1]
	\arrow["1", from=1-2, to=2-3]
\end{tikzcd}
\right\}
\end{array}
\]
\end{proof}

\begin{lemma}
  Given a regular game tree language $L \subseteq \Tree(\cG_{i + 1, k}, \{\vec X,  Y\})$
and $\gamma \in \{\mu, \zeta, \nu\}$,
we can compute another game tree language $\gamma Y. L \subseteq \Tree(\cG_{i, k}, \{\vec X\})$
such that
\[\ffunctorify{\gamma Y. L} \qquad \cong \qquad \vec X \mapsto \gamma (Y \mapsto \ffunctorify{L}(\vec X, Y))\]
\end{lemma}
\begin{proof}
Let $\cA = (Q, I,\Delta, F, c)$ be an automaton recognizing $L$.
Let us call
$F'$ the intersection of $F$ with $\{\vec X\} \times Q$.
We proceed
by case anlaysis depending on the value of $\gamma$; we give the construction and leave
checking correctness to the reader.
\begin{itemize}
\item If $\gamma = \mu$, we take $\cA_\mu = (Q, I, \Delta_\mu, F', c_{\mu})$
with $\Delta_\mu$  and $c_\mu$ defined by
\[
\begin{array}{llcl}
& \Delta_\mu(a,q,d) &=& \Delta(a,q,d) \cup \{ r \in I \mid \exists q' \in \Delta(a,q,d). \; (Y, q') \in F\}\\\\
  \text{and} &
  c_\mu(q) &=& \left\{ \begin{array}{ll}
      \min\{ k \in \bbN \mid \text{$k$ is odd and $k \ge \max\limits_{q \in Q} c(q)$}\} & \text{if $q \in I$}
    \\
    c(q) & \text{otherwise}
\end{array} \right.
    \end{array}
\]
\item If $\gamma = \zeta$ or $\nu$, let $j_\gamma$ be the even element of $\{k, k + 1\}$
if $\gamma = \zeta$ or otherwise let it be the odd one.
We take $\cA_\gamma = (\{q_0, \bot\} \uplus Q, \{q_0\}, \Delta_\gamma, F', c_{\gamma})$
with $\Delta_\gamma$  and $c_\gamma$ defined by (assuming by convention that $\Delta((P, m), q, d) = \emptyset$ for $m \not\in \{i+1, \ldots, k\}$):
\[
\begin{array}{llclr}
& \Delta_\gamma((P, m),q,d) &=&
\left\{
\begin{array}{l}
\Delta((P,m),q,d) ~~~~ \cup\\
\{ q_0 \mid \exists q' \in \Delta((P,m),q,d). \; (Y, q') \in F\}\\
\end{array}
\right. \text{when $q \in Q$}
\\
& \Delta_\gamma((\Even, j_\gamma), q_0, 0) &=& I \\
& \Delta_\gamma((\Even, j_\gamma), q_0, 1) &=& \{\bot\} \\
& \Delta_\gamma(\bot, \bot, i) &=& \{\bot\} \\
& \Delta_\gamma(a, q, m) &=& \emptyset \qquad\qquad \text{otherwise} \\
\\
&  \text{and}
\qquad\qquad
  c_\gamma(q) &=& \left\{ \begin{array}{ll}
    c(q) & \text{if $q \in Q$}\\
      \min\{ k \in \bbN \mid \text{$k$ is even and $k \ge \max\limits_{q \in Q} c(q)$}\} & \text{if $q \in \{q_0, \bot\}$}\\
\end{array} \right.
    \end{array}
\]

\end{itemize}
\end{proof}

\begin{example}
  \label{ex:zetaParityDet}
  Let $\mathrm{ParityD}_k$ be the expression
$\gamma X_{k}. \ldots \zeta X_2. \nu X_1. \; \widehat{\widetilde{\sum\limits_{i = 1}^{k} X_i}}$
with  $\gamma \in \{ \zeta, \nu\}$ and $\gamma = \zeta$ if and only if $k$ is even.
The language $L_{{\mathrm{Parity}}_k}$ given by~\Cref{thm:zetaExprAut} is recognized by the following automaton
where all states have priority $0$. All transitions from $\bot$ states are omitted, under the assumption that
$\Delta(\bot, i, \bot) = \{\bot\}$ and $\Delta(b, i, \bot) = \emptyset$ for $b \neq \bot$.

\begin{center}
\begin{tikzpicture}[shorten >=1pt,node distance=1.25cm and 2.25cm,on grid,auto]
\tikzstyle{every node}=[font=\scriptsize]
\node[state,initial, initial where=left, initial text={}]  (K) {$k$};
\node  (KT) [right=of K] {$\ldots$};
\node[state]  (T) [right=of KT] {$2$};
\node[state]  (O) [right=of T] {$1$};
\node[state]  (ODD) [right=of O] {$\widehat{-}$};
\node[state]  (EVE) [right=of ODD] {$\widetilde{-}$};
\node  (bt)  [above=of T] {$\bot$};
\node[state]  (K1) [below= of K] {};
\node  (K11) [below= of K1] {};
\node  (K1b) [below= of K11] {$\bot$};
\node  (KT1) [right=of K1] {$\ldots$};
\node[state]  (T1) [right=of KT1] {};
\node[state]  (O1) [right=of T1] {};
\node  (KT1b) [below=of KT1] {$\ldots$};
\node  (T1b) [right=of KT1b] {$\bot$};
\node  (O1b) [right=of T1b] {$\bot$};
\node[state]  (K2) [below= of K] {};
\node  (KT2) [below=of KT1b] {$\ldots$};
\node[state]  (T2) [right=of KT2] {};
\node[state]  (O2) [right=of T2] {};

\path[->]    (K)   edge node [above] {$(\Even, k); 0$} (KT)
edge [bend left] node [left] {$(\Even, k) ; 1\phantom{aaa}$} (bt)
             (KT)  edge node [above] {$(\Even, 3) ; 0$} (T)
             (T)  edge node [above] {$(\Even, 2) ; 0$} (O)
                  edge node [left] {$(\Even, 2) ; 1$} (bt)
             (O)  edge node [above] {$(\Even, 1) ; 0$} (ODD)
                  edge [bend right] node [right] {$~~~(\Even, 1) ; 1$} (bt)
             (ODD)  edge node [above] {$(\Odd, 0) ; 0$} (EVE)
                  edge [loop above] node [above] {$(\Odd, 0) ; 1$} ()
             (EVE)  edge [bend left] node [below] {$(\Even, 1) ; 0$} (O1)
                    edge [bend left] node [right] {$(\Even, 1) ; 0$} (O2)
                    edge [loop above] node [above] {$(\Even, 1) ; 1$} ()
             (O1)   edge node [right] {$(\Even, 0) ; 1$} (O1b)
                    edge node [right] {$(\Even, 0); 0$} (O)
             (O2)   edge node [right] {$(\Even, 0) ; 1$} (O1b)
                    edge node [below] {$(\Even, 0); 0$} (T2)
                    edge [bend left] node [right] {$(\Even, 0); 0$} (T1)
             (T1)   edge node [left] {$(\Even, 0) ; 1$} (T1b)
                    edge node [left] {$(\Even, 0); 0$} (T)
             (T2)   edge node [left] {$(\Even, 0) ; 1$} (T1b)
                    edge node [below] {$(\Even, 0); 0$} (KT2)
             (KT2)
                    edge [bend left] node [right] {$(\Even, 0); 0$} (K1)
             (K1)   edge node [right] {$(\Even, 0); 0$} (K)
                    edge node [left] {$(\Even, 0); 1$} (K1b)
;
\path[-]     (T2)
                    edge [bend left] node [left] {} (KT1b)
                    ;
\end{tikzpicture}
\end{center}
A similar analysis as the previous example can show us that $L_{\mathrm{ParityD}_k}$
is isomorphic to the space of labellings $\ell : (\bbN^2)^* \to \{1, \ldots, k\}$.
Winning $\Even$ strategies in a game tree corresponding to $\ell$ are
in one-to-one correspondence with strategies in a game where players alternate
picking natural numbers forever and where the second player wins in a play $p : \bbN \to \bbN^2$
if and only if the $\limsup\limits_{n \to +\infty} \ell(p[n])$ is even.
\end{example}

\subsection{The other Weihrauch problems from~\Cref{fig:pblmsAsZetaMKA}}

Closed choice on $\Baire$ is handled
by countably-branching tree encodings.

\begin{proposition}
\label{prop:zetaCBaire}
$\C_{\Baire}$ is the problem taking as inputs trees $t \subseteq \bbN^*$
which have at least one infinite path and outputs one such path.
\end{proposition}
\begin{proposition}
$\C_{\Baire}$ is Weihrauch-equivalent to the answerable part of $\ffunctorify{\zeta X. \; \widetilde{X + 1}}$.
\end{proposition}
\begin{proof}
Questions to
$\ffunctorify{\zeta X. \; \widetilde{X + 1}}$ can be regarded as trees
$t \subseteq \bbN^*$ for every $n \in \bbN$. Accordingly, answers are infinite paths
through $t$ (here having $1$ in the expression (instead of, say, $\unit$) means that finite paths can't
yield valid answers).
We thus have a trivial Weihrauch equivalence reduction $\mkAnswerable\left(\ffunctorify{\zeta X. \; \widetilde{X + 1}}\right) \equiv C_{\Baire}$.
\end{proof}

Going above closed sets, we can also express instances of $\mathbf{\Pi}^0_2$-choice.

\begin{definition}
The infinite pigeonhole principle with $k$ colors $\RT^1_k$ is the problem
whose questions are sequences $e \in k^\bbN$, which have for answers colors
$i < k$ such that $\exists^\infty n. \; e_n = i$.
\end{definition}
\begin{proposition}
\label{prop:RT}
$\RT^1_k$ is Weihrauch-equivalent to the answerable part of
\[\ffunctorify{\underbrace{\left(\zeta X. \nu Y. \; Y + X\right)
\times \ldots \times \left(\zeta X. \nu Y. \; Y + X\right)}_{\text{$k$ times}}}\]
\end{proposition}
\begin{proof}
As discussed in~\Cref{ex:zetaRT}, the problem $\ffunctorify{\zeta X. \nu Y. \; Y + X}$
takes as input sequences $p \in 2^\bbN$, and there is an output (which is necessarily trivial)
if and only if $p$ has infinitely many $1$s. The
problem 
$\mkAnswerable\left(\ffunctorify{\left(\zeta X. \nu Y. \; Y + X\right)
\times \ldots \times \left(\zeta X. \nu Y. \; Y + X\right)}\right)$
therefore
takes as input $k$-many such sequences $(p_i)_{i < k}$ with the guarantee that
there is some $i < k$
such that $p_i$ has infinitely many ones, and it outputs such an $i$.
Much like in the proof of~\Cref{prop:cchoicefin}, we have easy reductions whose
backwards parts are trivial. The forward part of $\RT^1_k \le 
\mkAnswerable\left(\ffunctorify{\left(\zeta X. \nu Y. \; Y + X\right)
\times \ldots \times \left(\zeta X. \nu Y. \; Y + X\right)}\right)$
maps $e : \bbN \to k$ to $((\delta_i^{e_n})_{n \in \bbN})_{i < k}$ (using
Kronecker's $\delta$ notation), and the other way around, we turn $(p_i)_{i < k}$
into $e : \bbN \to k$ defined by
$e_n = \min \{ i < k \mid (p_m)_i = 1 \wedge m = \min \{ m' \mid \exists i < k. \; (p_m')_i = 1\}\}$.
\end{proof}

This leads to one of two natural formulations of the degree of K\"onig's lemma (\Cref{def:kl}) in this framework.

\begin{proposition}
$\KL$ is Weihrauch-equivalent to the answerable parts of
\[
\widehat{\ffunctorify{{\left(\zeta X. \nu Y. \; Y + X\right)
\times \left(\zeta X. \nu Y. \; Y + X\right)}}}
\qquad \text{and} \qquad 
\ffunctorify{\zeta X. \; (\nu Y. \; X + Y) \times (\nu Y. \; X + Y)}
\]
\end{proposition}
\begin{proof}
The first equivalence follows from $\KL \equiv \widehat{\RT^1_2}$~\cite[Corollary 8.12]{survey-brattka-gherardi-pauly},
$\mkAnswerable(\widehat{P}) \cong \widehat{\mkAnswerable(P)}$ and~\Cref{prop:RT}.
Assuming $P$ is a container,
$\functorify{\nu Y. \; X + Y}(P)$
can be understood as follows:
\begin{itemize}
\item a $\functorify{\nu Y. \; X + Y}(P)$-question is either a question to $P$,
or a coinductively delayed question to $\functorify{\nu Y. \; X + Y}(P)$.
Concretely, if we have $\shape(P) \subseteq \Cantor$, questions are $\omega$-words $w \in (\{\bot\} \uplus 2)^\bbN$
such that $w_j = \bot \Rightarrow w_i = \bot$ whenever $i < j$ and such that,
if $w = \bot^k u$ for some $u \in \Cantor$, then $u \in \shape(P)$.
\item An answer to such a question is either a (prompt!) answer to the corresponding $P$-question.
If there is no $P$-question asked, then there is no forthcoming $P$-answer.
When $I \subseteq \Cantor$, this can be formalized as
\[\left(\functorify{\nu Y. \; X + Y}(P)\right)^{-1}(w)  \cong 
\{ (k, a) \in \bbN \times P_\mathrm{dir} \mid \exists u \in I. \; w = \bot^k u \wedge a \in P^{-1}(u)\}\]
\end{itemize}
With this understanding, it follows that
$
\ffunctorify{\zeta X. \; (\nu Y. \; X + Y) \times (\nu Y. \; X + Y)}
$
is the problem whose questions are binary trees $t \subseteq 2^*$
given by enumerations of their paths in prefix-order, which can be answered by
infinite paths through $t$. The answerable part of this problem is thus
literally $\WKL'$ (modulo ensuring that paths are enumerated in prefix-order in
$\WKL'$, which can easily be done computably) and we can conclude since 
$\KL \equiv \WKL'$ (see e.g.~\cite[Theorem 8.10]{survey-brattka-gherardi-pauly}).
\end{proof}

For \Cref{prop:lpo}, we conspicuously had an expression of the shape $E \times E^\bot$, where
$E \mapsto E^\bot$ is a syntactic map that flips the polarity of answers
as follows (with $X = X^\bot$ for variables):
\[
\small
\begin{array}{r@{~~} c@{~~}l !\qquad r@{~~}c@{~~}l !\qquad r@{~~}c@{~~}l}
(E_1 + E_2)^\bot &=& E_1^\bot + E_2^\bot
&
(E_1 \tensor E_2)^\bot &=& E_1^\bot \times E_2^\bot
&
(E_1 \times E_2)^\bot &=& E_1^\bot \tensor E_2^\bot
\\
(\mu X. \; E)^\bot &=& \mu X. \; E^\bot
&
(\zeta X. \; E)^\bot &=& \nu X. \; E^\bot
&
(\nu X. \; E)^\bot &=& \zeta X. \; E^\bot
\end{array}
\]
The following proposition admits a completely analogous proof to that of~\Cref{prop:lpo}, which we skip.
\begin{proposition}
\label{prop:lpoj}
$\LPO'$ is the problem where questions are sequences $p \in \Cantor$ which
are answered by the bit $1$ if $\exists^\infty n \in \bbN. \; p_n = 1$ or by the
bit $0$ otherwise. It is Weihrauch-equivalent to the answerable part of
$\ffunctorify{
\left(\nu X. \; \zeta Y. \; X + Y\right) \times \left(\zeta X. \; \nu Y. \; X + Y\right)}$.
\end{proposition}

The pattern also emerges in a characterization of 
$\mathbf{\Pi}^1_1$-comprehension by a $\zeta$-expression.

\begin{definition}
\label{def:sTCBaire}
$\WF$ is the problem whose questions are trees $t \subseteq \bbN^*$ which
are answered by a bit telling whether $t$ is well-founded or not.
The $\bfP^1_1$-comprehension problem ($\PiooCA$) is defined as the infinite parallelization $\widehat{\WF}$ of $\WF$.

$\sTCBaire$, the strong total continuation of $\C_\Baire$~\cite[Definition 4.17]{MVRamsey},
is the
problem which takes as input a tree $t \subseteq \bbN^*$
and outputs a pair $(b, x) \in 2 \times \Baire$ such that $b = 0$ if and only if
$t$ is well-founded and $x$ is an infinite path through $t$ if $b = 1$.
\end{definition}

We have $\sTCBaire < \PiooCA$. For the reduction, given an input $t$ to $\sTCBaire$,
ask in parallel for every subtree of $t$ if it is well-founded; if it is ill-founded,
then a path can be reconstructed greedily by only going into ill-founded subtrees.
Strictness of the inequality can be inferred from the discussion in~\cite[Proposition 4.19]{MVRamsey}\footnote{More explicitly,
if we had $\PiooCA \equiv \sTCBaire$, then $\sTCBaire$ would be closed under infinite
parallelization giving $\widehat{\TCBaire} \le \widehat{\sTCBaire} \equiv \sTCBaire$,
contradicting~\cite[Proposition 4.19]{MVRamsey}.}. Conversely, it is clear that
$\PiooCA \le \widehat{\sTCBaire}$; so using similar arguments as before, it
suffices to characterize the $\sTCBaire$ as the answerable part of a $\zeta$-expression
to get such a characterization of $\PiooCA$.
We conjecture that it is impossible to
have such a characterization of $\WF$.

\begin{proposition}
$\sTCBaire$ is Weihrauch-equivalent to the answerable part of
\[\ffunctorify{\left(\nu X. \widehat{X + \unit}\right) \times  \left(\zeta X. \widetilde{X + 1} \right)}\]
\end{proposition}
\begin{proof}
As discussed in the proof of~\Cref{prop:zetaCBaire},
$\ffunctorify{\zeta X. \widetilde{X + 1}}$ is the problem that takes
a tree $t \subseteq \bbN^*$ to an infinite path through $t$. The dual problem
$\ffunctorify{\nu X. \widehat{X + \unit}}$ takes the same type of input, but
can only have trivial outputs; and a question has such an answer if and only if
the input tree is well-founded.
As such, much like in the proof of~\Cref{prop:lpo}, we have an easy reduction
\[\sTCBaire \qquad \le \qquad \mkAnswerable\left(\ffunctorify{\left(\nu X. \widehat{X + \unit}\right) \times  \left(\zeta X. \widetilde{X + 1}
\right)}\right)\]
whose forward part is simply the diagonal map $t \mapsto (t, t)$. The other way
around, the forward map is the first projection $(t, t') \mapsto t$: so if
$t$ is ill-founded, we get a path as requested, and otherwise, we answer that $t'$
is well-founded\footnote{Note that in the proof of~\Cref{prop:lpo}, we could have
a reduction by picking either the first or the second projection as a forward part,
but here we \emph{need} to pick the first projection, because otherwise we would be
left with $\C_{\Baire}$ to solve.}.
\end{proof}

\mysubparagraph{Determinacy statements}
By~\Cref{thm:zetaExprAut}, $\zeta$-expressions boil down to problems that ask for
winning strategies in certain sets of Gale-Stewart games, i.e., games where two players play natural numbers in turn
for $\omega$ rounds and where the winning condition is some subset $A \subseteq \Baire$.
Let us henceforth call $\cG(A)$ the Gale-Stewart game with winning condition $A$
for the second player. An important result in descriptive set theory is Martin's theorem:
all such games are determined when $A$ is Borel, that is, there exists a winning strategy
for either the first player ($\Odd$) or one for the second player ($\Even$)~\cite[\S 20.C]{kechris}. Those
strategies are typically highly non-computable, and accordingly, determinacy
statements 
have high proof-theoretic
strength,
even for subclasses of $\mathbf{\Delta}^0_3$. Here we focus on Weihrauch problems
stated as follows: given
some (code for some) $A$ with the promise that the second player wins in $\cG(A)$, find a
winning strategy\footnote{Variants where one does away with the promise that
a player wins, and one possibly only asks for who is the winner can also be considered~\cite{cantorDetBCO,kmp20}.}.

To define such problems, we need to fix coding schemes for sets of winning conditions.
For those, we shall only consider those pointclasses
$\mathbf{\Gamma}$ that are finite boolean combinations of the classes
$\mathbf{\Delta}^0_1$, $\mathbf{\Sigma}^0_1$, $\mathbf{\Delta}^0_2$ and $\mathbf{\Sigma}^0_2$.
Let us code them as follows, with the convention that
$(\mathbf{\Sigma}^0_i)_{k + 1}$ is the pointclass consisting of sets $A \setminus B$
with $A \in \mathbf{\Sigma}^0_i$ and $B \in (\mathbf{\Sigma}^0_i)_k$ for $i = 1, 2$,
$k \in \bbN$ and
$(\mathbf{\Sigma}^0_i)_0 = \mathbf{\Delta}^0_i$.
Given a pointclass $\mathbf{\Gamma}$, we write $\dualptcl{\bfG}$ for its
dual pointclass.

\begin{definition}
\label{def:pointclasses}
A $(\mathbf{\Sigma}^0_2)_\omega$-code for a set $A \subseteq \Baire$
is a map $s : \bbN^* \to \bbN$ such that $x \in A$ if and only if
$\limsup\limits_{n \to +\infty} s(x[n])$ is even.
For $i = 1, 2$, $(\mathbf{\Sigma}^0_i)_{k}$-code for $A \subseteq \Baire$
are then $(\mathbf{\Sigma}^0_i)_\omega$-codes $s$ for $A$ subject to further
restrictions: for $i = 1$, $s$ must be increasing along the prefix order, and
for $i \in \{1,2\}$, we have the following restrictions on the image of $s$,
elements of which we call \emph{priorities}:

\begin{center}
\begin{tabular}{|l||c|c|c|c|}
\hline
Pointclass
&
$(\bfS^0_1)_{k}$
&
$(\bfS^0_2)_{2k+1}$
&
$(\bfS^0_2)_{2k+2}$
\\
\hline 
Priorities
&
$\{1, \ldots, k + 1\}$
&
$\{0, \ldots, 2k + 1\}$
&
$\{1, \ldots, 2k + 3\}$
\\
\hline
\end{tabular}
\end{center}

$\mathbf{\Sigma}^0_i$ sets for $i = 1, 2$ are $(\mathbf{\Sigma}^0_i)_1$-sets
and coded accordingly. $\dualptcl{\mathbf{\Gamma}}$-sets are coded by $\bfG$-codes for their complements.
A code for a $\mathbf{\Delta}^0_i$ set $A$ is a pair of a codes for $A$
as a $\mathbf{\Sigma}^0_i$ and a $\mathbf{\Pi}^0_i$ set.
\end{definition}

\begin{definition}
Given a coding scheme for a pointclass $\bfG$,
questions to $\WS{\bfG}$ are $\bfG$-codes for sets $A \subseteq \Baire$
such that $\Even$ has a winning strategy in $\cG(A)$, and corresponding answers
are such winning strategies.
$\WSCantor{\bfG}$ is the same except we consider the game $\cG_2(A)$ where
both players are only allowed moves in $2$ (so only $A \cap \Cantor$ matters
when it comes to winning).
\end{definition}

\begin{remark}
For the pointclasses $\bfG$ below $\bfD^0_2$,
there are a number of known Weihrauch equivalences between
$\WS{\bfG}$ and some of the problems we have seen:
$\WS{\bfD^0_1} \equiv \WS{\bfS^0_1} \equiv \UC_{\Baire}$,
$\WS{\bfP^0_1} \equiv \C_{\Baire}$~\cite{kmp20} and
$\WSCantor{{(\bfS^0_1)_{k+1}^c}}$ is equivalent to the $k$th jump of
$\WKL$ for all $k \ge 0$~\cite{cantorDetBCO}. The characterization
we offer of $\WS{(\bfS^0_1)^c_k}$ in terms of iterations of $\PiooCA$
(Lemmas~\ref{lem:WSopenDBelow} and~\ref{lem:WSopenDAbove}) is similar and arguably folklore; the arguments are
similar to those at play in the analysis of $\bfD^0_2$-determinacy in
reverse mathematics~\cite{tanaka90} and descriptive set
theory~\cite{burgess1983classical1}.
\end{remark}

We already have explained characterizations of $\WS{\bfD^0_1}$ and
$\WS{(\bfS^0_2)_k}$.

\begin{proposition}
$\WS{(\bfS^0_2)_{k}}$ is Weihrauch-equivalent to the answerable part of
\[
\ffunctorify{
\nu X_{k+1}. \zeta X_k. \ldots \gamma X_1. \; \widehat{\widetilde{\sum\limits_{i = 1}^{k+1} X_i}}
}
\qquad \text{with $\gamma \in \{ \zeta, \nu\}$, $\gamma = \zeta \Longleftrightarrow \text{$k$ is odd}$}\]
\end{proposition}
\begin{proof}[Proof idea]
See~\Cref{ex:zetaParityDet}.
\end{proof}

\begin{remark}
$\WSCantor{(\bfS^0_2)_k}$ can be coded in the same way as 
$\WS{(\bfS^0_2)_k}$ by simply replacing the countable operators $\widehat{-}$
and $\widetilde{-}$ with their respective binary variants $- \tensor -$ and
$- \times -$. Similar observations hold for all the codings we give in the sequel.
\end{remark}

\begin{proposition}
$\WS{\bfD^0_1}$ is Weihrauch-equivalent to 
$\mkAnswerable\left(\ffunctorify{\mu X. ~ \widehat{\widetilde{X}} + \unit + 1}\right)$.
\end{proposition}
\begin{proof}[Proof idea]
For any $\bfD^0_1$-code $(s_0, s_1)$ and $p \in \Baire$, we have that there
exists $n \in \bbN$ such that $s_i(p[n]) = 0$ for $i = 0$ or $1$ (and not both).
This means that the outcome of the game is determined after a finite
amount of time. We can thus compute a well-founded game tree with the leaves
labelled by the winner, which is a valid input to
$\mkAnswerable\left(\ffunctorify{\mu X. ~ \widehat{\widetilde{X}} + \unit + 1}\right)$
(see~\Cref{ex:zetaClopenDet}). Conversely, the winning position for player $i$
determines a code $s_i$ for the open set containing the winning plays for $i$.
It is not hard to come up with a mapping of strategies to complete the reductions
both ways.
\end{proof}

A similar analysis is sufficient to also characterize $\WS{\bfD^0_2}$.

\begin{proposition}
\label{prop:deltaztdet}
$\WS{\bfD^0_2}$ is Weihrauch-equivalent to the answerable part of
\[\ffunctorify{\mu Z. \; \nu X. ~~ \widehat{\widetilde{X}} \; + \; \left(\zeta Y.\; \widehat{\widetilde{Y}} + Z\right)}\]
\end{proposition}
\begin{proof}[Proof idea]
A $\bfD^0_2$ code for $A$ can easily be computably turned into an increasing
$(\bfS^0_2)_\omega$ code $s$ for $A$ such that $\lim\limits_{n \to +\infty} s(x[n]) < + \infty$
for any $x \in \Baire$.
Furthermore, we can also assume
without loss of generality that those $s : \bbN^* \to \bbN$ start off at $1$
and only increase by at most $1$ on paths of even length; formally, for every
$w \in \bbN^*$ and $n \in \bbN$,
\[s(\varepsilon) = 1 \qquad\qquad s(wn) \le s(w) + 1 \qquad \text{and} \qquad \text{$|w|$ odd} ~~\Longrightarrow~~ s(wn) = s(w)\]
Conversely, any such $s$ can be computably turned into a $\bfD^0_2$-code by noting
\[U_s \quad \eqdef \quad \{x \in \Baire \mid \exists^\infty n \in \bbN. \; \text{$s(x[n])$ is even}\} \quad =
\quad \{x \in \Baire \mid \forall^\infty n \in \bbN. \; \text{$s(x[n])$ is even}\}\] 
Then reductions both ways are easy, as $\ffunctorify{\mu Z. \; \zeta X. ~~ \widehat{\widetilde{X}} \; + \; \left(\nu Y.\; \widehat{\widetilde{Y}} +
Z\right)}$ is isomorphic to the problem whose questions are such $s$ which
are answered by winning strategies for the second player in $\cG(U_s)$.
\end{proof}

$\WS{(\bfS^0_1)_{2k+1}}$
can be expressed as the answerable part of
the $k$-fold composition
of $\ffunctorify{\nu X. ~~ \widehat{\widetilde{X}} \; + \; \left(\zeta Y.\; \widehat{\widetilde{Y}} + Z\right)}$
applied to $0$; this can be done in a single $\zeta$-expression using syntactic substitutions.
That this is correct can be proven in the same way as~\Cref{prop:deltaztdet}.
Let us close this subsection by noting that the
$\WS{(\bfS^0_1)_{2k}}$ do form a strict hierarchy below
$\WS{\bfD^0_2}$ as they are linked to finite iterations of $\PiooCA$.
In what follows, call $(\PiooCA)^{\star k}$ the $k$-fold sequential composition
$\PiooCA \starW \ldots \starW \PiooCA$ in the Weihrauch degrees. Here $\starW$ means composition
in the usual sense of Weihrauch complexity rather than the functorial operation
$\star$ in $\ReprSp$ (see~\cite[\S 4.4]{PricePradic25} for a discussion of the differences).
This is important as the dovetailing used in the proof of~\Cref{lem:WSopenDAbove} is not extensional.

\begin{lemma}
  \label{lem:WSopenDBelow}
$\WS{(\bfS^0_1)_{k + 1}} \le
\UC_\Baire \starW (\PiooCA)^{\star k}$ and
$\WS{(\bfS^0_1)^c_{k + 1}} \le
\C_\Baire \starW (\PiooCA)^{\star k}$.
\end{lemma}
\begin{proof}[Proof]
  In what follows, let us bear in mind that $\UC_{\Baire} \equiv \WS{\bfS^0_1} < C_\Baire \equiv \WS{\bfP^0_1}$~\cite[Proposition 6.19 \& Theorem 6.21]{kmp20}.
For the reduction
$\WS{(\bfS^0_1)_{k + 1}} \le
\UC_\Baire \starW  (\PiooCA)^{\star k}$,
one can proceed in two steps: first, use the iterations of $\PiooCA$ to compute
the winning region for $\Even$ for priorities greater than $2$, and then
$\widehat{\WS{\bfS^0_1}}$ to compute the overall winning strategy (which is enough
to conclude as $\widehat{\UC_\Baire} \equiv {\UC_\Baire}$~\cite[Corollary 27]{brattka2025loops}).

For the first step, we can first remark that computing the winning regions in a
$\bfS^0_1$ game is equivalent to $\PiooCA$. Then given a $(\bfS^0_1)_{k + 1}$-code
$s : \bbN^* \to \{1, \ldots, k+2\}$, we can compute the winning regions for
$\Even$ for positions of priority $\ge 2$ by recursing over $k$:
\begin{itemize}
\item for $k = 0$, the winning region is simply $s^{-1}(2)$,
\item otherwise, by the induction hypothesis, the $\Even$-winning region $W \subseteq \bbN^*$ for
priorities greater than $3$ is computed by $k - 1$ sequential questions to $\PiooCA$.
Then consider the following $\bfS^0_1$-code $s' : \bbN^* \to \{1, 2\}$ defining the open set $U_{s'}$. 
\[
s'(w) \quad = \quad \left\{
\begin{array}{ll}
  1 & \text{if $s(w) \le 2$ or $w \in W$}\\
  2 & \text{otherwise}
\end{array}\right.
\]
Then the intersection of the winning region for $\Odd$ in $\cG(U_{s'})$
with $s^{-1}(\{2, \ldots k + 2\})$ is the winning region for $\Even$, which can
be computed by a further question to $\PiooCA$.
\end{itemize}
To compute a winning strategy from the winning region $W$ for $\Even$,
use $\widehat{\WS{\bfS^0_1}}$ to compute partial strategies $S_w$ from the
positions $w \in W$ with odd priorities in the winning region that eventually always lead to positions
with even priorities in the winning region (and then are undefined). Also do so for the initial position
$\varepsilon$ of the game. We can then put together a strategy
$S$ for $\Even$ contained in $\bigcup\limits_{w \in W \cup \{\varepsilon\}} S_w \cup (W \cap s^{-1}(\{2, 4, \ldots\}))$.

The reduction
$\WS{(\bfS^0_1)^c_{k + 1}} \le
\C_\Baire \starW  (\PiooCA)^{\star k}$ is similar, save for the treatment of the
last step, where, instead of using $\UC_\Baire$ to find the initial part of the
strategy, we use $\C_\Baire$ as $\C_\Baire \equiv \WS{\bfP^0_1}$. This yields
the desired result as $\UC_\Baire \le \C_\Baire$ and $\C_\Baire$ is parallelizable.
\end{proof}

\begin{lemma}
  \label{lem:WSopenDAbove}
$C_\Baire \starW (\PiooCA)^{\star k} \le \WS{(\bfS^0_1)^c_{k+1}}$.
\end{lemma}
\begin{proof}[Proof idea]
Let us write $T \subseteq \Baire$ for the set of
trees with $\bbN$-ary branching and, assuming $P$ is a \emph{deterministic} problem\footnote{I.e.,
that every question has exactly one answer.}, let us write $[P]$ for a section of $P$.
In particular we have that $[\mathsf{WF}] : T \to 2$ and $[\PiooCA] : T^\bbN \to \Cantor$.
A question to $(\PiooCA)^{\star k}$
is a code for a finite sequence of partial maps
$f_0, f_1, \ldots, f_{k}, g$ with $f_i : (\Cantor)^i \to T^\bbN$ and $g : (\Cantor)^{k + 1} \to T$. The answer to
such a question is a pair $(a, p)$ where $a \in (\Cantor)^{k}$ is the value defined by the recursion
$a_{i} = [\PiooCA](f_i((a_j)_{j < i}))$ for $i \le k$ and $p \in \Baire$ is an infinite
path through $g(a)$ (and we have the guarantee that such an infinite path exists).
Without loss of generality, we may assume that the $f_i$s and $g$ are actually
total maps\footnote{The idea is that there is a computable map $\rho : \bbN \to \bbN$ such that we can
computably turn an enumeration of a tree $t \in T$ into a tree $t' \in T$
such that $\rho^* : \bbN^* \to \bbN^*$ restricts to a surjective map $t' \to t$.}.

Now let us informally describe the $(\bfS^0_1)^c_{k + 1}$ game $G$ that our
forward reduction computes from the $f_i$s, assuming without loss of generality that $k \ge 1$.
The positions of $G$ will thus be labelled with priorities in $\{0, \ldots, k + 1\}$.
The game overall proceeds in $k+1$ phases corresponding to the priorities $\{0, \ldots, k\}$;
$k+1$ is kept around so that we may force a loss for either of the player in phase $k$.
For $i \in \bbN$, let $P_i \eqdef \Even$ when $i$ is even and $\Odd$ otherwise
and $\overline{P_i}$ the opponent of $P_i$.

In all phases, the game consists essentially of two processes $A$ and $C$
that are dovetailed. While at each round some progress will be made on $A$,
whether options for progress for $C$ are available depend on what progress is made
in $A$. Let us call $A_i$ and $C_i$ the phase $i$ part of those two processes.
In $A_i$, the player $P_i$ is asked to guess the value of
$(a_j)_{j \le k - i}$, bit-by-bit and $\overline{P_i}$ is asked whether they
agree or not. Should $\overline{P_i}$ disagree on a bit of $a_j$, the game goes into phase $j+1$
and $C_{{j+1}}$ is used to guarantee that $P_{j+1}$'s guess is correct,
assuming $\overline{P_{j+1}}$ does not challenge further their guesses in $A_{j+1}$.
We also force all players' guesses to be consistent throughout a play.
Let us now describe the $C_i$s in more details, which may terminate
the game by declaring a winner.
\begin{itemize}
\item For $C_{j + 1}$ with $j < k - 1$, what happened before in $C$ is ignored,
and what happens next depends on the nature of the disagreement that happened
leading to phase $j + 1$.
So assume that the disagreement is on the bit $(a_j)_n$, i.e., the decision on whether
$T = f_j((a_m)_{m < j})_n$ is well-founded or not.
For any potential guess $b \in (\Cantor)^{k-j-1}$ for $(a_m)_{m < k - j - 1}$, we can compute 
a (code for a) closed set $V_{b} \subseteq \Baire$ that corresponds to the
infinite paths in $T$ if $b$ is correct. We can extend that to codes
$V_{\tilde{b}}$ where $\tilde{b} \in (2^*)^{k-j-1}$ are finite prefixes of a
guess $b$, so that $V_b$ is the intersection of all the $V_{\tilde{b}}$ with
$\tilde{b}$ approximating $b$.
Then, the player who claimed that $T$ is ill-founded is asked in $C_{j+1}$
to also enumerate some $p \in \Baire$ that is supposed to be an infinite path in
$T$. If at any point it is noticed that $p \notin V_{\tilde{b}}$ where
$\tilde{b}$ is part of their guess for $(a_m)_{m < j}$, then $P_{j+1}$ loses the
whole game. Note that as long as both players are in phase $j + 1$, they agree
on their guesses for  $(a_m)_{m < k - j - 1}$.
\item In $C_0$, $\Even$ attempts to guess an infinite path $p$ of $g(a)$. Again
we can compute codes for closed sets $W_{\tilde{b}}$ for finite prefixes $\tilde{b}$
of potential guesses such that the intersection $W_a$ of all prefixes of $a$
contains exactly the infinite paths of $g(a)$. If at any point we notice that
$p \not\in W_{\tilde{b}}$ for the current guess $\tilde{b}$, $\Even$ loses the
game.
\end{itemize}
Now it is rather clear that $\Even$ has a winning strategy in $G$ which
consists in correctly guessing $a$ in $A$ and
possibly rising correctly to challenges in the $C_j$s.
Conversely, any winning strategy for $\Even$ needs to be guessing $a$ correctly;
otherwise, it could be defeated by $\Odd$ by having them challenge a wrong bit
$(\widetilde{a}_j)_n$ and then playing optimally.
\end{proof}
\begin{remark}
The previous two lemmas can be adapted to show that $\WS{\bfD^0_2} \equiv (\PiooCA)^{(\dagger)}$
where $(\dagger)$ is the countable ordinal sequential iteration of~\cite{paulycountableordinals}.
\end{remark}

\begin{corollary}
  \label{cor:WSdiffStrict}
For all $k > 0$, we have $\WS{(\bfS^0_1)_k^c} < \WS{(\bfS^0_1)_{k+1}^c} < \WS{\bfD^0_2}$
and $\WS{(\bfS^0_1)_{k}} < \WS{(\bfS^0_1)_{k + 2}} < \WS{\bfD^0_2}$.
\end{corollary}
\begin{proof}
This follows from the strictness of the inequality
\[
C_\Baire \star (\PiooCA)^{\star k} <
(\PiooCA)^{\star(k + 1)}\]
This inequality is strict because of the strictness of the Turing hyperjump (see~\cite[\S 16.6]{rogers1987}
for definitions):
$(\PiooCA)^{\star(k + 1)}$ can compute the $(k+1)$th Turing hyperjump from a computable
input, and
$C_\Baire \star (\PiooCA)^{\star k}$ cannot, as it can only compute elements of
the hyperdegree of the $k$th Turing hyperjump from a computable input.
\end{proof}

We conjecture that $\WS{(\bfS^0_1)_k} <
\WS{(\bfS^0_1)_{k+1}}$ and that $\WS{(\bfS^0_2)_k}$ gives a strictly
increasing chain of degrees\footnote{That the analogous game quantifiers are known to
form a strict hierarchy from the topological point of view~\cite{burgess1983classical2}
motivates the conjecture, although we should rather consider the complexity
of building strategies rather than just determining winning regions.}.

\end{document}